\documentclass[12pt]{scrartcl} 
\usepackage{listings, xcolor}
\usepackage{amsmath}
\usepackage{amssymb}
\usepackage{graphicx}
\usepackage{caption}
\usepackage{subcaption}
\usepackage[T1]{fontenc}
\usepackage[latin1]{inputenc}
\usepackage[english]{babel}
\usepackage[latin1]{inputenc}
\usepackage{lmodern}
\usepackage{bbm}
\usepackage{float} 
\usepackage{epstopdf}
\usepackage{array}
\usepackage{ulem}

\usepackage{pifont}
\usepackage{threeparttable}
\usepackage{setspace}
\usepackage[round]{natbib}
\usepackage{babelbib}
\usepackage{url}
\usepackage{rotating}
\usepackage{dcolumn}
\newcolumntype{d}[1]{D{.}{.}{#1}}
\usepackage{graphicx}
\usepackage{subcaption}

\usepackage{booktabs}
\usepackage{tabularx}
\usepackage{indentfirst}
\usepackage{multirow}
\usepackage{framed}
\usepackage{enumitem}
\usepackage{lscape}
\usepackage{longtable}
\usepackage{tabulary}  
\usepackage{rotating}

\usepackage{chngcntr}
\counterwithin{figure}{section}
\counterwithin{table}{section}
\counterwithin{equation}{section}
\usepackage{setspace}

\usepackage{colortbl}

\usepackage[left=2.54cm,right=2.54cm,top=2.54cm,bottom=2.54cm,marginparwidth=1.75cm]{geometry}
\usepackage{scrlayer-scrpage}
\usepackage{ragged2e}
\usepackage[overload]{empheq}
\ofoot[\pagemark]{\pagemark}

\usepackage{tikz}
\usetikzlibrary{shapes,decorations,arrows,calc,arrows.meta,fit,positioning}
\tikzset{
	-Latex,auto,node distance =1 cm and 1 cm,semithick,
	state/.style ={ellipse, draw, minimum width = 0.7 cm},
	point/.style = {circle, draw, inner sep=0.04cm,fill,node contents={}},
	bidirected/.style={Latex-Latex,dashed},
	el/.style = {inner sep=2pt, align=left, sloped}
}
\tikzset{
	vertex/.style = {
		circle,
		fill            = black,
		outer sep = 2pt,
		inner sep = 1pt,
	}
}

\tikzstyle{line} = [draw, -latex']
\usetikzlibrary{arrows,decorations.markings,patterns,calc}
\usetikzlibrary{shadows}
\tikzset{shadow scale=1, shadow xshift=-.5ex, shadow yshift=-.5ex,
	opacity=.5, fill=black!50, every shadow}

\color{black}
\usepackage[titletoc,title]{appendix}

\usepackage[hidelinks,breaklinks=true]{hyperref}

\newcommand{\dt}{\frac{\partial}{\partial t}}

\newtheorem{assumption}{Assumption}[section]
\newtheorem{lemma}{Lemma}[section]
\newtheorem{theorem}{Theorem}[section]

\newtheorem{proposition}{Proposition}[section]

\newenvironment{proof}[1][Proof]{\begin{trivlist}
\item[\hskip \labelsep {\bfseries #1.}]}{\end{trivlist}}

\newcommand\IF{\mathbbm{IF}}

\usepackage{algorithm}
\usepackage{algpseudocode}

\usepackage{float}
\floatstyle{ruled}
\restylefloat{algorithm}

\usepackage{booktabs}
\newcommand{\diffcell}[3]{%
  \makebox[3.2em][r]{$#1$}.$#2$#3%
}

\usepackage{dcolumn}
\newcolumntype{d}[1]{D{.}{.}{#1}}
\newcolumntype{Y}{>{\centering\arraybackslash}X}
\newcolumntype{Z}{>{\flushleft\arraybackslash}X}
\newcommand{\titleinfo}{Sharp Bounds and Inference in Sample Selection Models with Treatment Endogeneity}
\title{\titleinfo}
\def\authora{Yingying Dong}
\def\authorb{Phillip Heiler}

\def\emaila{\href{mailto:pheiler@econ.au.dk}{pheiler@econ.au.dk}}
\def\emailb{\href{mailto:yyd@uci.edu}{yyd@uci.edu}}

\date{ \today }  

\begin{document}
	\begin{titlepage}
		\title{\titleinfo \thanks{ \scriptsize We would like to thank Otavio Bartalotti, Giovanni Mellace, Vitor Possebom, Vira Semenova, and the participants at the IAAE 2025 and USC-UCLA Mini Workshop for helpful comments and discussion. Phillip Heiler would also like to thank the Danish National Research Foundation (DNRF186) and the Independent Research Fund Denmark for support (DFF 3099-00077B).  All remaining errors are ours.}}
		\author{
   \authora\thanks{{\scriptsize University of California Irvine, Department of Economics, Irvine, CA, 92697, USA.
 email: \emailb}} \quad 
  \authorb\thanks{{\scriptsize Aarhus University. Department of Economics and Business Economics, Aarhus Center for Econometrics (ACE), TrygFonden's Centre for Child Research, Universitetsbyen 51, 8000 Aarhus C, Denmark. email: \emaila}} \quad 
      }
	
		\date{\underline{This version}: \today}
		\maketitle
		\thispagestyle{empty}
		
		\begin{abstract} \singlespacing	\small
    This paper provides partial identification and inference for treatment effects in nonparametric sample selection models with endogenous treatment and (weak) sample selection monotonicity. Outcomes are observed only for a non-randomly selected subsample and treatment is endogenous because of noncompliance with assignment. The proposed bounds for intensive margin treatment effects among compliers are sharp and tighter than those of \cite{chen2015bounds}. For inference, we develop semiparametrically efficient orthogonal moments and a debiased machine learning procedure that permits valid root-$n$ inference under high-dimensional covariates and/or flexible functional forms. Simulation results indicate good finite sample performance. Applications to Job Corps and the Oregon Health Insurance Experiment show that the method can deliver substantially tighter effect bounds and confidence intervals than existing alternatives.
		\end{abstract}
	\end{titlepage}
	
\setcounter{page}{1}
	\newpage

\section{Introduction}

This paper provides sharp bounds and inference for the causal effect of a binary treatment in settings with nonrandom sample selection and endogenous treatment take-up -- two prominent features of many empirical applications. Our paper can be viewed from two complementary perspectives: From the perspective of the sample selection literature, we extend 
intensive margin bounds from reduced-form effects of treatment assignment to complier treatment effects in the presence of treatment take-up endogeneity. From the perspective of the instrumental variable/local average treatment effect (IV/LATE) literature, we extend the standard LATE framework to allow for sample selection, or partial observability of the outcome. Thus, this paper targets the intensive margin analogue of LATE for always-selected compliers. The framework nests the widely used \cite{lee2009training} 
bounds 
as a special case. Identification is obtained under standard  IV/LATE assumptions together with weak sample selection monotonicity for treatment compliers. We also develop semiparametrically efficient orthogonal moments and debiased machine learning (DML) estimators that allow valid root-$n$ inference on the treatment effect and its sharp bounds under high-dimensional covariates and flexible functional forms.

Partial identification of causal effects under sample selection or partial observability has been extensively studied. 
Much of the existing work, including \cite{lee2009training}, derives bounds for the effect of a randomly assigned instrument without accounting for endogenous treatment take-up, and is therefore limited to intention-to-treat (ITT) effects when assignment and treatment receipt differ, see, e.g., \cite{horowitz2000nonparametric}, \cite{zhang2003estimation}, \cite{imai2008sharp}, \cite{huber2015sharp}, \cite{heiler2024heterogeneous}, \cite{heiler2024treatmentevaluationintensiveextensive}, \cite{sun2024partiallyidentifiedheterogeneoustreatment}, or \cite{lee2025leeboundscontinuoustreatment}.\footnote{\cite{horowitz2000nonparametric} propose nonparametric bounds for randomized experiments with missing covariate and outcome data. \cite{zhang2003estimation} develop bounds for the survivor average causal effect, targeting individuals who are always selected regardless of treatment status using principal stratification \citep{frangakis2002principal}. They provide both assumption-free bounds and bounds under sample selection monotonicity and stochastic dominance assumptions. Their monotonicity based bounds are now regularly referred to as ``Lee bounds'' or ``Zhang-Rubin-Lee bounds'' \citep{andersen2023guide}. \cite{imai2008sharp} proves the sharpness of the \cite{zhang2003estimation} bounds and extends them to quantile treatment effects. \cite{huber2015sharp} derive bounds for additional subpopulations, such as individuals selected only under treatment or only under control, and for the observed subpopulation. \cite{heiler2024heterogeneous} provides heterogeneous treatment effect bounds. \cite{heiler2024treatmentevaluationintensiveextensive} provide identification for principal strata and extensive and intensive margin bounds. \cite{sun2024partiallyidentifiedheterogeneoustreatment} consider bounds for type-specific potential outcome means, assuming exogenous treatment assignment and using an excluded variable that shifts selection but not potential outcomes. \cite{lee2025leeboundscontinuoustreatment} consider intensive margin bounds for continuous treatments.}

In practice, however, realized treatment take-up is often endogenous in both experimental and quasi-experimental settings as assignment, eligibility, or encouragement may shift participation without perfectly determining treatment. For instance, in our empirical applications a nontrivial fraction of individuals do not comply with the assignment. In the National Job Corps (JC) Study, 26.2\% of the individuals provided with access to training did not enroll in JC, while 4.4\% of individuals assigned to the control group eventually enrolled in JC after randomization \citep{schochet2008does}. In the Oregon Health Insurance Experiment (OHIE), only 30\% of those eligible to apply for Medicaid successfully enrolled, and a small share of controls also enrolled \citep{finkelstein2012oregon}. Noncompliance of these magnitudes can substantially affect treatment effect bounds. 

Without sample selection, ITT and complier treatment effect differ only by a scaling factor, the share of compliers \citep{angrist1996identification}. We show that this simple relationship breaks down under sample selection: Sharp treatment effect bounds cannot be obtained by merely scaling sharp bounds of ITT effects via primitive probabilities. Our analysis therefore fills an important gap by delivering identification and inference for sharp treatment effect bounds in the more realistic case where both endogenous sample selection and noncompliance (implying treatment endogeneity) are present.

Other papers have leveraged IV/LATE-type assumptions for partial identification of various parameters in sample selection models: 
\cite{lechner2010partial} derive bounds for mean and quantile treatment effects for observable subpopulations, such as the ``treated and selected''. 
\cite{christelis2019partial} provide bounds on potential outcome distributions under selected samples and a monotone IV assumption \citep{manski1997monotone}. 

Within this literature, the closest contribution to ours is Chen and Flores (2015; CF). 
To our  knowledge, CF is the only paper studying bounds in a setting with both sample selection and noncompliance that also provides statistical inference. Unlike the CF bounds, our bounds are sharp. Related identification contributions include \cite{imai2007identification}\footnote{The derivation of bounds in \cite{imai2007identification} is based on a principal-strata approach that is complementary to ours. However, the corresponding bound expressions imported from \cite{imai2008sharp} appear to contain errors in the expressions for $Q_0$ and $Q_1$ in RESULT 1 on page 4. Hence, the resulting bounds cannot be reconciled with ours. Moreover,
\cite{imai2007identification} does not develop estimation or inference, and
sharpness is not formally established for the noncompliance-and-selection
parameter. Sharpness is only indirectly suggested by the arguments in
\cite{imai2008sharp}, which is formulated for ITT bounds under different
assumptions. By contrast, we provide corrected expressions for the bounds, prove
their sharpness, compare them with CF bounds, and develop a full estimation and
inference framework that extends to settings with both discrete and continuous
covariates and weak sample selection monotonicity.} and \cite{bartalotti2023identifying}.\footnote{\cite{bartalotti2023identifying}
study treatment-effect bounds under sample selection within a latent-index marginal treatment effect (MTE)
framework. Their Proposition 5 derives sharp bounds for always-observed LATEs
with multi-valued discrete instruments. In the binary-instrument, no-covariate
case, their latent-index interval $p_0<V\le p_1$ coincides with our complier
stratum as defined in Section \ref{sec:bounds}. Under the same strong sample selection monotonicity, their Proposition 5 bounds are algebraically equivalent
to our basic sharp bounds without covariates in Section \ref{sec:bounds}. They informally discuss plug-in estimation for MTE bounds without covariates, but do not develop accompanying inference guarantees for the discrete-instrument always-observed LATE bounds in Proposition 5.}
Relative to all three papers, we
contribute along multiple dimensions: First, we allow for the inclusion of covariates, thereby accommodating unconfounded rather than independently assigned instruments, which is often more realistic in quasi-experimental settings. 
Even when the instrument is independently assigned, incorporating covariates can improve efficiency, much as in standard average treatment effect (ATE) estimation. Second, we relax their individual-level strong sample selection monotonicity to weak, i.e.,~covariate-dependent, monotonicity. This is empirically relevant: \cite{semenova2025generalized} shows that strong sample selection monotonicity can be rejected in JC. We also find substantial violations in the OHIE. Third, we provide a simple covariate-profiling procedure for the targeted complier population at the intensive margin. This helps to assess both the external validity of the estimates and the substantive relevance of the subpopulation to which they apply, paralleling complier profiling approaches from the LATE literature \citep{abadie2003semiparametric,angrist2004treatmenteffectheterogeneity,singh2023doublerobustnessforcomplier}. Fourth, we propose a semiparametric DML procedure that 
delivers valid root-$n$ inference using generic nonparametric or machine learning methods for the nuisance components, allowing for potentially high-dimensional covariates and/or unknown functional forms. In strong contrast to \cite{chen2015bounds}, our estimators are asymptotically linear, and thus confidence intervals have the standard ``estimate $\pm$ standard error $\times$ critical value'' form and do not rely on alternative concepts such as half-median-unbiasedness \citep{chernozhukov2013intersection}.

Recent research has generalized estimation and inference of Lee-type bounds to high-dimensional or otherwise flexible settings using semiparametric DML methodology. 
\cite{semenova2025generalized} extends \cite{lee2009training} intensive margin bounds to multiple outcomes and provides orthogonal moments for debiased estimation. \cite{heiler2024heterogeneous} develops a DML procedure to obtain corresponding heterogeneous treatment effect bounds and inference under local misspecification. \cite{heiler2024treatmentevaluationintensiveextensive} provide refined DML methods for outer identification regions for intensive and extensive margin effects that admit regular inference in a larger class of DGPs. These contributions all focus on inference for (conditional) ITT effects. 
We instead construct semiparametrically efficient orthogonal moment functions for the sharp bounds on the treatment effect for compliers at the intensive margin. They nest several leading moments from the literature, including for ITT effect bounds (Semenova 2025, under perfect compliance), LATE (Fr\"olich 2007\nocite{frolich2007nonparametric}, in the absence of sample selection), and ATE (Hahn 1998\nocite{hahn1998role}, under perfect compliance and no sample selection). 
%
%

A crucial difference to the literature on DML inference on Lee-type ITT bounds is that our sharp bounds require trimming based on a population whose quantiles are identified only \textit{indirectly} via inversion of a compliance weighted cumulative distribution in the sense of \cite{abadie2003semiparametric}. The latter by itself is a linear, but not necessarily convex, combination of two conditional CDFs. As a result, debiasing relies on an implicit function characterization. We leverage this representation to impose learning rate requirements only on primary, reduced-form type conditional CDFs and conditional means instead of the counterfactual conditional quantiles used for trimming the relevant observed strata distributions. This also has the advantage of making it straightforward in practice to guarantee non-crossing properties for the conditional quantiles involved, e.g., via inversion of isotonic (distributional) regression \citep{henzi2021isotonic}. Our implementation also follows this inversion strategy. Monte Carlo simulations suggest good coverage and power properties of the inference method in finite samples. 

We apply and compare our method to evaluate the effects of Job Corps on earnings and Medicaid on healthcare utilization. Under the strong sample selection monotonicity assumption, our sharp bounds tighten the existing \cite{chen2015bounds} bounds by 68.4\% (JC) and 31.1\% to 69.4\% (OHIE, depending on outcome). 
For example, the sharp bounds without covariates for the intensive margin complier effect of JC on hourly wages are now $[0.019,0.067]$ compared to $[-0.022,0.130]$ for CF. 
Similar reductions apply to the 95\%-confidence intervals for the effects. Their widths are reduced by 46.7\% (JC) and 25.8\% to 36.9\% (OHIE, depending on outcome). 
Under weak sample selection monotonicity, our sharp DML bounds still tend to be contained by the restrictive CF bounds for most outcomes and, despite relying on weaker assumptions, continue to yield shorter confidence intervals of around 7.9\% (JC) and 7.3\% to 42.3\% (OHIE, depending on outcome).

The rest of the paper is organized as follows: Section \ref{sec:bounds} derives the basic bounds without covariates and compares them with Lee and CF bounds. Section \ref{sec:bounds_covariates} extends these bounds to incorporate covariates. Section \ref{sec:inference} develops estimation and inference. 
Section \ref{sec:profiling} presents the profiling method for the targeted complier population at the intensive margin. Section \ref{sec:JC} provides the JC application. 
Section \ref{sec:conclusion} concludes. The OHIE application, Monte Carlo simulations as well as proofs and supplementary derivations are in the Supplementary Appendix.

\section{Basic Bounds Without Covariates} \label{sec:bounds}

\subsection{Model and Identification of Basic Bounds}

\label{subsec:baseline_bounds}

Let $D\in\{0,1\}$ denote a binary treatment and $Z\in\{0,1\}$ a binary instrument. Let $S\in\{0,1\}$ be a sample selection indicator for a continuous outcome $Y\in\mathcal{Y}\subset\mathbb{R}$ that is observed only if $S=1$. The observed data are independent draws of $(SY,S,D,Z)$.
For $d,z\in\{0,1\}$, let $Y_{d,z}$ and $S_{d,z}$ denote, respectively, the potential outcome and selection indicator that would be observed if treatment and the instrument were set exogenously to $(d,z)$. Let $D_{z}$ denote the potential treatment if $Z$ were set to $z$. We impose the following IV assumptions: 

\begin{assumption}[IV]
\textbf{\label{ass:IV}}
\end{assumption}

\ref{ass:IV}.1 \textit{(Exclusion) For $d=0,1$, $Y_{d,1}=Y_{d,0}=:Y_{d}$ and $S_{d,1}=S_{d,0}=:S_{d}$.}

\ref{ass:IV}.2 \textit{(Independence) $(Y_{1},Y_{0},S_{1},S_{0},D_{1},D_{0})\perp Z$.}

\ref{ass:IV}.3 \textit{(Strong Treatment Response Monotonicity) $P(D_{1}\geq D_{0})=1$.}

\ref{ass:IV}.4 \textit{(First Stage) $P(D_{1}\ne D_{0})>0$.}

\ref{ass:IV}.5 \textit{(Non-trivial Assignment) $P(Z=1)\in(0,1)$.}

\medskip

These assumptions are the standard IV/LATE assumptions \citep{imbens1994identification,angrist1996identification} with the added requirement that instrument $Z$ is excluded from directly causing selection (2.1.1) and is allocated independently of potential selection (2.1.2). This matches classic encouragement or eligibility designs with selected samples where it is credible that the instrument causes selection and outcome only indirectly through the treatment. 

 
Within this framework, types and causal effects can be decomposed into extensive and intensive margins. Point identification of the extensive margin effect is possible only for \textit{compliers} ${D_{1}>D_{0}}$. In particular, the extensive-margin effect for compliers is
\begin{equation}
\theta_{E}:=E[S_{1}-S_{0}| D_{1}>D_{0}],
\end{equation}
which under Assumption \ref{ass:IV} is point identified by the usual Wald ratio
\begin{equation}
\theta_{E}=\frac{E[S| Z=1]-E[S| Z=0]}{E[D| Z=1]-E[D| Z=0]}. \label{eq:thetaE}
\end{equation}
Our main target parameter is the average causal effect for compliers at the intensive margin, i.e.,~the compliers $D_1 > D_0$ who are selected regardless of treatment response $S_1 = S_0 = 1$. 
We refer to this parameter as the \textit{always-selected} or \textit{survivor local average treatment effect} (SLATE):
\begin{equation}
\theta_{SLATE}:=E[Y_{1}-Y_{0}| S_{1}=S_{0}=1,D_{1}>D_{0}]. \label{eq:SLATE_def}
\end{equation}

\medskip
\par{\textbf{Remark 1.}}
\textit{Without noncompliance, $\theta_{SLATE}$ reduces to the intensive margin/always-selected ATE in \cite{lee2009training}. Without sample selection, $\theta_{SLATE}$  reduces to the LATE \citep{imbens1994identification}. Without both, it collapses to the ATE. $\theta_{SLATE}$ is the natural intensive margin analogue of the LATE: It focuses on always-selected compliers, whose treatment status is changed by the instrument while their selection status is not. It is therefore a policy effect on the level of the outcome for a principal stratum, uncontaminated by selection into the observed sample or noncompliance. The policy relevance of $\theta_{SLATE}$ has to be discussed on a case-by-case basis. We note that, in what follows, our assumptions will be able to
identify the share of always-selected compliers as well as their observable characteristics. This can be used to compare their features with those of the overall population or other relevant groups, thereby assessing how substantive the parameter is in a given empirical setting. We provide all details in Section \ref{sec:profiling} and empirical examples in Section \ref{sec:JC} and Appendix \ref{sec:OHIE}.}
\medskip

Under selected samples $\theta_{SLATE}$ is only partially identified because the joint selection status $(S_{1},S_{0})$ is not observed. In particular, even for compliers who are selected only under treatment $(S_{1}=1,S_{0}=0)$ or only under control $(S_{1}=0,S_{0}=1)$, either $Y_{0}$ or $Y_{1}$ is never observed, so their treatment effects are not identified without additional (untestable) assumptions. 
To make progress, we first state some standard IV identities under Assumption \ref{ass:IV}:
\begin{lemma}
\label{lem:IV_identities}
Suppose Assumption \ref{ass:IV} holds. Then, for any integrable function $g: \mathcal{Y}\rightarrow \mathbb{R}$,
{\begin{align}
P(D_{1}>D_{0})&=E[D| Z=1]-E[D| Z=0], \label{eq:PrC}\\
P(S_{1}=1,D_{1}>D_{0})&=E[DS| Z=1]-E[DS| Z=0], \label{eq:PrS1C}\\
E[g(Y_{1})| S_{1}=1,D_{1}>D_{0}]&=\frac{E[DSg(Y)| Z=1]-E[DSg(Y)| Z=0]}{E[DS| Z=1]-E[DS| Z=0]}. \label{eq:EY1S1C}
\end{align}}
Moreover, replacing $D$ with $(D-1)$ in \eqref{eq:PrS1C} and \eqref{eq:EY1S1C} identifies $P(S_{0}=1,D_{1}>D_{0})$ and $E[g(Y_{0})| S_{0}=1,D_{1}>D_{0}]$, respectively.
\end{lemma}
In particular, taking $g(Y)=\mathbbm{1}(Y\leq y)$ yields conditional CDFs $F_{Y_{1}| S_{1}=1,D_{1}>D_{0}}(y)$ and $F_{Y_{0}| S_{0}=1,D_{1}>D_{0}}(y)$, while taking $g(Y)=Y$ yields the corresponding conditional means.

Next, we introduce the complier strata that underlie the SLATE parameter. Among compliers ($D_{1}>D_{0}$), define

\begin{alignat}{2}
ac&:=\{S_{0}=S_{1}=1,\;D_{1}>D_{0}\} &\quad\text{always-selected compliers} \\
dc&:=\{S_{1}<S_{0},\;D_{1}>D_{0}\} &\quad\text{treatment-de-selected compliers} \\
cc&:=\{S_{1}>S_{0},\;D_{1}>D_{0}\} &\quad\text{treatment-selected compliers}
\end{alignat}
Rewriting \eqref{eq:SLATE_def} in terms of these strata yields

\begin{equation}
\theta_{SLATE}=E[Y_{1}| ac]-E[Y_{0}| ac].
\end{equation}
The set $\{S_{0}=1,D_{1}>D_{0}\}$ combines $ac$ and $dc$, while $\{S_{1}=1,D_{1}>D_{0}\}$ combines $ac$ and $cc$. Lemma \ref{lem:IV_identities} allows us to identify $P(S_{0}=1,D_{1}>D_{0})$ and $P(S_{1}=1,D_{1}>D_{0})$, which yield two linear restrictions on the shares of three unknown types ($ac$, $dc$, $cc$). To fully identify these shares, we further impose a strong sample selection monotonicity condition for the compliers:

\begin{assumption}[\textbf{Strong Sample Selection Monotonicity for Compliers}]
Either $P(S_{1}\geq S_{0}| D_{1}>D_{0})=1$ or $P(S_{1}\leq S_{0}| D_{1}>D_{0})=1$. \label{ass:monotoneS}
\end{assumption}
Assumption \ref{ass:monotoneS} imposes a common direction of treatment-induced sample selection within the complier stratum $D_{1} > D_{0}$. In particular, the treatment is assumed to either weakly increase or weakly decrease selection for all compliers. It cannot move some compliers into the sample and others out of it, effectively ruling out either $dc$ or $cc$ types.\footnote{This restriction is implied, for example, by an additively separable selection equation, see, e.g., \cite{vytlacil2002independence} for treatment selection or \cite{heiler2024treatmentevaluationintensiveextensive} for sample selection. However, these models impose monotonicity for the full population and are thus stronger than Assumption \ref{ass:monotoneS}.} Note that no analogous restriction is needed for strata with $D_{1} = D_{0}$ (always-takers and never-takers) as their treatment status does not vary with the instrument. In particular, never-takers satisfy $D_0=D_1=0$, so their observed selection status is governed only by $S_0$; always-takers satisfy $D_0=D_1=1$, so their observed selection status is governed only by $S_1$. The counterfactual selection status under the unrealized treatment therefore does not enter the identification argument.

We now provide identification for the case where $P(S_{1}\geq S_{0}| D_{1}>D_{0})=1$. The case where $P(S_{1}\leq S_{0}| D_{1}>D_{0})=1$ can be handled analogously.
Under Assumption \ref{ass:monotoneS}, the stratum ${S_{0}=1,D_{1}>D_{0}}$ consists only of always-selected compliers ($ac$), while ${S_{1}=1,D_{1}>D_{0}}$ combines $ac$ and $cc$. Lemma \ref{lem:IV_identities} then implies that

\begin{equation}
\beta_{0}:=E[Y_{0}| ac]=E[Y_{0}| S_{0}=1,D_{1}>D_{0}]
\end{equation}
is point identified by taking $g(Y)=Y$ and replacing $D$ with $(D-1)$ in \eqref{eq:EY1S1C}.
In contrast, $E[Y_{1}| ac]$ is not point identified. However, Lemma \ref{lem:IV_identities} and Assumption \ref{ass:monotoneS} allow us to identify the mixture distribution of $Y_{1}$ for $ac$ and $cc$,

\begin{equation}
F_{Y_{1}| ac\cup cc}(y):=F_{Y_{1}| S_{1}=1,D_{1}>D_{0}}(y),
\end{equation}
as well as the frequencies of $ac$ and $cc$ within this mixture. Let

\begin{equation}
\pi_{ac}:=P(S_{0}=S_{1}=1,D_{1}>D_{0}),\quad \pi_{cc}:=P(S_{1}>S_{0},D_{1}>D_{0}).
\end{equation}
Then, under Assumption \ref{ass:monotoneS},

\begin{equation}
\pi_{ac}=P(S_{0}=1,D_{1}>D_{0}),\qquad \pi_{ac}+\pi_{cc}=P(S_{1}=1,D_{1}>D_{0}),
\end{equation}
and both probabilities are identified. Moreover, define

\begin{equation}
p:=\frac{\pi_{ac}}{\pi_{ac}+\pi_{cc}}.
\end{equation}
$p$ is equal to the fraction of always-selected compliers in  mixture ${S_{1}=1,D_{1}>D_{0}}$. Since we only observe this mixture, sharp bounds are obtained by assigning $ac$ to the ``best'' or ``worst'' location within the mixture distribution of $Y_{1}$. Let $F_{1}(y):=F_{Y_{1}| S_{1}=1,D_{1}>D_{0}}(y)$ and $Q_{1}(u):=\inf\{y\in\mathcal{Y}:F_{1}(y)\geq u\}$ denote the identified CDF and quantile function of this mixture. 

To simplify exposition for the remainder of this section, we maintain the following regularity condition: the relevant identified distributions entering the trimming
formulas have continuous, strictly increasing CDFs. This rules out mass points at
the relevant trimming cutoffs and allows the bounds to be written as ordinary
trimmed means. This condition is not needed for identification per se.\footnote{A weaker local regularity condition would suffice: for each result below, it is enough that the corresponding identified CDF be continuous and strictly increasing in a neighborhood of the relevant trimming threshold(s). Without this regularity, the same conclusions also hold once the bounds are written in exact fractional-trimming (quantile-integral)
form.}

The lower and upper bounds for $E[Y_{1}| ac]$ are
\begin{align}
\beta_{L,1} &:=E[Y_{1}| Y_{1}\leq Q_{1}(p),ac \cup cc]=E[Y_{1}| Y_{1}\leq Q_{1}(p),S_{1}=1,D_{1}>D_{0}], \label{eq:betaL1} \\
\beta_{U,1} &:=E[Y_{1}| Y_{1}\geq Q_{1}(1-p),ac \cup  cc]=E[Y_{1}| Y_{1}\geq Q_{1}(1-p),S_{1}=1,D_{1}>D_{0}], \label{eq:betaU1}
\end{align}
i.e., the trimmed means obtained by placing all $ac$ individuals in the bottom $p$ fraction (for the lower bound) or the top $p$ fraction (for the upper bound) of the mixture. We obtain the following proposition:

\begin{proposition}[Sharpness]
\label{prop:sharp_bounds}
Suppose Assumptions \ref{ass:IV}-- \ref{ass:monotoneS} hold and $\pi_{ac}>0$. Then,

\begin{equation*}
\beta_{L,1}\leq E[Y_{1}| ac]\leq \beta_{U,1}.
\end{equation*}
Moreover, $\beta_{L,1}$ and $\beta_{U,1}$ are sharp: They are, respectively, the largest lower bound and the smallest upper bound for $E[Y_{1}| ac]$ that are consistent with the observed data and Assumptions \ref{ass:IV}--\ref{ass:monotoneS}. Any other valid bounds under these assumptions contain $[\beta_{L,1},\beta_{U,1}]$.
\end{proposition}

Combining the point-identified $\beta_{0}$ with these bounds yields sharp bounds for $\theta_{SLATE}$:

\begin{equation}
\beta_{L}:=\beta_{L,1}-\beta_{0},\qquad \beta_{U}:=\beta_{U,1}-\beta_{0}. \label{eq:theta_bounds}
\end{equation}
Using Assumption \ref{ass:monotoneS} and Lemma \ref{lem:IV_identities} with $g(Y)=Y\mathbbm{1}(Y\leq Q_{1}(p))$ and $g(Y)=Y\mathbbm{1}(Y\geq Q_{1}(1-p))$, all components of $\beta_{L}$ and $\beta_{U}$, namely $\pi_{ac}$, $p$, $F_{1}$, $Q_{1}$, and the trimmed means are functionals of the observed distribution of $(SY,S,D,Z)$. In particular,

{\begin{align}
\beta_{L}=\frac{1}{\pi_{ac}}\left\{
\begin{array}{c}
E\left[DSY\mathbbm{1}(Y\leq Q_{1}(p))+(1-D)SY| Z=1\right]\\
-\,E\left[DSY\mathbbm{1}(Y\leq Q_{1}(p))+(1-D)SY| Z=0\right]
\end{array}
\right\}, \label{eq:betaL} \\
\beta_{U}=\frac{1}{\pi_{ac}}\left\{
\begin{array}{c}
E\left[DSY\mathbbm{1}(Y\geq Q_{1}(1-p))+(1-D)SY| Z=1\right]\\
 -\,E\left[DSY\mathbbm{1}(Y\geq Q_{1}(1-p))+(1-D)SY| Z=0\right]
\end{array}
\right\}, \label{eq:betaU}
\end{align}}
where

\begin{align}
    p &=\frac{E\left[ \left( D-1\right) S| Z=1\right] -E\left[ \left(
D-1\right) S| Z=0\right] }{E\left[ DS| Z=1\right] -E\left[ DS| Z=0 \right] }, \\
Q_{1}(u ) &= \inf \{y \in \mathcal{Y}:F_{1}(y)\geq u \} \text{ for any } u \in \left( 0,1\right), \\
F_{1}(y) &=%
\frac{E\left[ \mathbbm{1}\{Y\leq y\}DS| Z=1\right] -E\left[ \mathbbm{1}%
\{Y\leq y\}DS| Z=0\right] }{E\left[ DS| Z=1\right] -E\left[ DS| Z=0%
\right] }.
\end{align} 
We next compare our bounds with two benchmarks: the ITT bounds of \cite{lee2009training} and the CF bounds, which target $\theta_{SLATE}$ under essentially identical identification assumptions.



\subsection{Comparison with Existing Bounds}
\label{subsec:Bound_comparison}
\subsubsection{Comparison with \cite{lee2009training} Bounds}
\label{subsec:Lee_comparison}

Without sample selection, the ITT effect
\begin{equation}
\tau_{Y}:=E[Y| Z=1]-E[Y| Z=0]
\end{equation}
and the average treatment effect for compliers,
\begin{equation}
\theta_{LATE} := \frac{\tau_{Y}}{\tau_{D}},\qquad \tau_{D}:=E[D| Z=1]-E[D| Z=0], \label{eq:LATE1}
\end{equation}
differ only by the scaling factor $\tau_{D}$, which corresponds to the share of compliers. That is, in the standard IV point identification setting, ITT can be converted into the complier average treatment effect by dividing by $\tau_{D}$. This does not generalize to treatment effect bounds under sample selection. 

In particular, in the presence of sample selection, \cite{lee2009training} derives sharp bounds for the ITT effect among the always-selected assuming strong sample selection monotonicity. These ``Lee bounds'' can be viewed as a special case of our basic bounds under perfect compliance, i.e., when treatment equals instrument, $D=Z$. We only consider the lower bound. The analysis for the upper bound is analogous and omitted for brevity. Setting $D=Z$ in Equation \eqref{eq:betaL}, we obtain the bound

{\begin{align}
\beta_{L}^{Lee} &:=E[Y| Y\leq Q_{1}^{Lee}(p^{Lee}),S=1,Z=1]-E[Y| S=1,Z=0] \notag \\
&= \beta_{L,1}^{Lee}-\beta_{0}^{Lee}, \label{eq:Lee_lower_def}
\end{align}}
where $Q_{1}^{Lee}(u):=\inf \{y\in \mathcal{Y}: F_{Y|S=1,Z=1}(y)\ge u \}$  is the $u$-quantile of $\left(Y|S=1,Z=1\right)$, and $p^{Lee}:=E[S| Z=0]/E[S| Z=1]$ is the trimming fraction.

To relate \eqref{eq:Lee_lower_def} to our SLATE bounds, it is useful to decompose the selected population into latent types. In addition to the complier types introduced in Section~\ref{subsec:baseline_bounds}, let

\begin{align}
an &:=\{S_{0}=1,D_{0}=D_{1}=0\}\quad\text{always-selected never-takers}, \\
aa &:=\{S_{1}=1,D_{0}=D_{1}=1\}\quad\text{always-selected always-takers},
\end{align}
and denote their shares by $\pi_{an}$ and $\pi_{aa}$ respectively. Under Assumptions \ref{ass:IV}--\ref{ass:monotoneS}, the selected group in each arm of the experiment is a mixture of these types. We show in Appendix~\ref{app:Lee_decomposition} that

\begin{align}
F_{Y| S=1,Z=1}(y)&=\omega F_{Y_{0}| an}(y)+(1-\omega)F_{Y_{1}| ac\cup cc\cup aa}(y), \label{eq:FY_SZ1_mix}\\
F_{Y| S=1,Z=0}(y)&=\delta F_{Y_{0}| ac\cup an}(y)+(1-\delta)F_{Y_{1}| aa}(y), \label{eq:FY_SZ0_mix}
\end{align}
for suitable mixture weights $\omega,\delta\in(0,1)$ depending only on type shares. Let $q :=Q^{Lee}_1(p^{Lee})$ denote the Lee trimming threshold. The trimmed mean in the first term of \eqref{eq:Lee_lower_def} can be written as
\begin{align}
\beta^{Lee}_{L,1}
&:= E\!\left[Y \mid Y \le q,\, S=1,\, Z=1\right] \notag \\
&= w(q)\,E\!\left[Y_0 \mid Y_0 \le q,\, an\right]
+ \bigl(1-w(q)\bigr)\,E\!\left[Y_1 \mid Y_1 \le q,\, ac\cup cc\cup aa\right],    
\end{align}
where $w\left( q\right)=\tfrac{\omega F_{Y_0\mid an}(q)}{p^{Lee}}$ is the effective $an$ stratum weight from trimming applied after mixing. The second term of \eqref{eq:Lee_lower_def} can be written as 
{ \begin{align}
\beta _{0}^{Lee} &:=E\left[ Y|S=1,Z=0\right]  \notag \\
&=E\left[ Y_{0}|ac \cup an\right] \delta +E\left[ Y_{1}|aa\right] \left( 1-\delta \right) .
\end{align}}%
Thus, $\beta_{L}^{Lee}$ is a difference of two mixtures that involve not only treatment compliers ($ac$, $cc$), but also always-takers and never-takers ($aa$, $an$). In contrast, our lower bound for SLATE,

\begin{equation}
\beta_{L}=E[Y_{1}| Y_{1}\leq Q_{1}(p),ac \cup cc]-E[Y_{0}| ac],
\label{eq:sharp_basic_lower}
\end{equation}
depends only on the treatment complier types ($ac$ and $cc$), with $p$ equal to the fraction of always-selected compliers within that complier mixture.

Except in knife-edge cases where the shares of always-selected always-takers and never-takers vanish ($\pi_{aa}=\pi_{an}=0$), the $Y_{1}$ term in $\beta_{L}^{Lee}$ does not reduce to $E[Y_{1}| Y_{1}\leq Q_{1}(p),ac \cup cc]$, and the $Y_{0}$ term does not reduce to $E[Y_{0}| ac]$. Consequently, \cite{lee2009training} bounds cannot be transformed into sharp bounds for $\theta_{SLATE}$ by simple rescaling via primitive probabilities. We summarize this result formally in the following proposition. 
\begin{proposition}[Lee bounds do not rescale to SLATE bounds]
\label{prop:Lee_not_scaling}
Suppose Assumptions~\ref{ass:IV}--\ref{ass:monotoneS} hold and $\pi_{ac}>0$.
If $\pi_{an}+\pi_{aa}>0$, then, for generic joint distributions of $(Y_0,Y_1)$
across latent types,
\begin{align*}
\beta_L^{Lee}\neq \beta_L.
\end{align*}
Moreover, there does not exist a scalar function $c$, depending only on the
primitive probabilities (i.e.~the principal-strata shares and instrument probabilities), such that
\begin{align*}
\beta_L = c\,\beta_L^{Lee}
\end{align*}
uniformly over DGPs.
\end{proposition}


Proposition \ref{prop:Lee_not_scaling} implies that \cite{lee2009training} ITT bounds cannot be converted into sharp bounds for SLATE by a simple rescaling analogous to $\tau_{Y}/\tau_{D}$ in the standard IV case without sample selection.

\subsubsection{Comparison with \cite{chen2015bounds} Bounds}\label{subsec:CF_comparison} 

CF study partial identification of $\theta_{SLATE}$ under essentially identical model assumptions. We show that our bounds are 
weakly tighter uniformly over DGPs, and strictly tighter on a nonempty subset of DGPs than the CF bounds. For brevity, we focus on the lower bounds. The upper bounds can be treated analogously.

CF propose two lower bounds for $E[Y_{1}| ac]$, based on trimming the distribution of $Y$ among selected treated units, $Y| DSZ=1$, equivalent to $Y_1|ac\cup cc\cup aa$ under Assumptions \ref{ass:IV}--\ref{ass:monotoneS}, and then taking the maximum of the two. In contrast, our sharp bounds trim within the smaller identified mixture $Y_1|ac\cup cc$. Let $Q_{1}^{CF}(u)$ denote the $u$-quantile of $Y| DSZ=1$. Their basic lower bound is

{\begin{align}
\beta_{L}^{mix}&:=E\left[Y| DSZ=1,Y\leq Q_{1}^{CF}(p^{CF}_1)\right]-\beta_{0} \nonumber\\
&=\beta_{L,1}^{mix}-\beta_{0},
\label{eq:CF_basic_lower}
\end{align}}

where $p^{CF}_1:=\pi_{ac}/(\pi_{ac}+\pi_{cc}+\pi_{aa})$ and $\beta_{0}:=E[Y_{0}| ac]$. Their alternative lower bound uses additional information from the always-taker stratum ($aa$),

{\begin{align}
\beta_{L}^{adj}&:=\left\{
\begin{array}{l}
\left(1+\dfrac{\pi_{aa}}{\pi_{ac}}\right)
E\left[Y\mid DSZ=1,Y\leq Q_{1}^{CF}(p^{CF}_2)\right]\\[6pt]
-\dfrac{\pi_{aa}}{\pi_{ac}}E\left[Y\mid DS(1-Z)=1\right]
\end{array}
\right\}-\beta_{0}\nonumber\\
&=\beta_{L,1}^{adj}-\beta_{0}, \label{eq:CF_alt_lower}
\end{align}}
where $p^{CF}_2:=(\pi_{ac}+\pi_{aa})/(\pi_{ac}+\pi_{cc}+\pi_{aa})$ and $E[Y| DS(1-Z)=1]=E[Y_{1}| aa]$ is point identified. The CF lower bound is then
\begin{equation}
\beta_{L}^{CF}:=\max\{\beta_{L}^{mix},\beta_{L}^{adj}\}.
\end{equation}

Under Assumptions \ref{ass:IV}--\ref{ass:monotoneS}, the conditional distribution $Y| DSZ=1$ coincides with $Y_{1}$ for the mixture of types $ac$, $cc$, and $aa$:

\begin{equation}
F_{Y| DSZ=1}(y)=F_{Y_{1}| ac\cup cc\cup aa}(y),
\end{equation}
so both $\beta_{L}^{mix}$ and $\beta_{L}^{adj}$ are obtained by trimming this mixture distribution and then reweighting to isolate $ac$ and $aa$. CF show that $\beta_{L}^{mix}$ corresponds to the worst-case lower bound when all $ac$ individuals are placed in the bottom $p^{CF}_1$-
fraction of $Y_{1}| ac\cup cc\cup aa$, while $\beta_{L}^{adj}$ corresponds to the worst case when $ac$ and $aa$ together occupy the bottom $p^{CF}_2$-
fraction of that distribution, using the fact that $E[Y_{1}| aa]$ is identified.

In contrast, our lower bound $\beta_{L}$ uses the smallest stratum that is point identified, namely $Y_{1}| ac\cup cc$. As shown in Lemma \ref{lem:F_Y1_ac_cc_identification} in Appendix \ref{app:CF_dominance}, the distribution $F_{Y_{1}| ac\cup cc}$ can be recovered from the two observable selected distributions

\begin{equation}
F_{Y| DSZ=1}(y)=F_{Y_{1}| ac\cup cc\cup aa}(y),\qquad F_{Y| DS(1-Z)=1}(y)=F_{Y_{1}| aa}(y).
\end{equation}
We then construct the sharp lower bound

\begin{equation}
\beta_{L}=E\left[Y_{1}| Y_{1}\leq Q_{1}(p),ac \cup cc\right]-\beta_{0},\qquad p:=\frac{\pi_{ac}}{\pi_{ac}+\pi_{cc}},
\end{equation}
by trimming $Y_{1}$ only within this complier mixture.

Intuitively, $\beta_{L}$ assumes, in a worst-case fashion, that all $ac$ individuals are located below all $cc$ individuals in the distribution of $Y_{1}$ among compliers, but it imposes no restrictions on the relative ranking of $ac$ versus $aa$. By contrast, the CF bounds are based on worst-case restrictions on the joint ranking of $ac$, $cc$, and $aa$ within $Y_{1}| ac\cup cc\cup aa$. This makes them conservative. In particular, we obtain the following proposition:



\begin{proposition}[Dominance over CF bounds]
\label{prop:CF_dominance_main}
Suppose Assumptions \ref{ass:IV}--\ref{ass:monotoneS} hold and $\pi_{ac}>0$. Then,
\begin{align*}
\beta_{L,1}\geq \beta_{L,1}^{CF}\geq \beta_{L,1}^{mix} \quad \text{ and } \quad
\beta_{U,1}\leq \beta_{U,1}^{CF}\leq \beta_{U,1}^{mix},
\end{align*}
where $\beta_{L,1}$ is defined in \eqref{eq:betaL1}, $\beta_{L,1}^{mix}$ in
\eqref{eq:CF_basic_lower}, and
$\beta_{L,1}^{CF}:=\max\{\beta_{L,1}^{mix},\beta_{L,1}^{adj}\}$,
with $\beta_{L,1}^{adj}$ defined in \eqref{eq:CF_alt_lower}. The upper-bound
counterparts $\beta_{U,1}$, $\beta_{U,1}^{mix}$, $\beta_{U,1}^{adj}$, and
$\beta_{U,1}^{CF}$ are defined analogously. Hence, imposing either of the two restrictions used by CF cannot tighten our sharp lower or upper bound
for $E[Y_1\mid ac]$.

\noindent
Moreover, these inequalities are strict on a nonempty class of DGPs:
\begin{align*}
\beta_{L,1}>\beta_{L,1}^{CF}
&\iff
\pi_{aa}>0 \quad\text{and}\quad 0<F_{Y_1\mid aa}\!\bigl(Q_1(p)\bigr)<1,  \\
\beta_{U,1}<\beta_{U,1}^{CF}
&\iff 
\pi_{aa}>0 \quad\text{and}\quad0<F_{Y_1\mid aa}\!\bigl(Q_1(1-p)\bigr)<1.
\end{align*} 
\end{proposition}

\section{Bounds with Covariates}
\label{sec:bounds_covariates}

We now extend the unconditional baseline bounds in Section~\ref{subsec:baseline_bounds} to incorporate predetermined covariates.
Incorporating such covariates serves three purposes: First, it allows for an unconfounded instead of a fully randomly assigned instrument. Second, even under fully randomized assignment, it can improve efficiency by conditioning on relevant covariates when constructing trimmed means. Third, it allows us to relax strong sample selection monotonicity. In particular, we permit the direction of sample selection for treatment compliers to vary with covariates. For simplicity, we maintain the strong treatment response monotonicity assumption \ref{ass:IV}.3. 
This no-defiers restriction is natural in eligibility and encouragement
designs such as JC and OHIE: departures from perfect compliance are largely one-sided, arising mainly from incomplete take-up among those assigned to treatment rather than from substantial treatment receipt among controls. Relaxing this assumption is possible via further covariate partitioning analogous to the sample selection case in what follows. 

Let $X\in\mathcal{X}\subset\mathbb{R}^{d_{x}}$ denote a vector of predetermined covariates. The observed data are $(SY,S,D,Z,X)$. We impose the following conditional IV assumptions: 

\begin{assumption}[Conditional IV]
\label{ass:conditional_IV}
For $d=0,1$:

\ref{ass:conditional_IV}.1 (Exclusion) $Y_{d,1}=Y_{d,0}=:Y_{d}$, $S_{d,1}=S_{d,0}=:S_{d}$.

\ref{ass:conditional_IV}.2 (Independence) $(Y_{1},Y_{0},S_{1},S_{0},D_{1},D_{0})\perp Z| X$.

\ref{ass:conditional_IV}.3 (Strong Treatment Response Monotonicity) $P(D_{1}\geq D_{0})=1$. 

\ref{ass:conditional_IV}.4  (First Stage) $P(D_1 \ne D_0 \mid X=x)>0$ for $P_X$-almost every $x \in \mathcal{X}$

\ref{ass:conditional_IV}.5 
(Non-trivial Assignment)
$P(Z=1\mid X=x)\in(0,1)$ for $P_X$-almost every $x\in\mathcal{X}$.

\end{assumption}

These assumptions are standard extensions of the conditional IV/LATE assumptions \citep{frolich2007nonparametric} generalizing Assumption \ref{ass:IV} with instrument $Z$ again being excluded from directly causing selection and being independently allocated of potential selection given covariates. 

It is important to note that Assumptions \ref{ass:conditional_IV}.4 and \ref{ass:conditional_IV}.5 are not
intended to impose substantive restrictions beyond their unconditional counterparts
Assumptions \ref{ass:IV}.4 and \ref{ass:IV}.5. Rather, they should be interpreted
as support restrictions for the conditional analysis: if the conditional first-stage
or non-trivial assignment condition fails on a set of positive $P_X$-probability,
$\mathcal{X}$ should be understood as restricted to the support on which these
conditions hold.\footnote{Failure of the first stage and failure of non-trivial assignment have slightly different interpretations. If $P(D_1\ne D_0\mid X=x)=0$, there are no compliers at that $x$. so complier-based conditional objects are not meaningful there. If $P(Z=1\mid X=x)\in\{0,1\}$, there is no within-$x$ variation in the instrument, so the conditional IV comparisons are not identified there.}

Under Assumption~\ref{ass:conditional_IV}, the identities in Lemma~\ref{lem:IV_identities} hold pointwise in $x$.
Next, we relax strong sample selection monotonicity to a weaker conditional version that allows the direction of monotonicity to vary with covariates.


\begin{assumption}[Weak Sample Selection Monotonicity for Compliers]\label{ass:conditional_monotoneS} 
Either \\ \noindent $P( S_1 \geq S_0 | D_{1}>D_{0},X=x)=1$ or 
    $P( S_1 \leq S_0 | D_{1}>D_{0}, X=x)=1.$
\end{assumption}
Assumption \ref{ass:conditional_monotoneS} yields subsets of the covariate space where treatment compliers can have either a weakly positive or negative selection responses to treatment. At their intersection, treatment does not affect selection. Without loss of generality, we now define a distinct partitioning where this intersection is combined with the weakly positive responders. Formally, let \begin{align*}
    \mathcal{X}^0 &=\left\{ x \in \mathcal{X}: P(S_1=S_0|D_{1}>D_{0},X=x)=1 \right\}
\end{align*}
and partitions
\begin{align}
    \mathcal{X}^+ &= \{ x \in \mathcal{X}: P(S_1 \geq S_0|D_{1}>D_{0},X=x)=1,  \notag \\ &\qquad  P(S_1>S_0|D_{1}>D_{0},X=x)>0 \} \cup  \mathcal{X}^0,  \\
    \mathcal{X}^- &=\{  x \in \mathcal{X}:   P(S_1\leq S_0|D_{1}>D_{0},X=x)=1,  \notag \\ &\qquad P(S_1<S_0|D_{1}>D_{0},X=x)>0 \}.
\end{align}




We note that the presence of $\mathcal{X}^0$ is harmless for identification. For regular semiparametric inference, however, it will be necessary that $P(\mathcal{X}^0) = 0$ \citep{heiler2024treatmentevaluationintensiveextensive}. We return to this in Section \ref{sec:inference}.

\medskip
\par{\textbf{Remark 2.}}
\textit{Assumption \ref{ass:conditional_monotoneS} weakens strong sample selection monotonicity by allowing the sign of the treatment effect on selection to vary with observed covariates. Thus, treatment may increase selection for some values in $\mathcal{X}$ and decrease it for others. The restriction is nevertheless substantive: conditional on $X=x$ and complier status, residual heterogeneity in the selection response must be one-sided. In particular, the assumption rules out the coexistence of treatment-induced entry and exit among compliers with the same covariate values.
Its credibility therefore depends on whether $\mathcal{X}$ is rich enough to absorb the economically relevant heterogeneity in the direction of selection. It is more plausible when the main determinants of the sign of the selection response are observed, such as baseline characteristics governing participation, employment, survival, etc. It is less plausible when unobserved factors can generate opposing responses within the same covariate cell, see \cite{heiler2024treatmentevaluationintensiveextensive} and \cite{semenova2025generalized} for additional discussion. 
}
\medskip

Let $\pi_{ac}(x):=P(ac| X=x)$ denote the conditional share of always-selected compliers, and let $\pi_{ac} := E[\pi_{ac}(X)]$ denote their unconditional share. Further let $\mathcal X_{ac}:=\{x\in\mathcal X:\pi_{ac}(x)>0\}$.\footnote{For identification of the unconditional bounds, the substantive requirement is $\pi_{ac}>0$. $\pi_{ac}(x)$ may be zero on subsets of the covariate support as they receive zero weight in numerator and denominator of the unconditional bounds. 
As a convention for display, we treat $\theta_{SLATE}(x)=0$ for $x\in\mathcal X\setminus\mathcal X_{ac}$.} 
For any $x \in \mathcal X_{ac}$, define the conditional SLATE parameter

\begin{equation}
\theta_{SLATE}(x):=E[Y_{1}-Y_{0}| ac,X=x].
\end{equation}
Let
\begin{equation}
\lambda_{d}(x):=P(S_{d}=1,D_{1}>D_{0}| X=x),\quad d=0,1,
\end{equation}
and write $\pi_{cc}(x):=P(cc| X=x)$ and $\pi_{dc}(x):=P(dc| X=x)$ for the conditional shares of $cc$ and $dc$, respectively. Under Assumption~\ref{ass:conditional_monotoneS},
\begin{align}
\lambda_{0}(x) &=\pi_{ac}(x),\quad \lambda_{1}(x)=\pi_{ac}(x)+\pi_{cc}(x)\quad \text{for }x\in\mathcal{X}^{+}, \\
\lambda_{1}(x) &=\pi_{ac}(x),\quad \lambda_{0}(x)=\pi_{ac}(x)+\pi_{dc}(x)\quad \text{for }x\in\mathcal{X}^{-}.
\end{align}
This yields the ratio
\begin{equation}
p(x):=\frac{\lambda_{0}(x)}{\lambda_{1}(x)}.
\end{equation}
Thus $0 \le p(x) \leq 1$ on $\mathcal{X}^+_{ac}:=\mathcal{X}^+\cap\mathcal{X}_{ac}$ and $p(x) > 1$ on $\mathcal{X}^-_{ac}:=\mathcal{X}^-\cap\mathcal{X}_{ac}$. For $d=0,1$ and $x\in \mathcal{X}_{ac}$, define the conditional quantile
\begin{equation}
Q_{d}(u,x):=\inf\left\{y\in\mathcal{Y}:F_{Y_{d}| S_{d}=1,D_{1}>D_{0},X=x}(y)\geq u\right\}.
\end{equation}

To simplify exposition, for the remainder of this section we again maintain that the
relevant identified conditional CDF entering each trimming formula is continuous
and strictly increasing in a neighborhood of the corresponding trimming
threshold(s). This rules out mass points at the trimming thresholds and allows the conditional sharp bounds below to be written as
ordinary trimmed means.\footnote{For $x\in \mathcal{X}^+_{ac}$, the relevant CDF
is $F_{Y_{1}| S_{1}=1,D_{1}>D_{0},X=x}(\cdot)$ and the relevant thresholds are $Q_1(p(x),x)$ and
$Q_1(1-p(x),x)$. For $x\in \mathcal{X}^-_{ac}$, the relevant CDF is $F_{Y_{0}| S_{0}=1,D_{1}>D_{0},X=x}(\cdot)$ and the
relevant thresholds are $Q_0(1-1/p(x),x)$ and $Q_0(1/p(x),x)$. Without this
regularity, the same bounds can be written in exact fractional-trimming form.}
Then Proposition~\ref{prop:sharp_bounds} can be applied pointwise in $x$. This yields the following sharp lower and upper bounds for $\theta_{SLATE}(x)$:
{\begin{align}
\beta_{L}^{+}(x)&:=E\left[Y_{1}| S_{1}=1,D_{1}>D_{0},Y_{1}\leq Q_{1}(p(x),x),X=x\right] \notag \\
&\quad - E\left[Y_{0}| S_{0}=1,D_{1}>D_{0},X=x\right], \label{eq:betaL_cond_plus}\\
\beta_{U}^{+}(x)&:=E\left[Y_{1}| S_{1}=1,D_{1}>D_{0},Y_{1}\geq Q_{1}(1-p(x),x),X=x\right] \notag \\
&\quad -E\left[Y_{0}| S_{0}=1,D_{1}>D_{0},X=x\right], \label{eq:betaU_cond_plus}
\end{align}} for $x\in\mathcal{X}^{+}_{ac}$, and

{\begin{align}
\beta_{L}^{-}(x)&:=E\left[Y_{1}| S_{1}=1,D_{1}>D_{0},X=x\right] \notag \\
&\quad -E\left[Y_{0}| S_{0}=1,D_{1}>D_{0},Y_{0}\geq Q_{0}(1-1/p(x),x),X=x\right], \label{eq:betaL_cond_minus}\\
\beta_{U}^{-}(x)&:=E\left[Y_{1}| S_{1}=1,D_{1}>D_{0},X=x\right] \notag \\
&\quad -E\left[Y_{0}| S_{0}=1,D_{1}>D_{0},Y_{0}\leq Q_{0}(1/p(x),x),X=x\right], \label{eq:betaU_cond_minus}
\end{align}}
for $x\in\mathcal{X}^{-}_{ac}$. As a notational convention, set $\beta_B^\pm(x)=0$ for $x\in \mathcal{X}\setminus \mathcal X_{ac}$, $B\in\{L,U\}$. These values are used only in the reduced-forms $\beta_B^\pm(x)\pi_{ac}(x)$ entering the unconditional bounds. We note that when treatment does not cause selection, conditional effects bounds collapse to a point, i.e.~for any $x \in  (\mathcal{X}^0 \cap \mathcal{X}_{ac}) \subseteq \mathcal{X}^+_{ac}$ we have that \begin{align}
\beta_{L}^{+}(x)& = \beta_{U}^{+}(x) =E\left[Y_{1}| S_{1}=1,D_{1}>D_{0},X=x\right] - E\left[Y_{0}| S_{0}=1,D_{1}>D_{0},X=x\right].  \label{eq:betaU_cond_plus}
\end{align}
Now let
\begin{equation}
\mathbbm{1}^{+}(x):=\mathbbm{1}(x\in\mathcal{X}^{+}),\qquad \mathbbm{1}^{-}(x):=\mathbbm{1}(x\in\mathcal{X}^{-}), \label{eq:indicators1-def}
\end{equation}
and define for any $B \in \{L,U\}$
\begin{equation}
\beta_{B}(x):=\mathbbm{1}^{+}(x)\beta_{B}^{+}(x)+\mathbbm{1}^{-}(x)\beta_{B}^{-}(x).
\end{equation}
The unconditional SLATE parameter can be written as
\begin{equation}
\theta_{SLATE}=\frac{E[\theta_{SLATE}(X)\pi_{ac}(X)]}{\pi_{ac}},
\end{equation}
The corresponding sharp lower and upper bounds are
\begin{equation}
\beta_{L}=\frac{E[\beta_{L}(X)\pi_{ac}(X)]}{\pi_{ac}},\qquad
\beta_{U}=\frac{E[\beta_{U}(X)\pi_{ac}(X)]}{\pi_{ac}}.
\end{equation}

To make the link with observed data explicit, we next provide the estimands of the denominator and numerator moment functions in $\beta_{L}$ and $\beta_{U}$. Define the instrument propensity score $e(x):= P(Z=1| X=x)$, and the inverse probability weight 

{\begin{equation}
W(Z,X):=\frac{Z}{e(X)}-\frac{1-Z}{1-e(X)}. \label{eq:ipw}
\end{equation}}
Extending Lemma~\ref{lem:IV_identities} to hold conditional on $X$, one can show that

{\begin{align}
\pi_{ac}&=E\left[\mathbbm{1}^{+}(X)\lambda_{0}(X)+\mathbbm{1}^{-}(X)\lambda_{1}(X)\right] \notag \\
&=E\left[\min\{\lambda_{0}(X),\lambda_{1}(X)\}\right]. \label{eq:Pr_ac_cov_uncond}
\end{align}}
where now

{\begin{align}
\lambda_{0}(x)&= E\left[W(Z,X)(D-1)S| X=x\right],\\
\lambda_{1}(x)&= E\left[W(Z,X)DS| X=x\right],\\
\mathbbm{1}^{+}(x)&=\mathbbm{1} (\lambda_{0}(x)\le \lambda_{1}(x)),\\
\mathbbm{1}^{-}(x)&=\mathbbm{1} (\lambda_{0}(x) > \lambda_{1}(x)).
\end{align}}
Thus $\pi_{ac}$ is identified. In addition, for any $y$ and $x \in \mathcal{X}_{ac}$,

{\begin{align}
F_{1}(y| x)&:=F_{Y_{1}| S_{1}=1,D_{1}>D_{0},X=x}(y)
=\frac{E\left[W(Z,X)DS\mathbbm{1}(Y\leq y)| X=x\right]}{E\left[W(Z,X)DS| X=x\right]}, \label{eq:F1_cond_main}\\
F_{0}(y| x)&:=F_{Y_{0}| S_{0}=1,D_{1}>D_{0},X=x}(y)
=\frac{E\left[W(Z,X)(D-1)S\mathbbm{1}(Y\leq y)| X=x\right]}{E\left[W(Z,X)(D-1)S| X=x\right]}. \label{eq:F0_cond_main}
\end{align}}
The conditional quantiles $Q_{d}(u,x)$ used for trimming in the conditional bounds in Equations  \eqref{eq:betaL_cond_plus} -- \eqref{eq:betaU_cond_minus} are defined as the generalized inverses of these CDFs and are thus identified. Using these identities, the following proposition collects the explicit estimands for the unconditional bounds $\beta_{L}$ and $\beta_{U}$. 

\begin{proposition}
\label{prop:betaL_betaU_covariates}
Suppose Assumptions \ref{ass:conditional_IV}--\ref{ass:conditional_monotoneS} hold and $\pi_{ac}>0$. Then the sharp lower and upper bounds for $\theta_{SLATE}$ are identified as
{\begin{equation*}
\beta_{L}=\frac{E[\beta_{L}(X)\pi_{ac}(X)]}{\pi_{ac}},\qquad
\beta_{U}=\frac{E[\beta_{U}(X)\pi_{ac}(X)]}{\pi_{ac}}, 
\end{equation*}}
where $\pi_{ac}$ is given by Equation \eqref{eq:Pr_ac_cov_uncond} and
\begin{align*}
\beta_{L}(x) &:=\mathbbm{1}^{+}(x)\beta_{L}^{+}(x)+\mathbbm{1}^{-}(x)\beta_{L}^{-}(x), \\ 
\beta_{U}(x) &:=\mathbbm{1}^{+}(x)\beta_{U}^{+}(x)+\mathbbm{1}^{-}(x)\beta_{U}^{-}(x), 
\end{align*}
with
{\small
\begin{align}
\beta_{L}^{+}(x)\pi_{ac}(x)&:=E\left[W(Z,X)\left(DSY\mathbbm{1}(Y\leq Q_{1}(p(x),x))+(1-D)SY\right)| X=x\right], \label{eq:num-plus-lower}\\
\beta_{L}^{-}(x)\pi_{ac}(x)&:=E\left[W(Z,X)\left(DSY+(1-D)SY\mathbbm{1}(Y\geq Q_{0}(1-1/p(x),x))\right)| X=x\right], \label{eq:num-minus-lower}\\
\beta_{U}^{+}(x)\pi_{ac}(x)&:=E\left[W(Z,X)\left(DSY\mathbbm{1}(Y\geq Q_{1}(1-p(x),x))+(1-D)SY\right)| X=x\right], \label{eq:num-plus-upper}\\
\beta_{U}^{-}(x)\pi_{ac}(x)&:=E\left[W(Z,X)\left(DSY+(1-D)SY\mathbbm{1}(Y\leq Q_{0}(1/p(x),x))\right)| X=x\right],\label{eq:num-minus-upper}
\end{align}}
on $\mathcal{X}_{ac}$, and $\beta_B^\pm(x)\pi_{ac}(x)=0$ by convention on $\mathcal X\setminus\mathcal X_{ac}$. $\lambda_d(x)$ and $W(z,x)$ are point identified on $\mathcal{X}$ and conditional quantiles $Q_d(\cdot,x)$ are point identified on $\mathcal X_{ac}$.
\end{proposition}

In summary, $(\beta_{L},\beta_{U})$ are functionals of the observed law of $(SY,S,D,Z,X)$ and a collection of nuisance functions (instrument propensity score, conditional means, and (inverted) conditional distribution functions). In the next section, we derive orthogonal moment conditions and semiparametric influence functions for these functionals, which will allow us to construct semiparametrically efficient debiased machine learning estimators for $(\beta_{L},\beta_{U})$ and confidence intervals for $\theta_{SLATE}$.

\section{Estimation and Inference}
\label{sec:inference}
\subsection{Nuisance Functions and Target Parameter}

The previous section shows that for any $B\in\{L,U\}$, the unconditional bound can be written as
\begin{equation}
\beta_{B}=\frac{E[\beta_{B}(X)\pi_{ac}(X)]}{\pi_{ac}}, \label{eq:betaB_cov_repeat}
\end{equation}
with $\beta_{B}(X)$ defined in Proposition \ref{prop:betaL_betaU_covariates} and $\pi_{ac}$ in \eqref{eq:Pr_ac_cov_uncond}. For notational convenience, we decompose the numerator of \eqref{eq:betaB_cov_repeat} as
\begin{equation}
E[\beta_{B}(X)\pi_{ac}(X)]=N_{0}^{+}+N_{0}^{-}+N_{B,1}^{+}+N_{B,1}^{-}, \label{eq:N0-NB-all1}
\end{equation}
where, for $B\in\{L,U\}$ and $\ast\in\{+,-\}$, we define
\begin{equation}
N_{0}^{\ast}:=E[\beta_{0}^{\ast}(X)\pi_{ac}(X)], \qquad
N_{B,1}^{\ast}:=E[\beta_{B,1}^{\ast}(X)\pi_{ac}(X)].
\end{equation}

Table \ref{tab:nuisance} contains all primary nuisance functions as well as all derived quantities that are used for the remainder of this section. We collect the required primary nuisance functions in vector
\begin{align}
\eta(x):=(e(x),\{r(z,x),m(z,x),\mu(z,x),\nu(z,x),F(\cdot|  d,z, x),G_{L}(\cdot| d,z, x),G_{U}(\cdot| d,z, x)\}_{d,z}).
\end{align}

\begin{table}[!h] 
\caption{Primary and Derived Nuisance Functions for Bounds $(\beta_{L},\beta_{U})$}
\label{tab:nuisance}
\centering{\footnotesize 
\begin{tabular}{c|c} \hline \\[-0.5ex]
Primary Nuisance Parameter in $\eta$ & Definition \\[1ex]\hline
&\\
$e(x)$ & $P(Z=1| X=x)$\\
$r(z,x)$ & $E[DS| Z=z,X=x]$\\
$m(z,x)$ & $E[(1-D)S| Z=z,X=x]$\\
$\mu(z,x)$ & $E[(1-D)SY| Z=z,X=x]$\\
$\nu(z,x)$ & $E[DSY| Z=z,X=x]$\\
$F(y| d,z,x)$ & $E[\mathbbm{1}(Y\leq y)| D=d,S=1,Z=z,X=x]$\\
$G_{L}(y| d,z,x)$ & $E[Y\mathbbm{1}(Y \leq y)| D=d,S=1,Z=z,X=x]$\\
$G_{U}(y| d,z,x)$ & $E[Y\mathbbm{1}(Y > y)| D=d,S=1,Z=z,X=x]$\\\hline
&\\[-0.5ex]
Derived Nuisance Parameter or Variable & Definition \\[1ex]\hline
&\\[-0.5ex]
$\lambda_{0}(x)$ & $-(m(1,x)-m(0,x))$\\
$\lambda_{1}(x)$ & $r(1,x)-r(0,x)$\\
$\pi_{ac}(x)$ & $\min\{\lambda_{0}(x),\lambda_{1}(x)\}$\\
$p(x)$ & $\dfrac{\lambda_{0}(x)}{\lambda_{1}(x)}$\\
$\mathbbm{1}^+(x)$ & $\mathbbm{1}(p(x) \leq 1)$ \\
$\mathbbm{1}^-(x)$ & $\mathbbm{1}(p(x) > 1)$ \\
$W(z,x)$ & $\dfrac{z}{e(x)}-\dfrac{1-z}{1-e(x)}$\\
$F_{1}(y| x)$ 
& $\dfrac{F(y| 1,1,x)r(1,x)-F(y| 1,0,x)r(0,x)}{r(1,x)-r(0,x)}$\\
$F_{0}(y| x)$ 
& $\dfrac{F(y| 0,1,x)m(1,x)-F(y| 0,0,x)m(0,x)}{m(1,x)-m(0,x)}$\\
$Q_{1}(u,x)$ & $\inf\{y\in\mathcal{Y}:F_{1}(y| x)\geq u\}$\\
$Q_{0}(u,x)$ & $\inf\{y\in\mathcal{Y}:F_{0}(y| x)\geq u\}$\\
$\psi_{L}^{+}(z,x)$ & $G_{L}(Q_{1}(p(X),X)| 1,z,x)$\\
$\psi_{U}^{+}(z,x)$ & $G_{U}(Q_{1}(1-p(X),X)| 1,z,x)$\\
$\psi_{L}^{-}(z,x)$ & $G_{L}(Q_{0}(1-1/p(X),X)| 0,z,x)$\\
$\psi_{U}^{-}(z,x)$ & $G_{U}(Q_{0}(1/p(X),X)| 0,z,x)$\\
&\\ \hline
\end{tabular}
}
\end{table}

\subsection{Semiparametric Influence Functions and Nuisance Functions} \label{sec:moment-functions1}
For any $B \in \{L,U\}$, we treat $\pi_{ac}$ in \eqref{eq:Pr_ac_cov_uncond} as well as the four components in \eqref{eq:N0-NB-all1} as target functionals and derive their semiparametric influence functions/efficient orthogonal moments. They are then combined to yield the efficient influence function and estimator for the respective $\beta_B$.
Denote the influence function operator $\IF: \Theta \rightarrow L^2(P,\mathbb{R}^q)$ as the operator that, for a parameter $\theta:\mathcal{P}\rightarrow \mathbb{R}^q$, returns its influence function, i.e.,~the standard Riesz-representer of the pathwise derivative \citep{bickel1993efficient, kennedy2024semiparametric}.\footnote{The influence function is defined via the functional derivative of the target parameter with respect to perturbations along the tangent space evaluated at the true probability distribution. For example, if we observe iid data $X\sim \mathcal{P}$ and the target functional is $E[X]$, then for $\IF(E[X]) = X - E[X]$. For a regular parameter under the nonparametric model $\mathcal{P}$, this is the unique semiparametric efficient influence function, see, e.g., \cite{hines2022demystifying}, \cite{kennedy2024semiparametric} or \cite{heiler2026heterogeneity} for additional examples.}

We now present the influence function for the target parameters $\beta_B$ for $B \in \{L,U\}$. We suppress dependence on nuisances and data in what follows whenever it does not cause confusion. By linearity of the $\IF$ operator, we have that  \begin{align}
     \IF(\beta_B)
     &= \frac{1}{\pi_{ac}}\left[ \IF(N_0^+)  + \IF(N_0^-) + \IF(N_{B,1}^+) + \IF(N_{B,1}^-) - \beta_B \IF(\pi_{ac})\right] \label{eq:IF_betaB_general1}
\end{align}

\begin{table}[!h]
\caption{Influence Functions of Components} \label{tab:IF-components}
\centering
{\footnotesize
\renewcommand{\arraystretch}{1.1}
\begin{tabular}{>{\raggedright\arraybackslash}p{0.12\textwidth}|
                >{\centering\arraybackslash}p{0.82\textwidth}}
\hline \\[-1.5ex]
Component & Influence Function $\IF$ \\ \hline \\[-0.5ex]

$\pi_{ac}$ &
$\begin{aligned}[t]
& -\mathbbm{1}^{+}(X)\Big[ W(Z,X)((1-D)S-m(Z,X)) + m(1,X)-m(0,X) \Big] \\
& + \mathbbm{1}^{-}(X)\Big[ W(Z,X)(DS-r(Z,X)) + r(1,X)-r(0,X) \Big] - \pi_{ac}
\end{aligned}$ \\[8ex]

$N_{0}^{+}$ &
$\mathbbm{1}^{+}(X)\Big\{ W(Z,X)((1-D)SY-\mu(Z,X)) + \mu(1,X)-\mu(0,X) \Big\}
- N_{0}^{+}$ \\[4ex]

$N_{L,1}^{+}$ &
$\begin{aligned}[t]
& \mathbbm{1}^{+}(X)\Big\{ 
- Q_{1}(p(X),X) W(Z,X)\Big[ (1-D)S-m(Z,X) \\
&\quad + (DS-r(Z,X))F(Q_{1}(p(X),X)|1,Z,X) \\
&\quad + DS(\mathbbm{1}(Y\le Q_{1}(p(X),X)) - F(Q_{1}(p(X),X)|1,Z,X)) \Big] \\
&\quad + W(Z,X)\Big[ DSY\mathbbm{1}(Y\le Q_{1}(p(X),X)) - \psi_L^{+}(Z,X)r(Z,X) \Big] \\
&\quad + \psi_L^{+}(1,X)r(1,X) - \psi_L^{+}(0,X)r(0,X)
\Big\}
- N_{L,1}^{+}
\end{aligned}$ \\[24ex]

$N_{U,1}^{+}$ &
$\begin{aligned}[t]
& \mathbbm{1}^{+}(X)\Big\{ 
- Q_{1}(1-p(X),X) W(Z,X)\Big[ (1-D)S-m(Z,X) \\
&\quad + (DS-r(Z,X))(1-F(Q_{1}(1-p(X),X)|1,Z,X)) \\
&\quad + DS(\mathbbm{1}(Y>Q_{1}(1-p(X),X)) - (1-F(Q_{1}(1-p(X),X)|1,Z,X))) \Big] \\
&\quad + W(Z,X)\Big[ DSY\mathbbm{1}(Y>Q_{1}(1-p(X),X)) - \psi_U^{+}(Z,X)r(Z,X) \Big] \\
&\quad + \psi_U^{+}(1,X)r(1,X) - \psi_U^{+}(0,X)r(0,X)
\Big\}
- N_{U,1}^{+}
\end{aligned}$ \\[24ex]

$N_{0}^{-}$ &
$\mathbbm{1}^{-}(X)\Big\{ W(Z,X)(DSY-\nu(Z,X)) + \nu(1,X)-\nu(0,X) \Big\}
- N_{0}^{-}$ \\[4ex]

$N_{L,1}^{-}$ &
$\begin{aligned}[t]
& \mathbbm{1}^{-}(X)\Big\{ 
- Q_{0}(1-\tfrac{1}{p(X)},X) W(Z,X)\Big[ DS-r(Z,X) \\
&\quad + ((1-D)S-m(Z,X))(1-F(Q_{0}(1-\tfrac{1}{p(X)},X)|0,Z,X)) \\
&\quad + (1-D)S(\mathbbm{1}(Y\ge Q_{0}(1-\tfrac{1}{p(X)},X)) - (1-F(Q_{0}(1-\tfrac{1}{p(X)},X)|0,Z,X))) \Big] \\
&\quad + W(Z,X)\Big[ (1-D)SY\mathbbm{1}(Y\ge Q_{0}(1-\tfrac{1}{p(X)},X)) - \psi_L^{-}(Z,X)m(Z,X) \Big] \\
&\quad + \psi_L^{-}(1,X)m(1,X) - \psi_L^{-}(0,X)m(0,X)
\Big\}
- N_{L,1}^{-}
\end{aligned}$ \\[24ex]

$N_{U,1}^{-}$ &
$\begin{aligned}[t]
& \mathbbm{1}^{-}(X)\Big\{ 
- Q_{0}(\tfrac{1}{p(X)},X) W(Z,X)\Big[ DS-r(Z,X) \\
&\quad + ((1-D)S-m(Z,X))F(Q_{0}(\tfrac{1}{p(X)},X)|0,Z,X) \\
&\quad + (1-D)S(\mathbbm{1}(Y\le Q_{0}(\tfrac{1}{p(X)},X)) - F(Q_{0}(\tfrac{1}{p(X)},X)|0,Z,X)) \Big] \\
&\quad + W(Z,X)\Big[ (1-D)SY\mathbbm{1}(Y\le Q_{0}(\tfrac{1}{p(X)},X)) - \psi_U^{-}(Z,X)m(Z,X) \Big] \\
&\quad + \psi_U^{-}(1,X)m(1,X) - \psi_U^{-}(0,X)m(0,X)
\Big\}
- N_{U,1}^{-}
\end{aligned}$ \\[24ex]

\hline
\end{tabular}
}
\end{table}

The influence functions of the components are in Table \ref{tab:IF-components}.
Putting upper and lower bounds together then yields \begin{align}
    \IF(\beta) = \begin{pmatrix}
        \IF(\beta_L) \\
        \IF(\beta_U)
    \end{pmatrix}
\end{align}
The influence functions $\IF(\beta_B)$ for $B \in \{L,U\}$ nest the existing literature on LATE and sample selection bounds. In particular, we obtain the following proposition: \begin{proposition} \label{prop:nesting-of-IF}
 For any $B \in \{L,U\}$, the efficient influence function in \eqref{eq:IF_betaB_general1} under (i) perfect compliance, (ii) no sample selection or (iii) both collapse to their respective efficient influence functions for Lee bounds, LATE, or ATE:  
\begin{enumerate}[label=(\roman*), leftmargin=*] 
    \item If $P(D = Z) = 1$, then $\IF(\beta_B) = \IF(\beta_B^{Lee})$ a.s. 
    \item If $P(S = 1) = 1$, then $\IF(\beta_B) = \IF(\theta_{LATE})$ a.s.
    \item If $P(S=1, D=Z) = 1$, then $\IF(\beta_B) = \IF(\theta_{ATE})$ a.s. 
\end{enumerate}
    
\end{proposition}
Practically, this means that, as $P(D = Z) \rightarrow 1$ or $P(S = 1) \rightarrow 1$, our influence functions will get closer in probability to the  semiparametrically efficient influence functions for Lee-bounds on the intensive margin ATE \citep{heiler2024treatmentevaluationintensiveextensive,semenova2025generalized} or the LATE \citep{frolich2007nonparametric} respectively. If both apply, the functions approach the efficient influence of the ATE \citep{hahn1998role}. 

\subsection{Finite Sample Implementation} \label{sec:estimation1}
We now outline the steps to estimate bounds and effect confidence intervals and provide more details regarding nuisance function estimation. For a generic random variable $X$ we denote $E_n[X] = \frac{1}{n}\sum_i^n X_i$ in what follows.  Assume we have independent data $O_i = (S_iY_i, S_i, D_i, Z_i, X_i)'$ for $i=1,\dots,n$. Algorithm \ref{alg:estimation} contains a step-by-step explanation of how to obtain bounds and inference. Estimation of bounds is fairly standard within the DML framework for composite ratio parameters. In particular, we estimate all their components separately with cross-fitted nuisances to obtain an estimate for the eventual bounds and their respective influence functions. The latter then yield the full variance-covariance matrix estimates that can be used for standard-normal-based inference on bounds or, more importantly, confidence intervals for the effect using refined critical values \citep{imbens2004confidence,stoye2020simple}.
All of these have at least $(1-\alpha)$ asymptotic coverage under some regularity conditions on the conditional distribution and suitable rate conditions on the nuisances, see Section \ref{sec:largesample1} for more details. 

\begin{algorithm}[!h]
  \caption{Estimation and Inference} \label{alg:estimation}
  \begin{algorithmic}[1]
   \Require Data $\{O_i\}_{i=1}^n$, number of folds $K$.
   \State Randomly partition data into $K$ disjoint folds $I_1,\dots,I_K$ of approximately equal size.
    \For{each $f \in \{1,\dots,K\}$}
    \State estimate primary nuisance parameters using data in $I_f^c$.
    \State evaluate the nuisances on $I_f$.
    \EndFor
    
    \For{each $N \in \{N_0^+,N_0^-,N_{L,1}^+,N_{L,1}^-,N_{U,1}^+,N_{U,1}^-,\pi_{ac}\}$}
    \State for all $i=1,\dots,n$ plug in all nuisances $\hat{\eta} $ into the orthogonal moments $\IF(\cdot)$.
      \State obtain $\hat{N}$ by setting the sample mean of the orthogonal moments equal to zero $$E_n[\IF(\hat{N},\hat{\eta})] = 0$$
    \EndFor
    \State
    \Return the bound estimators as $$\hat{\beta}_B = \frac{\hat{N}_0^+ + \hat{N}_0^- + \hat{N}_{B,1}^+ + \hat{N}_{B,1}^-}{\hat{\pi}_{ac}}$$
    \State   \Return standard errors via estimated composite moment: $$\hat{\sigma}_{B,n} = \sqrt{\frac{E_n[\IF(\hat{\beta}_B,\hat{\eta})^2]}{n}}$$
    where the composite moment is given by $$\IF(\hat{\beta}_B,\hat{\eta}) = \frac{1}{\hat{\pi}_{ac}}\left[ \IF(\hat{N}_0^+,\hat{\eta})  + \IF(\hat N_0^-,\hat{\eta}) + \IF(\hat N_{B,1}^+,\hat{\eta}) + \IF(\hat N_{B,1}^-,\hat{\eta}) - \hat\beta_B \IF(\hat\pi_{ac},\hat{\eta})\right]$$
    \State \Return $(1-\alpha)$ effect confidence intervals as 
    \begin{align*}
        CI_{1-\alpha}(\theta_{SLATE}) = \left[\hat\beta_L - c_{L,\alpha}\hat\sigma_{L,n},\  \hat\beta_U + c_{U,\alpha}\hat\sigma_{U,n}\right]
    \end{align*}
    where refined critical values $c_{L,\alpha},c_{U,\alpha} \leq z_{1-\alpha/2}$ can be chosen according to \cite{imbens2004confidence} or \cite{stoye2020simple}. Using $z_{1-\alpha/2}$ is valid for inference on the bounds.   
  \end{algorithmic}
\end{algorithm}

We now discuss primary nuisance function estimation and how to obtain derived nuisances as defined in Table \ref{tab:nuisance}. Our high-level assumptions in Section \ref{sec:largesample1} match these primitive objects that are conditional means/probabilities, conditional CDFs, and trimmed conditional means.
For primary nuisance functions that are simple conditional means or probabilities, $e,r,m,\mu, \nu$, a plethora of off-the-shelf nonparametric and machine learning methods such as neural networks, forests or high-dimensional sparse parametric models are available. 
The derived nuisances $\lambda_0,\lambda_1, \pi_{ac}, p, \mathbbm{1}^+, \mathbbm{1}^-, W$ only depend on the primary and can be obtained via simple plug-in versions.

The functional parameters $F$ and $G_B$ require special attention. In particular, we suggest to estimate the primary conditional CDF $F$ via machine learning analogues of distributional regression \citep{foresi1995conditional,klein2024distributional}. We also recommend direct imposition or post-processing, e.g., via isotonic regression \citep{henzi2021isotonic} that further refine these estimates by enforcing nondecreasing estimated conditional CDFs. Together with $r$ and $m$, these yield plug-in versions of $F_0$ and $F_1$. The $G_B$ components can be similarly obtained and refined as a sequence of regressions with outcome variable equal to trimming indicator times actual outcome.  

To obtain the derived conditional quantiles $Q_0$ and $Q_1$, inversion of the previously obtained conditional CDFs, $F_0$ and $F_1$, can be used. These yield the trimming indicators that are required to evaluate the conditional expectation models $G_B$ at their respective trimming quantiles to obtain derived nuisances $\psi_B$. 
The inversion-based approach ensures algebraic compatibility between the different nuisance quantities. Thus, we use this approach in all of our simulations and applications in Section \ref{sec:JC} as well as Appendix \ref{sec:simulation} and \ref{sec:OHIE}.

\subsection{Large Sample Properties} \label{sec:largesample1}

We now present additional technical assumptions and the resulting large sample properties of the efficient influence function based estimators of the causal effect bounds.
Denote the true nuisances as $\eta(x) =: \eta \in \mathcal{H}$ where $\mathcal{H}$ is a convex subset of a suitably normed vector space. Denote $\mathcal{H}_n \subset \mathcal{H}$ the realization set of the estimated nuisance quantities $\hat{\eta}(x) =: \hat{\eta}$, i.e.,~the set containing estimated nuisances with probability $1-u_n$ where $u_n = o(1)$. All nuisances are cross-fitted according to Definition 3.2 in \cite{chernozhukov2018double}, see Algorithm \ref{alg:estimation}.

For the remainder, write for generic nuisance $h=h(x)$ its estimation error $\Delta h=\hat h-h$. Denote $\|\cdot\|_2$ as $L^2(P)$ norm and $\|\cdot\|_\infty$ as the uniform norm. 
By abuse of notation, if the object depends on $Z=z$ and/or $D=d$, we suppress dependence and take all norms to be uniform over these finite dimension as well, e.g., \begin{align}
    ||r||_p = \sup_{z,d \in \{0,1\}}||r(z,x)||_p  =\sup_{z \in \{0,1\}}||r(z,x)||_p 
    = \sup_{z \in \{0,1\}} \left(\int r^p(x,z)dP(x)\right)^{1/p}
\end{align} 
and equivalently for other nuisances. For a given $x$, we also denote 
$\|\cdot\|_{\infty,\mathcal{N}_x}$ as the uniform norm over a neighborhood $\mathcal{N}_x$. In particular, for a generic object $A(\cdot|x)$,
\begin{align}
||\hat{A}(y|x) - A(y|x)||_{\infty,\mathcal{N}_x} &= \sup_{y \in \mathcal{N}_x}|\hat{A}(y|x) - A(y|x)|. 
\end{align}
We denote the supremum over these neighborhoods as \begin{align}
||\hat{A}(y|x) - A(y|x)||_{\infty,\mathcal{N}} &= \sup_x\sup_{y \in \mathcal{N}_x}|\hat{A}(y|x) - A(y|x)|.
\end{align} Moreover, we write shorthand \begin{align}
 ||\Delta G ||_p &= ||\Delta G_L||_{p,\mathcal{N}} + ||\Delta G_U||_{p,\mathcal{N}}  
\end{align} and $a_n \lesssim b_n$ and $a_n \lesssim_P b_n$, whenever $a_n = O(b_n)$ or $a_n = O_p(b_n)$ respectively. If not stated differently, the following assumptions are all uniformly over $n$.

\noindent {\textbf{Assumption A (Regularity, Overlap, and Learning Rates)}} 
\textit{
\begin{enumerate}[itemsep=0pt] \singlespacing
\item[A.1] (Moments) The conditional potential outcome moments are bounded, i.e.,~for some $m > 0$, \begin{align*}\sup_{x\in\mathcal{X},d\in\{0,1\}}E[|Y(d)|^{2+m}|X=x] \lesssim 1.\end{align*}
\item[A.2] (Eigenvalues) The variance-covariance matrix of the influence function $E[\IF(\beta)\IF(\beta)']$ has finite eigenvalues bounded away from zero.
\item[A.3] (No Point Mass at Trimming Points) The conditional distributions are continuous at the trimming points. For $z \in \{0,1\}$, and $x$ on the relevant support,\footnote{Here and in A.4, the relevant support is $\mathcal X^+_{ac}$ for conditions involving $(Q_1(p(x),x)$ or $Q_1(1-p(x),x)$, and $\mathcal X^-_{ac}$ for
conditions involving $Q_0(1-1/p(x),x)$ or $Q_0(1/p(x),x)$.}\begin{align*}
    &P(Y=Q_1(p(x),x)| DS=1,Z=z,X=x) = 0,\\
    &P(Y=Q_1(1-p(x),x)| DS=1,Z=z,X=x) = 0,\\ 
&P(Y=Q_0(1-1/p(x),x)| (1-D)S=1,Z=z,X=x) = 0, \\ 
&P(Y=Q_0(1/p(x),x)| (1-D)S=1,Z=z,X=x) = 0.
\end{align*}
\item[A.4] (Bounded Mixture Outcome Density and Local Lipschitz Trimmed Means) Let $f_1(\cdot|x)$ and $f_0(\cdot|x)$ denote the respective densities of $F_1(\cdot|x)$ and $F_0(\cdot|x)$. Assume they are bounded at the trimming thresholds and the trimmed conditional means are locally Lipschitz, i.e.,~there exist a $C>0$ and constants $0<f_{\min}\le f_{\max}<\infty$, $L_G<\infty$ such that (i) for all $x$ on the relevant support:
\begin{align*}
&f_1(Q_1(p(x),x)\mid x) \in [f_{\min},f_{\max}], \\
&f_1(Q_1(1-p(x),x)\mid x) \in [f_{\min},f_{\max}], \\
&f_0(Q_0(1-1/p(x),x)\mid x) \in [f_{\min},f_{\max}], \\
&f_0(Q_0(1/p(x),x)\mid x) \in [f_{\min},f_{\max}],
\end{align*}
and (ii) uniformly in $(z,x)$, for $|u|\le C$,
\begin{align*}
&|G_L(Q_1(p(x),x) + u|1,z,x) - G_L(Q_1(p(x),x)|1,z,x)| \leq L_G |u|, \\
&|G_U(Q_1(1-p(x),x) + u|1,z,x) - G_U(Q_1(1-p(x),x)|1,z,x)| \leq L_G |u|, \\
&|G_L(Q_0(1-1/p(x),x) + u|0,z,x) - G_L(Q_0(1-1/p(x),x)|0,z,x)| \leq L_G |u|, \\
&|G_U(Q_0(1/p(x),x) + u|0,z,x) - G_U(Q_0(1/p(x),x)|0,z,x)| \leq L_G |u|. 
\end{align*}
\item[A.5] (Margin Condition) The distribution of the positive and negative monotonicity type is well-behaved around the margin of indifference, i.e.,~there exist $C_M<\infty$ and $\kappa>0$ such that
\[
P\!\left(\big|-[m(1,X)-m(0,X)]-[r(1,X)-r(0,X)]\big|\le t\right)\le C_M\,t^\kappa\quad\text{for all }t>0.
\]
\item[A.6] (Strong Overlap) There are comparable units across instrument levels and within the differently treated and selected populations. Moreover, always-selected complier probabilities are bounded away from
zero on their relevant supports, i.e.,~there exists some $\underline{c} \in (0,1)$ such that \begin{align*}
    \underline{c} < \inf_{z,x}\{r(z,x),m(z,x),e(x)\} \leq \sup_{z,x}\{r(z,x),m(z,x),e(x)\} < 1-\underline{c}.
\end{align*}
and
 \begin{align*}
\inf_{x\in\mathcal X_{ac}^+}-[m(1,x) - m(0,x)]> \underline{c},\qquad
\inf_{x\in\mathcal X_{ac}^-}[r(1,x) - r(0,x)]> \underline{c}.
\end{align*}
\item[A.7] (Machine Learning Bias) Let $u_n = o(1)$. For all folds, the nuisance parameters obtained via cross-fitting belong to a shrinking neighborhood $\mathcal{H}_n$ around $\eta$ with probability of at least $1-u_n$, such that, uniformly over the neighborhood, the nuisance functions are consistent
\begin{align*}
    \|\Delta\mu\|_2\ + \|\Delta\nu\|_2\ + \|\Delta e\|_2\ + \|\Delta m\|_{\infty}\ + \|\Delta r\|_{\infty} + ||\Delta G||_{\infty,\mathcal{N}} + ||\Delta F ||_{\infty,\mathcal{N}}  &= o(1), 
\end{align*}
and obey convergence rates \begin{align*}
 &(\|\Delta\mu\|_2\ + \|\Delta\nu\|_2) \|\Delta e\|_2\ + (||\Delta m||_{\infty} + ||\Delta r||_{\infty})^{\kappa + 1}  +(||\Delta e ||_2 + ||\Delta m ||_2 + ||\Delta r ||_2) \times \\ &\bigg[||\Delta e ||_{\infty} + ||\Delta m ||_{\infty} + ||\Delta r ||_{\infty} + ||\Delta G||_{\infty,\mathcal{N}} + ||\Delta F ||_{\infty,\mathcal{N}} \bigg] = o(n^{-1/2}).
\end{align*}
\end{enumerate}}

Assumption A.1 and A.2 are simple regularity conditions that rule out heavy tails and degenerate DGPs where bounds are close or equal to a point or otherwise degenerate.

Assumption A.3 rules out point masses at the trimming thresholds in the observed selected outcome distributions used to construct the reduced-form nuisance functions. This ensures that the trimming indicators are unambiguous at the cutoff and that weak and strict inequality conventions coincide at the relevant thresholds.\footnote{The theory can be extended to mass points as discussed in Section \ref{sec:bounds}. However, in contrast to identification where mass points are harmless once one uses fractional trimming \citep{huber2015sharp,kitagawa2021theidentificationregion}, inference can be affected as the functional may be nonregular. This problem arises only in boundary cases where the trimming probability coincides exactly with the edge of a mass point. For DGPs in which the cutoff lies in the interior of a mass point, fractional trimming yields a regular functional, so standard root-$n$ semiparametric inference can proceed using the corresponding influence function.}

Assumption A.4(i) is complementary to A.3 and provides primitive density regularity that guarantees local invertibility and differentiability. It controls error propagation through the quantile mapping (Bahadur-type representation). A.4(ii) adds local smoothness to the trimmed mean functions, ensuring uniform control over the nuisance functions evaluated close to the trimming points. 

Assumption A.5 is a margin condition that controls the mass of observations near the boundary between positive and negative selection monotonicity types among compliers. In particular, it rules out excessive concentration of probability that would make correct classification difficult. In the case of a bounded density, A.5 holds with $\kappa = 1$, see, e.g.,~\cite{audibert2007fast} or \cite{heiler2024treatmentevaluationintensiveextensive} for related assumptions and discussion. The larger $\kappa$, the less demanding the convergence requirements in A.7 for learning the conditional joint probability of being selected and in treatment/control status. It implies the necessary regularity condition $P(\mathcal{X}^0) = 0$.

Assumption A.6 assures that there are comparable units for instrument and the treatment within the selected group and a relevant share of always-selected compliers. Strong overlap is imposed to obtain a finite variance bound and avoid irregular identification \citep{khan2010irregular,HEILER2021valid}.

Assumption A.7 imposes the standard DML requirement that all nuisances are consistent and converge to their truth sufficiently fast for the second-order remainder of the orthogonal von Mises expansion to be \(o(n^{-1/2})\), but it does so in a relatively weak form tailored to our target parameter and estimation procedure and its primitives.  In particular, in contrast to much of the Lee bounds-type literature, we do not assume global uniform convergence rates for the estimated quantiles or trimmed mean functionals. Instead, it only imposes local sup-norm control of the estimated CDFs 
\(F\) and trimmed means \(G_B\) on neighborhoods relevant for trimming, with the regularity of quantiles and trimmed means derived from these local conditions. The first product term in A.7 matches the usual \(L^2\)-rate requirement in \cite{chernozhukov2018double}, while the additional terms capture the non-smooth features of monotonicity type classification and trimming indicators. This is related to rate conditions in \cite{heiler2024heterogeneous} and \cite{semenova2025generalized} for generalized Lee bounds as well as \cite{heiler2024treatmentevaluationintensiveextensive} for more general intensive and extensive margin treatment effects, but expressed directly in terms of CDF errors 
rather than through separate uniform convergence assumptions on the associated quantile-based 
functionals. Examples for global uniform rates of nonparametric and machine learning estimation of conditional CDFs can be found in, e.g., \cite{xie2023uniform} or \cite{cattaneo2025uniformestimationinferencenonparametric} respectively. 

We obtain the following Theorem: \begin{theorem} \label{thm_asyN1}
    Under Assumptions \ref{ass:conditional_IV}, \ref{ass:conditional_monotoneS}, and A.1--A.7, the estimated bounds are jointly asymptotically normal and semiparametrically efficient, i.e.\begin{align*}
    \sqrt{n}(\hat{\beta} - \beta) \overset{d}{\rightarrow} \mathcal{N}\big(0,E[\IF(\beta)\IF(\beta)']\big).
\end{align*}
\end{theorem}
Theorem \ref{thm_asyN1} can be used directly for inference on the bounds using the usual standard normal critical values. Importantly, the assumptions imply that the identified set always has a non-empty interior in the population. Thus, Theorem \ref{thm_asyN1} is sufficient to construct tighter confidence intervals for the effect $\theta_{SLATE}$ directly using \cite{imbens2004confidence} or \cite{stoye2020simple} critical values. We use the latter in our empirical applications.

\section{Always-selected Complier Profiling} \label{sec:profiling}
 $\theta_{SLATE}$ is the average treatment effect for the target population $ac$, the always-selected compliers. While individual members of this group are not identified and the group may differ from observed subpopulations or the overall population in both observed and unobserved characteristics, its observable features can nevertheless be characterized, much like complier profiling for the LATE.\footnote{See, for example, \cite{abadie2003semiparametric}, \cite{angrist2004treatmenteffectheterogeneity}, and \cite{singh2023doublerobustnessforcomplier} for various approaches to complier profiling.} This makes it possible to compare always-selected compliers with other populations of interest and thereby better assess the external validity and substantive importance of the effect bounds.

Observable $ac$ characteristics can be identified using the same ingredient, $\pi_{ac}(x)$, as the SLATE bounds in Section \ref{sec:bounds_covariates}. Debiased estimation and inference follow analogously from our influence functions. In particular, for any integrable $g:\mathcal{X}\rightarrow \mathbb{R}$, consider target parameter \begin{align}
    E[g(X)|ac] = \frac{E[g(X)\pi_{ac}(X)]}{E[\pi_{ac}(X)]}.
\end{align}
This is the average $g(x)$ in the population of always-selected compliers. Its influence function is given by 
\begin{align}
    \IF(E[g(X)|ac]) &= \frac{1}{\pi_{ac}}\big(\left(\IF(\pi_{ac}) + \pi_{ac}\right)\left(g(X) - E[g(X)|ac]\right)\big), \label{eq:IFg(X)}
\end{align}
where nuisances and influence functions of the components can be found in Section \ref{sec:moment-functions1}. The corresponding estimator can be obtained by solving the empirical analogue as in Algorithm \ref{alg:estimation}.  
It is root-$n$ consistent and asymptotically normal analogously to Theorem \ref{thm_asyN1}. Moreover, influence function \eqref{eq:IFg(X)} can directly be used for statistically valid comparisons between mean covariates $g(x)$ of $ac$ and other populations such as unconditional, treated or assigned units. We provide some specific examples in Section \ref{sec:JC} and Appendix \ref{sec:OHIE}.

\section{Empirical Study I: Job Corps Revisited}\label{sec:JC}
\subsection{Data and Methods}
In this section, we re-evaluate the earnings effect of participating in JC, a large US federally funded training program providing free academic education, vocational training and employment assistance to disadvantaged youth. We make use of the National Job Corps Study by Mathematica Policy Research. This experiment implemented stratified randomized assignment of applicants, incorporating over 15,400 individuals between ages 16 and 24. Multiple outcomes such as earnings and job status were gathered at various points after assignment.

Our data and main variables are identical to \cite{chen2015bounds}. In particular, we use a subset of 9,090 units (3,599 control and 5,491 treated) with non-missing work hours, earnings, and participation information. The outcome is log hourly wages which is only observed for the employed. Assignment is given by the original randomization. Treatment is defined as eventual JC participation within the 208 weeks of the evaluation period. We additionally make use of 
socio-economic pre-assignment covariates including job and earnings history, education, parental background and more, matching the variables used in \cite{lee2009training} for ITT bounds. The list of covariates along with sample summary statistics are provided in Table \ref{tab:jc_balance}.

The data are suitable for our method: Assignment was based on stratified randomization justifying conditional independence. Exclusion is credible as any earnings effects likely require actual training and not just assignment. Importantly, there was a significant amount of noncompliance with the randomized treatment assignment. In particular, only 73.8\% of individuals assigned to the treatment group ever participated in JC. Additionally, a small fraction (4.4\%) of individuals assigned to the control group also ended up participating.\footnote{Among the 4.4\%, 1.2\% of controls enrolled in JC before the end of the
embargo, while 3.2\% enrolled afterward.} 

We evaluate the always-selected complier effect $\theta_{SLATE}$ at week $208$ after assignment using (i) \cite{chen2015bounds} bounds and (ii) sharp-basic bounds -- both assuming strong sample selection monotonicity -- as well as (iii) sharp DML bounds with covariates under weak monotonicity. Implementation of (i) follows \cite{chen2015bounds}. (ii) uses simple sample analogues of \eqref{eq:betaL} and \eqref{eq:betaU}. (iii) is obtained via the procedure in Section \ref{sec:estimation1}, see Appendix \ref{app:supp-empirical} for additional details. We also provide a profiling analysis of always-selected compliers using the method from Section \ref{sec:profiling}.

\subsection{Results: Bounds and Shares}
\begin{table}[!htbp]
\caption{Bounds for the Effect of Job Corps on Log Wages}
\label{tab:jc_bounds2}\centering
{\footnotesize
\begin{tabular}{@{}lccc}
\toprule & CF & Sharp-basic & Sharp-DML \\
\midrule Estimates & $[-0.022,\ 0.130]$ & $[0.019,\ 0.067]$ & $[-0.013,\ 0.124]$
\\
Standard Errors & --- & $(0.023,\ 0.022)$ & $(0.023,\ 0.022)$ \\
95\% Confidence Interval & $[-0.061,\ 0.168]$ & $[-0.018,\ 0.104]$ & $[-0.051,\ 0.160]$ \\[1ex]
Share of positive sample selection $^{a}$ & \multicolumn{2}{c}{$1 \text{ (assumed)}$} & \multicolumn{1}{l}{$0.933\ (0.004)$}  \\
Share of $ac$ $^{b}$ & \multicolumn{2}{c}{$0.391\ (0.010)$} & \multicolumn{1}{l}{$0.393\ (0.010)$} \\
\bottomrule &  &  &
\end{tabular}
}
{\footnotesize 
\begin{minipage}{0.88\linewidth}
\footnotesize
\textit{Notes:} $N=9,090$. All calculations use design weights.
\textbf{CF} refers to the \cite{chen2015bounds} bounds under strong sample selection monotonicity without covariates using half-median-unbiased estimates with 95\% confidence intervals following \cite{chernozhukov2013intersection}.
CF reports no standard errors as inference relies on the CLR projection.
\textbf{Sharp-basic} refers to the sharp bounds without covariates under strong sample selection monotonicity.
\textbf{Sharp-DML} refers to the sharp bounds estimated by DML under weak sample selection monotonicity.
Standard errors for the two bounds are in parentheses as
$(\widehat{SE}_{\mathrm{lower}},\,\widehat{SE}_{\mathrm{upper}})$,
and the 95\% CI for $\theta_{SLATE} = E[Y_{1}-Y_{0}\mid ac]$ is calculated using \cite{stoye2020simple}.
\\
$^{a}$ Estimated share of positive sample selection~$= E[\mathbb{I}^{+}(X)]$, the fraction of observations
in the positive sample selection class (Sharp-DML only, CF and Sharp-basic assume this fraction
equals~1).\\
$^{b}$ Estimated share of always-selected compliers $\pi_{ac}$. 
\end{minipage}
}
\end{table}

Table \ref{tab:jc_bounds2} reports the evaluation results. It delivers four main messages. First, the target population is empirically
relevant: the estimated share of always-selected compliers is about 39\% across specifications, so $\theta_{SLATE}$ pertains to a sizable subpopulation. Second, the data do not support imposing strong sample selection monotonicity: the estimated share in the positive sample selection region is 93.3\% and
strong sample selection monotonicity is rejected ($p<0.01$).\footnote{Negative employment effects at week 208 are economically plausible because
employment responses to Job Corps are heterogeneous across applicants. Consistent with this, \cite{semenova2025generalized} shows that positive conditional employment effects do not arise for all applicants
even four years after random assignment and documents subgroups with significantly
negative employment effects at later horizons.} Thus,
Sharp-DML appears to be the most credible specification.
Third, under the strong monotonicity benchmark without covariates, the sharp bounds and effect confidence interval are substantially tighter than CF, by 68.4\% and 46.7\%, respectively. Moreover, the estimated basic bounds no longer include zero and the corresponding confidence interval rules out negative effects below  -1.8\%. 
Fourth, estimates using Sharp-DML are broadly comparable to CF despite relying only on weak sample selection monotonicity. Both rule out negative effects beyond -5.1\% and -6.1\% respectively. However, despite imposing less restrictive assumptions, Sharp-DML delivers tighter bounds and confidence intervals than CF by 9.8\% and 7.9\% respectively, highlighting the relevance of both sharpness and efficient inclusion of covariates.


The results also speak to the importance of the no-scaling result in Proposition \ref{prop:Lee_not_scaling}. In particular, naively scaling basic ITT bounds $[-0.019,0.093]$ \citep{lee2009training} with the JC compliance probability of 69.4\% would suggest SLATE bounds of $[-0.027,  0.134]$ which vastly exceed our Sharp-basic bounds of $[0.019, 0.067]$. 


\subsection{Always-selected Complier Profiling}
\begin{table}[!h]
\centering
\caption{Job Corps Covariates: Always-Selected Compliers vs.\ Full Sample}
\label{tab:jc_profiling}
\footnotesize
\begin{tabular}{@{}l l l l@{}}
\toprule
Covariate & $ac$ & Full sample & ~\quad ~Difference \\
\midrule
Female                        & $0.400$~$(0.014)$ & $0.443$~$(0.005)$ & \diffcell{-0}{044}{~$(0.013)^{***}$} \\
Age (in yrs.) at baseline     & $18.45$~$(0.058)$ & $18.44$~$(0.023)$ & \diffcell{0}{013}{~$(0.053)$} \\
Black, non-Hispanic           & $0.433$~$(0.014)$ & $0.500$~$(0.005)$ & \diffcell{-0}{067}{~$(0.013)^{***}$} \\
Hispanic                      & $0.194$~$(0.010)$ & $0.172$~$(0.004)$ & \diffcell{0}{022}{~$(0.009)^{**}$} \\
Other race/ethnicity          & $0.067$~$(0.007)$ & $0.071$~$(0.003)$ & \diffcell{-0}{004}{~$(0.007)$} \\
Married                       & $0.015$~$(0.004)$ & $0.022$~$(0.002)$ & \diffcell{-0}{007}{~$(0.004)^{*}$} \\
Living together               & $0.033$~$(0.006)$ & $0.041$~$(0.002)$ & \diffcell{-0}{008}{~$(0.005)$} \\
Separated                     & $0.019$~$(0.004)$ & $0.023$~$(0.002)$ & \diffcell{-0}{004}{~$(0.004)$} \\
Has children                  & $0.162$~$(0.011)$ & $0.204$~$(0.004)$ & \diffcell{-0}{042}{~$(0.010)^{***}$} \\
Number of children            & $0.214$~$(0.017)$ & $0.291$~$(0.007)$ & \diffcell{-0}{077}{~$(0.017)^{***}$} \\
Education (in yrs.)           & $10.23$~$(0.042)$ & $10.12$~$(0.016)$ & \diffcell{0}{111}{~$(0.039)^{***}$} \\
Mother's education            & $11.54$~$(0.064)$ & $11.48$~$(0.024)$ & \diffcell{0}{061}{~$(0.059)$} \\
Father's education            & $11.54$~$(0.060)$ & $11.45$~$(0.023)$ & \diffcell{0}{096}{~$(0.054)^{*}$} \\
Ever arrested                 & $0.235$~$(0.011)$ & $0.248$~$(0.005)$ & \diffcell{-0}{013}{~$(0.011)$} \\

Household income: & & & \\
\quad\quad\quad\quad$[\$ 3,000, \$ 6,000)$  & $0.189$~$(0.009)$ & $0.208$~$(0.003)$ & \diffcell{-0}{019}{~$(0.008)^{**}$} \\
\quad\quad\quad\quad$[\$ 6,000, \$ 9,000)$  & $0.122$~$(0.007)$ & $0.116$~$(0.003)$ & \diffcell{0}{006}{~$(0.006)$} \\
\quad\quad\quad\quad$[\$ 9,000, \$ 18,000)$ & $0.258$~$(0.010)$ & $0.244$~$(0.004)$ & \diffcell{0}{014}{~$(0.009)$} \\
\quad\quad\quad\quad$\geq \$ 18,000$        & $0.198$~$(0.009)$ & $0.179$~$(0.003)$ & \diffcell{0}{019}{~$(0.008)^{**}$} \\

Personal income: & & & \\
\quad\quad\quad\quad$[\$ 3,000, \$ 6,000)$  & $0.130$~$(0.009)$ & $0.131$~$(0.003)$ & \diffcell{-0}{001}{~$(0.008)$} \\
\quad\quad\quad\quad$[\$ 6,000, \$ 9,000)$  & $0.066$~$(0.006)$ & $0.051$~$(0.002)$ & \diffcell{0}{015}{~$(0.005)^{***}$} \\
\quad\quad\quad\quad$\geq \$ 9,000$         & $0.038$~$(0.005)$ & $0.033$~$(0.002)$ & \diffcell{0}{004}{~$(0.005)$} \\

At baseline: & & & \\
\quad\quad\quad\quad Has job                          & $0.242$~$(0.011)$ & $0.195$~$(0.004)$ & \diffcell{0}{047}{~$(0.010)^{***}$} \\
\quad\quad\quad\quad Months worked, previous year     & $4.377$~$(0.119)$ & $3.566$~$(0.045)$ & \diffcell{0}{810}{~$(0.109)^{***}$} \\
\quad\quad\quad\quad Had a job, previous year         & $0.716$~$(0.013)$ & $0.629$~$(0.005)$ & \diffcell{0}{087}{~$(0.012)^{***}$} \\
\quad\quad\quad\quad Earnings, previous year          & $3470$~$(137.0)$  & $2870$~$(58.00)$  & \multicolumn{1}{r}{$600.0~(111.0)^{***}$}  \\ 
\quad\quad\quad\quad Weekly hours, most recent job    & $24.35$~$(0.560)$ & $21.42$~$(0.220)$ & \diffcell{2}{930}{~$(0.520)^{***}$} \\
\quad\quad\quad\quad Weekly earnings, most recent job & $121.8$~$(5.800)$ & $107.6$~$(2.900)$ & \multicolumn{1}{r}{$14.10~(3.800)^{***}$} \\ 
\bottomrule
\end{tabular}
\par\medskip
\begin{minipage}{0.95\linewidth}
\footnotesize
\textit{Notes:} $N=9,090$. The table reports estimated means for the always-selected complier ($ac$) subpopulation and for the full sample, along with their difference. Standard errors are in parentheses. All estimates use design weights and are based on the weak-monotonicity DML specification with GRF learners. Joint $\chi^{2}$ tests for the null that all differences within a category are jointly zero:
Race/ethnicity ($p < 0.001$);
Marital status ($p = 0.084$);
Household income ($p = 0.011$);
Personal income ($p = 0.027$).
$^{*}$~$p<0.1$; $^{**}$~$p<0.05$; $^{***}$~$p<0.01$.
\end{minipage}
\end{table}

We now conduct the $ac$ profiling analysis as discussed in Section \ref{sec:profiling}.
Table \ref{tab:jc_profiling} contains the baseline characteristics of always-selected compliers, $ac$, and those of the full sample. Demographic differences are modest but systematic: On average, $ac$ individuals are less likely to be female (-4.4 pp) and Black (-6.7 pp), slightly more likely to be Hispanic (+2.2 pp) and have fewer children (-4.2 pp in incidence; -0.08 in number). Own education is marginally higher (+0.11 years), with parental education and age being similar across groups.  $ac$ individuals are slightly less concentrated in the lowest household income bracket and more represented in the highest (1.9 pp each). The most pronounced differences arise in pre-assignment labor market outcomes: $ac$ individuals are significantly more likely to be employed at baseline (+4.7 pp), more likely to have worked in the previous year (+8.7 pp), and exhibit higher labor supply and earnings, including +0.81 months worked, +\$600 annual earnings, +2.9 weekly hours, and +\$14 weekly earnings.

Taken together, these differences indicate that $ac$ subpopulation is positively selected, in particular on baseline labor-market attachment. The stronger pre-program employment and earnings profiles suggest that $ac$ may be better positioned to translate JC participation into earnings gains, but also imply more favorable counterfactual trajectories in the absence of treatment. At the same time, prior evidence on JC shows that impacts operate through channels such as increased GED, increased vocational credential attainment and reduced criminal involvement, which are not confined to the most labor-market-ready participants \citep{schochet2008does,flores2012estimating}. Thus, while the observable composition of the $ac$ group points to potential attenuation when extrapolating our estimates for $\theta_{SLATE}$ to the full sample, baseline differences alone do not fully pin down the direction or magnitude of the extrapolation.

\section{Concluding Remarks}\label{sec:conclusion}


For both identification and semiparametric inference, this paper provides a synthesis of treatment evaluation under sample selection as well as noncompliance building on Lee-type bounds and the LATE framework. 
Given the analytic form of our population bounds, it is 
relatively simple to incorporate additional restrictions, such as stochastic dominance of $ac$ potential outcome distribution over that of $cc$ to further tighten bounds \citep{zhang2003estimation,huber2015sharp,heiler2024treatmentevaluationintensiveextensive}. Moreover, given the simple ratio-of-linear-moment structure, our influence function components can be directly leveraged for estimation and inference on heterogeneous group-specific $ac$ effect bounds by combining them with nonparametric projection methods as in \cite{heiler2024heterogeneous}. 

\newpage

\addcontentsline{toc}{section}{References}	
	 {\setstretch{1}
	\bibliography{DH1}
	}

\subsection*{Declaration of generative AI and AI-assisted technologies in the manuscript preparation process}	

During the preparation of this work, the authors used OpenAI's ChatGPT to assist with mathematical proofs, coding, formatting of tables and outputs, and spelling. The authors reviewed and edited the content as needed and take full responsibility for the content of the article. 

\newpage

\appendix
{\huge \centering \textbf{Supplementary Appendix}}
\section{Supplementary Material for Section \ref{sec:bounds}}

\subsection{\textbf{Proof of Proposition \ref{prop:sharp_bounds}}}

The proof has two steps. Let $\beta_1:=E\left[ Y_1|ac\right]$.

Step 1, we\ show $\beta _{L,1}\leq \beta _{1}\leq \beta _{U,1}$, i.e., $\beta
_{L,1}$ is the smallest feasible value and $\beta _{U,1}$ is the largest
feasible value for $\beta _{11}$ that is consistent with the observed data and maintained assumptions.
It suffices to show $\beta _{1}\geq \beta _{L,1}$, and the other part $\beta
_{1}\leq \beta _{U,1}$ can be proved similarly.
We have shown that under Assumptions \ref{ass:IV}--\ref{ass:monotoneS} and assuming $\pi_{ac}>0$, $F_{Y_{1}|ac\cup cc}\left( y\right) $ is
uniquely identified. In particular, 

{\footnotesize\begin{align*}
F_{Y_{1}|ac \cup cc}\left( y\right) 
&=\frac{E\left[ DS\cdot \mathbbm{1}\left\{ Y\leq y\right\} |Z=1\right] -E\left[
DS\cdot \mathbbm{1}\left\{ Y\leq y\right\} |Z=0\right] }{E\left[ DS|Z=1\right] -E\left[
DS|Z=0\right] }.
\end{align*}}
Further, 

{\footnotesize\begin{align*}
F_{Y_{1}|ac\cup cc}\left( y\right) =pF_{Y_{1}|ac}\left( y\right) +\left(
1-p\right) F_{Y_{1}|cc}\left( y\right)
\end{align*}}%
where $p$ is uniquely identified:%

{\footnotesize\begin{align*}
p=\frac{E\left[ \left( D-1\right) S|Z=1\right] -E\left[ \left( D-1\right)
S|Z=0\right] }{E\left[ DS|Z=1\right] -E\left[ DS|Z=0\right] }.
\end{align*}}

If $\pi_{cc}=0$, then $p=1$ and $ac\cup cc = ac$, so $E[Y_1\mid ac]$ is point-identified and $\beta_{L,1}=\beta_{U,1}=E[Y_1\mid ac]$. Hence, in what follows, suppose $0<p<1$.

$\beta _{L,1}$ assumes 

{\footnotesize\begin{align*}
F_{Y_{1}|ac}\left( y\right) =F_{Y_{1}|ac}^{dh}\left( y\right) :=%
\begin{cases}
\frac{1}{p}F_{Y_{1}|ac \cup cc}\left( y\right) &\text{ if }y<Q_{1}\left( p\right)
\\ 
1 &\text{ if }y\geq Q_{1}\left( p\right)%
\end{cases}
\end{align*}}%
and hence%

{\footnotesize\begin{align*}
\beta _{L,1}=\int_{-\infty }^{\infty }yF_{Y_{1}|ac}^{dh}\left( dy\right) =%
\frac{1}{p}\int_{-\infty }^{Q_{1}\left( p\right) }yF_{Y_{1}|ac \cup cc}\left(
dy\right) .
\end{align*}}

By \cite{horowitz1995identification}, Corollary 4.1, $\beta _{1}\geq \beta _{L,1}$.
Further, $\beta _{L,1}$ is smallest feasible value that is consistent with
the observed data because both $F_{Y_{1}|ac\cup cc}\left( y\right) $ and $p$ are
uniquely determined by the observed data under the maintained assumptions.

Step 2, we show that every point in $\left[ \beta _{L,1},\beta _{U,1}\right] 
$ is compatible with our assumptions and the observed data and hence cannot
be ruled out, which implies that $\left[ \beta _{L,1},\beta _{U,1}\right] $
is contained in any other valid bounds that impose the same assumptions.

Note that $\beta _{L,1}\leq E\left[ Y_{1}|ac \cup cc\right] \leq E\left[
Y_{1}|ac \cup cc,Y_{1}>Q_{1}\left( p\right) \right] $, where both conditional
mean functions are uniquely identified, as the distribution
function $F_{Y_{1}|ac \cup cc} $ and the quantile level $p$ are
uniquely identified.

For any point $\delta $ between $\beta _{L,1}$ and $E\left[ Y_{1}|ac \cup cc\right] $, there exists a value $\lambda \in \left[ 0,1\right] $ such that 

{\footnotesize\begin{align*}
    \delta =\lambda \beta _{L,1}+\left( 1-\lambda \right) E\left[ Y_{1}|ac \cup cc,Y_{1}>Q_{1}\left( p\right) %
\right] 
\end{align*}}
Let 

{\footnotesize\begin{align*}
h_l\left( y\right) =
\begin{cases}
\frac{1}{p}F_{Y_{1}|ac \cup cc}\left( y\right) &\text{ if }y<Q_{1}\left( p\right)
\\ 
1&\text{ if }y\geq Q_{1}\left( p\right)%
\end{cases}%
\end{align*}}%
and 

{\footnotesize\begin{align*}
h_u\left( y\right) =
\begin{cases}
0&\text{ if }y<Q_{1}\left( p\right) \\ 
\frac{1}{1-p}\left( F_{Y_{1}|ac \cup cc}\left( y\right) -p\right) &\text{ if }%
y\geq Q_{1}\left( p\right)%
\end{cases}%
\end{align*}}%
We can construct distribution functions for $Y_{1}|ac$ and $Y_{1}|cc$,
respectively, as 

{\footnotesize\begin{align*}
F_{Y_{1}|ac}\left( y\right) &= \lambda h_l\left( y\right) +\left( 1-\lambda
\right) h_u\left( y\right) , \\
F_{Y_{1}|cc}\left( y\right) &= \left( \frac{p}{1-p}-\frac{p}{1-p}\lambda
\right) h_l\left( y\right) +\left( 1-\frac{p}{1-p}\left( 1-\lambda \right)
\right) h_u\left( y\right) .
\end{align*}}%
Note

{\footnotesize
\[
E[Y_1\mid ac\cup cc]
=
p\,\beta_{L,1}
+
(1-p)\,E[Y_1\mid ac\cup cc,\;Y_1>Q_1(p)],
\]
}
and \(\delta\in[\beta_{L,1},\,E[Y_1\mid ac\cup cc]]\), so the corresponding
\(\lambda\) must satisfy \(\lambda\in[p,1]\). Hence,

{\footnotesize
\[
0\le \frac{p}{1-p}(1-\lambda)\le 1.
\]
}
The coefficients used below in the definition of \(F_{Y_1\mid cc}\) are
then nonnegative and sum to one, meaning that the constructed \(F_{Y_1\mid ac}\) and
\(F_{Y_1\mid cc}\) are valid distribution functions. 
By construction, 

{\footnotesize\begin{align*}
pF_{Y_{1}|ac}\left( y\right) +\left( 1-p\right) F_{Y_{1}|cc}\left( y\right)
&= ph_l\left( y\right) +\left( 1-p\right) h_u\left( y\right) \\
&= F_{Y_{1}|ac \cup cc}\left( y\right) .
\end{align*}}%
i.e., the mixture of these two distribution functions replicates the
uniquely identified distribution function $F_{Y_{1}|ac \cup cc}\left( y\right) $%
and hence they are compatible with our assumptions and the observed data.
Further by construction, 

{\footnotesize \begin{align*}
E\left[ Y_{1}|ac\right] =\int_{-\infty }^{+\infty }yF_{Y_{1}|ac}\left(
dy\right) 
&=\int_{-\infty }^{+\infty }y\left\{ \lambda h_l\left( dy\right) +\left(
1-\lambda \right) h_u\left( dy\right) \right\} \\
&=\lambda \int_{-\infty }^{+\infty }yh_l\left( dy\right) +\left( 1-\lambda
\right) \int_{-\infty }^{+\infty }yh_u\left( dy\right) \\
&=\lambda \beta _{L,1}+\left( 1-\lambda \right) E\left[
Y_{1}|ac \cup cc,Y_{1}>Q_{1}\left( p\right) \right] \\
&=\delta .
\end{align*}}%
That is, $\delta $ is a feasible point that is compatible with our
assumptions and the observed data.

A symmetric argument can be made about any point $\delta $ in between $E%
\left[ Y_{1}|ac \cup cc\right] $ and $\beta _{U,1}$. Therefore, any point within
the interval $\left[ \beta _{L,1},\beta _{U,1}\right] $ is feasible and
compatible with our assumptions and observed data.

Together, the above two results, 1) $\beta _{L,1}\leq \beta_1
\leq \beta _{U,1}$, and 2) any point within the interval $\left[ \beta
_{L,1},\beta _{U,1}\right] $ is feasible and compatible with our assumptions
and observed data, suggest $\beta _{L,1}$ ($\beta _{U,1}$)\ is sharp, i.e., $%
\beta _{L,1}$ ($\beta _{U,1}$)\ is the largest (smallest)\ lower\ (upper)\
bound for $\beta_1\text{}:=E\left[ Y_{1}|ac\right] $ that is consistent with our assumptions and observed data.

\subsection{Supplementary Material for Section \protect\ref%
{subsec:Lee_comparison}}
\subsubsection{\cite{lee2009training} ITT
bounds and SLATE bounds}
\label{app:Lee_decomposition}

Below we formalize the relationship between the \cite{lee2009training}
bounds and our SLATE bounds. Recall the complier types defined in Section~\ref{subsec:baseline_bounds},
and define in addition 

{\footnotesize\begin{align*}
an:=\{S_{0}=1,\,D_{0}=D_{1}=0\}\quad \text{(always-selected
never-takers)},
\end{align*}}%

{\footnotesize\begin{align*}
aa:=\{S_{1}=1,\,D_{0}=D_{1}=1\}\quad \text{(always-selected
always-takers)}.
\end{align*}}%
We also allow for types that are never selected, 

{\footnotesize \begin{align*}
nc&{:=}\{S_{0}=S_{1}=0,\,D_{0}=0,D_{1}=1\},\\ na&{:=}%
\{S_{1}=0,\,D_{0}=D_{1}=1\},\\ nn&{:=}\{S_{0}=0,%
\,D_{0}=D_{1}=0\},
\end{align*}}%
which will play no role in the Lee-type selected distributions because $S=1$
never occurs for them. Let $\pi _{t}$:=$P (\text{type }t)$ denote the
population share of type $t$.

\begin{lemma}[Decomposition of selected distributions]
\label{lem:Lee_decomposition} Suppose Assumptions~\ref{ass:IV}--\ref%
{ass:monotoneS} hold and the relevant latent types exist. Then the distribution of $Y$ among the selected units
satisfies 

{\footnotesize \begin{align}
F_{Y| S=1,Z=1}(y)& =\omega \,F_{Y_{0}| an}(y)+(1-\omega
)\,F_{Y_{1}| ac\cup cc\cup aa}(y),  \\ \label{eq:FY_SZ1_mix_appendix} \\
F_{Y| S=1,Z=0}(y)& =\delta \,F_{Y_{0}| ac\cup an}(y)+(1-\delta
)\,F_{Y_{1}| aa}(y),  \label{eq:FY_SZ0_mix_appendix}
\end{align}}%
where 

{\footnotesize\begin{align*}
\omega :=\frac{\pi _{an}}{\pi _{ac}+\pi _{cc}+\pi _{aa}+\pi _{an}}%
,\qquad \delta :=\frac{\pi _{ac}+\pi _{an}}{\pi _{ac}+\pi _{an}+\pi
_{aa}}.
\end{align*}}
\end{lemma}

\begin{proof}
{\footnotesize \begin{align*}
F_{Y|S=1,Z=1}\left( y\right)  &=F_{Y|(1-D)SZ=1}\left( y\right)
E\left[(1-D)SZ|SZ=1\right]  \\
&+F_{Y|DSZ=1}\left( y\right) E\left[ DSZ|SZ=1\right]  \\
&=F_{Y_{0}|S_{0}=1,D_{1}=0}\left( y\right) \frac{E\left[ S_{0}\left(
1-D_{1}\right) \right] }{E\left[ S_{1}D_{1}\right] +E\left[ S_{0}\left(
1-D_{1}\right) \right] } \\
&+F_{Y_{1}|S_{1}=1,D_{1}=1}\left( y\right) \frac{E\left[ S_{1}D_{1}\right] 
}{E\left[ S_{1}D_{1}\right] +E\left[ S_{0}\left( 1-D_{1}\right) \right] } \\
&=F_{Y_{0}|an}\left( y\right) \omega +F_{Y_{1}|ac\cup cc\cup aa}\left(
y\right) \left( 1-\omega \right) ,
\end{align*}}%
where $\omega $:=$\frac{\pi _{an}}{\pi _{ac}+\pi _{cc}+\pi _{aa}+\pi _{an}}$.

Similarly,

{\footnotesize \begin{align*}
F_{Y|S=1,Z=0}\left( y\right)  &=F_{Y|S=1,D=0,Z=0}\left( y\right) E\left[
(1-D)S \left( 1-Z\right) |S\left( 1-Z\right)=1\right]  \\
&+F_{Y|S=1,D=1,Z=0}\left( y\right) E\left[ DS\left( 1-Z\right) |S\left( 1-Z\right)=1\right] 
\\
&=F_{Y_{0}|S_{0}=1,D_{0}=0}\left( y\right) \frac{E\left[ S_{0}\left(
1-D_{0}\right) \right] }{E\left[ S_{1}D_{0}\right] +E\left[ S_{0}\left(
1-D_{0}\right) \right] } \\
&+F_{Y_{1}|S_{1}=1,D_{0}=1}\left( y\right) \frac{E\left[ S_{1}D_{0}\right] 
}{E\left[ S_{1}D_{0}\right] +E\left[ S_{0}\left( 1-D_{0}\right) \right] } \\
&=F_{Y_{0}|ac\cup an}\left( y\right) \delta +F_{Y_{1}|aa}\left( y\right)
\left( 1-\delta \right) ,
\end{align*}}%
where $\delta $:=$\frac{\pi _{ac}+\pi _{an}}{\pi _{ac}+\pi
_{aa}+\pi _{an}}$
\end{proof}

We can now make precise the relationship between \cite{lee2009training} lower bound and our
lower bound for SLATE.
Let $\beta _{L}$ denote our lower bound for SLATE, 

{\footnotesize\begin{align*}
\beta _{L}=E\left[ Y_{1}| Y_{1}\leq Q_{1}(p),\,ac \cup cc\right] -E[Y_{0}|
ac],
\end{align*}}%
with $p:=\frac{\pi _{ac}}{\pi _{ac}+\pi _{cc}}$ is the fraction of
always-selected compliers within the complier mixture and $Q_{1}\left( u
\right) :=\inf\left\{ y\in \mathcal{Y}:F_{Y_{1}|ac \cup cc}\left( y\right) \geq u
\right\} $

Let $\beta _{L}^{Lee}$ denote the \cite{lee2009training} lower bound for the
ITT effect of $Z$ on $Y$ among the selected, as in %
\eqref{eq:Lee_lower_def}. Recall $p^{Lee}:=\frac{E\left[S|Z=0\right]}{E\left[S|Z=1\right]}=\frac{\pi _{ac}+\pi _{aa}+\pi _{an}}
{\pi _{ac}+\pi _{cc}+\pi _{aa}+\pi _{an}}$, and $Q_{1}^{Lee}\left( u \right) :=\inf\left\{ y\in \mathcal{Y}%
:F_{Y_{0}|an}\left( y\right) \omega +F_{Y_{1}|ac \cup cc \cup aa}\left( y\right)
\left( 1-\omega \right) \geq u \right\} $. For notational convenience, let the trimming threshold be $q:=Q^{Lee}_1 \left(p^{Lee} \right)$. Note that evaluating \eqref{eq:FY_SZ1_mix_appendix} at $y=q$ yields 

{\footnotesize
\begin{align*}
p^{Lee}=\omega F_{Y_0|an}\left(q\right)+\left(1-\omega \right)F_{Y_1|ac \cup cc \cup aa}\left(q\right). 
\end{align*}}
Combining Lemma~\ref{lem:Lee_decomposition} with
the definition of $\beta _{L}^{Lee}$, one can write the first term in $\beta
_{L}^{Lee}$ as%

{\footnotesize 
\begin{align*}
\beta^{Lee}_{L,1}
&:= E\!\left[Y \mid Y \le Q^{Lee}_1(p^{Lee}),\, S=1,\, Z=1\right]\\
&= w(q)\,E\!\left[Y_0 \mid Y_0 \le q,\, an\right]
+ \bigl(1-w(q)\bigr)\,E\!\left[Y_1 \mid Y_1 \le q,\, ac\cup cc\cup aa\right],    
\end{align*}}
where

{\footnotesize \begin{align*}
w(q)=\frac{\omega F_{Y_0\mid an}(q)}{p^{Lee}}
\end{align*}}
and the second term in $\beta_{L}^{Lee}$ as 

{\footnotesize \begin{align*}
\beta _{0}^{Lee} &:=E\left[ Y|S=1,Z=0\right]  \\
&=E\left[ Y_{0}|ac \cup an\right] \delta +E\left[ Y_{1}|aa\right] \left( 1-\delta \right) .
\end{align*}}%
So together, 

{\footnotesize\begin{align*}
\beta^{Lee}_L
&=
\{\text{weighted average of } Y_0 \text{ for } an \text{ and } Y_1 \text{ for } ac,cc,aa\}\\
&-
\{\text{weighted average of } Y_0 \text{ for } ac,an \text{ and } Y_1 \text{ for } aa\},
\end{align*}}%
where the exact weights depend on $(\pi_{ac},\pi_{cc},\pi_{an},\pi_{aa})$, the trimming
fraction $p^{Lee}$, and the component CDFs entering $F_{Y\mid S=1,Z=1}$ evaluated at
the trimming threshold $q$.

\subsubsection{Proof of Proposition \ref{prop:Lee_not_scaling}}
\begin{proof}
By Lemma~\ref{lem:Lee_decomposition}, the first term in $\beta _{L}^{Lee}$
is a trimmed mean of $Y$ over the mixture $\{ac \cup cc\cup aa \cup an\}$ in the selected
treated group, and the second term is an untrimmed mean over a different
mixture of types $\{ac,an,aa\}$ in the selected control group. Both mixtures
involve always-takers and/or never-takers whenever $\pi _{an}+\pi _{aa}>0$.

By contrast, $\beta _{L}$ depends only on complier types ($ac$ and $cc$):
its $Y_{1}$ component is a trimmed mean of $Y_{1}$ restricted to $\{ac \cup cc\}$%
, and its $Y_{0}$ component is the mean of $Y_{0}$ for always-selected
compliers ($ac$). Unless $\pi _{an}=\pi _{aa}=0$, the distributions of $Y_{0}
$ and $Y_{1}$ entering $\beta _{L}^{Lee}$ necessarily differ from those
entering $\beta _{L}$ because of the additional type-specific components.

Holding the primitive probabilities fixed, one can vary the distribution of $Y_0$ for
$an$ and/or the distribution of $Y_1$ for $aa$ while leaving the distributions for $ac$
and $cc$ unchanged; this changes $\beta^{Lee}_L$ through the additional mixture
components, and in the treated term also through the trimming threshold $q$ and the effective $an$ stratum weight $w(q)$, while leaving $\beta_L$ unchanged.

Therefore, for generic joint distributions of $(Y_0,Y_1)$ across types,
$\beta^{Lee}_L$ and $\beta_L$ are distinct functionals, and they cannot be linked  uniformly
over DGPs by a
single multiplicative function depending only on the type probabilities.
\end{proof}

\subsection{Supplementary Material for Section~\protect\ref{subsec:CF_comparison}%
}

\label{app:CF_dominance}

In this appendix we formalize the comparison between our bounds and those of
Chen and Flores (2015). The strict comparisons are under the maintained Section \ref{sec:bounds} regularity for the trimmed-mean representation.

\subsubsection{Identification of $F_{Y_{1}| ac\cup cc}$}

We first show how the distribution of $Y_{1}$ for the complier mixture ($ac$
and $cc$) can be recovered from observable selected distributions and identified type shares.

\begin{lemma}
\label{lem:F_Y1_ac_cc_identification} Suppose Assumptions~\ref{ass:IV}--\ref%
{ass:monotoneS} hold and $\pi_{ac}>0$. If $\pi_{aa}>0$, then 

{\footnotesize\begin{align*}
F_{Y_{1}| ac\cup cc}(y)=\frac{(\pi _{ac}+\pi _{cc}+\pi _{aa})F_{Y_{1}|
ac\cup cc\cup aa}(y)-\pi _{aa}F_{Y_{1}| aa}(y)}{\pi _{ac}+\pi _{cc}}.
\end{align*}}%

Equivalently, in terms of observable distributions, 
{\footnotesize\begin{align}
F_{Y_{1}| ac\cup cc}(y)=\frac{(\pi _{ac}+\pi _{cc}+\pi _{aa})F_{Y|
DSZ=1}(y)-\pi _{aa}F_{Y| DS(1-Z)=1}(y)}{\pi _{ac}+\pi _{cc}},
\label{eq:F_Y1_ac_cc_iden_appendix}
\end{align}}%
where 

{\footnotesize\begin{align*}
F_{Y| DSZ=1}(y)=F_{Y_{1}| ac\cup cc\cup aa}(y),\qquad F_{Y|
DS(1-Z)=1}(y)=F_{Y_{1}| aa}(y).
\end{align*}}
If $\pi_{aa}=0$, then \eqref{eq:F_Y1_ac_cc_iden_appendix} reduces to $F_{Y_{1}| ac\cup cc}(y)=F_{Y| DSZ=1}(y)$.
\end{lemma}

\begin{proof}
Among units with $DSZ=1$, the only types that can appear are $ac$, $cc$, and 
$aa$, all with $S_{1}=1$ and $D_{1}=1$. Thus when $\pi_{aa}>0$,

{\footnotesize\begin{align*}
F_{Y| DSZ=1}(y)=F_{Y_{1}| ac\cup cc\cup aa}(y).
\end{align*}}%
Otherwise, 

{\footnotesize\begin{align*}
F_{Y| DSZ=1}(y)=F_{Y_{1}| ac\cup cc}(y).
\end{align*}}%
Further when $\pi_{aa}>0$, among units with $DS(1-Z)=1$ the only possible type is $aa$, or those with $S_{1}=1$ and $D_{0}=1$, so 

{\footnotesize\begin{align*}
F_{Y| DS(1-Z)=1}(y)=F_{Y_{1}| aa}(y).
\end{align*}}%
Writing $F_{Y_{1}| ac\cup cc\cup aa}$ as a mixture of $F_{Y_{1}|
ac\cup cc}$ and $F_{Y_{1}| aa}$, we have 

{\footnotesize\begin{align*}
F_{Y_{1}| ac\cup cc\cup aa}(y)=\frac{\pi _{ac}+\pi _{cc}}{\pi _{ac}+\pi
_{cc}+\pi _{aa}}F_{Y_{1}| ac\cup cc}(y)+\frac{\pi _{aa}}{\pi _{ac}+\pi
_{cc}+\pi _{aa}}F_{Y_{1}| aa}(y),
\end{align*}}%
and solving for $F_{Y_{1}| ac\cup cc}(y)$ yields %
\eqref{eq:F_Y1_ac_cc_iden_appendix}.
\end{proof}

\subsubsection{Dominance over CF bounds: Preliminaries}

For clarity, observe the following equivalences of conditional distributions:

{\footnotesize\[
\bigl(Y_1 \mid S_1 = 1,\; D_1 > D_0\bigr) \equiv \bigl(Y_1 \mid ac \cup cc\bigr), \qquad
\bigl(Y \mid DSZ = 1\bigr) \equiv \bigl(Y_1 \mid ac \cup cc \cup aa\bigr),
\]}%
{\footnotesize\[
\bigl(Y \mid DS(1-Z) = 1\bigr) \equiv \bigl(Y_1 \mid aa\bigr).\] }%
As in the main text, let $\beta _{L,1}$ denote the first component (pertaining to $Y_{1}$) of our
lower bound (i.e.~excluding $\beta _{0}$), 

{\footnotesize\begin{align*}
\beta _{L,1}\text{=}E\left[ Y_{1}| Y_{1}\leq Q_{1}(p),ac \cup cc\right]
,\qquad p:=\frac{\pi _{ac}}{\pi _{ac}+\pi _{cc}}.
\end{align*}}%
Let $\beta_{L,1}^{mix}$ and $\beta_{L,1}^{adj}$ be the first components of the CF basic and alternative lower bounds obtained by trimming $Y| DSZ=1$ or equivalently
$Y_{1}| ac\cup cc\cup aa$ at fractions $p^{CF}_1$ and $p^{CF}_2$:

{\footnotesize\begin{align*}
\beta_{L,1}^{mix}\text{=}E\left[ Y_{1}| Y_{1}\leq
Q_{1}^{CF}(p^{CF}_1),ac \cup cc \cup aa\right] ,\qquad p^{CF}_1:=\frac{\pi _{ac}}{\pi
_{ac}+\pi _{cc}+\pi _{aa}},
\end{align*}}%

{\footnotesize\begin{align*}
\beta_{L,1}^{adj}\text{=}\left( 1+\frac{\pi _{aa}}{\pi _{ac}}\right)E\left[ Y_{1}| Y_{1}\leq Q_{1}^{CF}(p^{CF}_2
),ac \cup cc \cup aa\right] -\frac{\pi
_{aa}}{\pi _{ac}}\,E\left[ Y_{1}| aa\right] ,\qquad p^{CF}_2 :=\frac{%
\pi _{ac}+\pi _{aa}}{\pi _{ac}+\pi _{cc}+\pi _{aa}}.
\end{align*}}%
The CF lower bound uses 

{\footnotesize\begin{align*}
\beta_{L,1}^{CF}:=\max \{\beta_{L,1}^{mix},\beta_{L,1}^{adj}\}.
\end{align*}}
The formal proof of Proposition \ref{prop:CF_dominance_main} is provided below.

\begin{proof}
We sketch the argument for the lower bounds; the upper bounds are treated
analogously by symmetry.

First of all, if $\pi_{aa}=0$, i.e., $Pr\left(DS\left(1-Z\right)=1\right)=0$, then $F_{Y_1|ac \cup cc \cup aa}\left(y\right)$ reduces to $F_{Y_1|ac \cup cc}\left(y\right)$ and the trimming threshold $p^{CF}_1=p^{CF}_2=p$, so $\beta_{L,1}=\beta_{L,1}^{CF}$ and $\beta_{U,1}=\beta_{U,1}^{CF}$. In the following, we assume $\pi_{aa}>0$. 

$\beta _{L,1}$ is obtained by assuming that $ac
$ individuals are at the bottom $p$:=$\frac{\pi _{ac}}{\pi _{ac}+\pi _{cc}}$
fraction of the conditional distribution $Y_{1}|ac \cup cc$, i.e., $\beta _{L,1}$ assumes%

{\footnotesize\begin{align}
F_{Y_{1}|ac}\left( y\right) =F_{Y_{1}|ac}^{dh}\left( y\right) :=
\begin{cases}
\frac{1}{p}F_{Y_{1}|ac\cup cc}\left( y\right) &\text{ if }y<Q_{1}\left(
p\right)  \\ 
1&\text{ if }y\geq Q_{1}\left( p\right) 
\end{cases}
\label{eq:F_Y1_ac}
\end{align}}%
where $Q_{1}\left( p\right) :=F_{Y_{1}|ac\cup cc}^{-1}\left( p\right) $%
. Following Equation (\ref{eq:F_Y1_ac_cc_iden_appendix}), 

{\footnotesize \begin{align}
F_{Y_{1}|ac\cup cc}\left( y\right)  =\left( 1+\frac{\pi _{aa}}{\pi _{ac}+\pi
_{cc}}\right) F_{Y_{1}|ac\cup cc\cup aa}(y)  
-\frac{\pi _{aa}}{\pi _{ac}+\pi _{cc}}F_{Y_{1}|aa}(y).  \label{eq:F_Y1_ac,cc}
\end{align}}%
Plugging in Equation (\ref{eq:F_Y1_ac,cc}) into (\ref{eq:F_Y1_ac}) shows that $\beta_{L,1}$ assumes%

{\footnotesize\begin{align}
F_{Y_1|ac}\left( y\right)= F_{Y_{1}|ac}^{dh}\left( y\right) =
\begin{cases}
\frac{1}{p^{CF}_1}F_{Y_{1}|ac\cup cc\cup aa}(y)-\frac{\pi _{aa}}{\pi _{ac}}%
F_{Y_{1}|aa}(y) &\text{ if } y<Q_{1}\left( p\right)  \\ 
1 & \text{ if } y\geq Q_{1}\left( p\right) 
\end{cases}
\label{eq:F_Y1_ac_dh}
\end{align}}

In contrast, $\beta_{L,1}^{mix}$ is obtained by assuming that $ac$
individuals are at the bottom $p^{CF}_1$:=$\frac{\pi _{ac}}{\pi _{ac}+\pi _{cc}+\pi
_{aa}}$ fraction of the $Y_{1}|ac\cup cc\cup aa$ distribution (the CF restriction 1), i.e.,

{\footnotesize\begin{align}
F_{Y_{1}|ac}\left( y\right) =F_{Y_{1}|ac}^{mix}\left( y\right) :=%
\begin{cases}
\frac{1}{p^{CF}_1}F_{Y_{1}|ac\cup cc\cup aa}(y) &\text{ if } y<Q_{1}^{cf}\left(
p^{CF}_1\right)  \\ 
1 &\text{ if } y\geq Q_{1}^{cf}\left( p^{CF}_1\right) 
\end{cases}%
,  \label{eq:F_Y1_ac_basic}
\end{align}}%
where $Q_{1}^{cf}\left( p^{CF}_1\right) :=F_{Y_{1}|ac\cup cc \cup aa}^{-1}( p^{CF}_1)$.

Putting Equations (\ref{eq:F_Y1_ac_dh}) and (\ref{eq:F_Y1_ac_basic}) together and noticing $%
Q_{1}^{cf}\left( p^{CF}_1\right) \leq Q_{1}\left( p\right) $, we have 

{\footnotesize\begin{align}
F^{mix}_{Y_1|ac}(y)-F^{dh}_{Y_1|ac}(y)
=
\begin{cases}
\displaystyle
\frac{\pi_{aa}}{\pi_{ac}}\,F_{Y_1|aa}(y),
& y< Q^{cf}_1\!\left(p^{CF}_1\right),\\[1em]
\displaystyle
1-\frac{1}{p}F_{Y_1|ac\cup cc}(y)>0,
& Q^{cf}_1\!\left(p^{CF}_1\right)\le y< Q_1(p),\\[1em]
0,
& y\ge Q_1(p).
\end{cases}\label{eq:Diff_F_Y1_ac}
\end{align}
}
To show $Q_{1}^{cf}\left( p^{CF}_1\right) \leq Q_{1}\left( p\right)$, evaluating Equation (\ref{eq:F_Y1_ac,cc}) at $y=Q_{1}\left( p \right)$ and rearranging yields 

{\footnotesize\begin{align}
F_{Y_{1}|ac\cup cc\cup aa}(Q_{1}\left( p\right) ) &=  p^{CF}_1+\frac{\pi _{aa}}{\pi
_{ac}+\pi _{cc}+\pi _{aa}}F_{Y_{1}|aa}(Q_{1}\left( p \right)) \label{eq:F_Y1_ac_cc_aa_eval} \\ 
&\geq p^{CF}_1=F_{Y_{1}|ac\cup cc\cup aa}(Q_{1}^{cf}\left( p^{CF}_1\right) )\nonumber. 
\end{align}}
$Q_{1}^{cf}\left( p^{CF}_1\right) \leq Q_{1}\left( p\right)$ follows from that $F_{Y_{1}|ac\cup
cc\cup aa}(\cdot)$ is strictly increasing.

Note that under the maintained regularity, Equation~(\ref{eq:F_Y1_ac,cc})
implies the following equivalence: 

{\footnotesize\begin{align*}
 F_{Y_1\mid aa}(Q_1(p))>0 \iff Q_1^{cf}(p_1^{CF})<Q_1(p)\iff F_{Y_1\mid aa}(Q_1^{cf}(p_1^{CF}))>0. 
\end{align*}}
Equation (\ref{eq:Diff_F_Y1_ac}) implies $F_{Y_{1}|ac}^{mix}\left( y\right) \geq
F_{Y_{1}|ac}^{dh}\left( y\right) $ in general, i.e., the
conditional distribution $Y_{1}|ac$ assumed in $\beta _{L,1}$ weakly stochastically
dominates that assumed in $\beta_{L,1}^{mix}$. Furthermore, $F_{Y_{1}|ac}^{mix}\left( y\right) >
F_{Y_{1}|ac}^{dh}\left( y\right) $ iff $\pi_{aa}>0$ and $F_{Y_1|aa}(Q_1(p))>0$. The mean respects stochastic dominance, so 

{\footnotesize\begin{align*}
\beta _{L,1}\geq \beta_{L,1}^{mix} 
\end{align*}}%
in general, as well as  

{\footnotesize\begin{align*}
\beta _{L,1 }> \beta_{L,1}^{mix}, \quad \text{iff} \quad \pi_{aa}>0 \quad\text{and}\quad F_{Y_1|aa}(Q_1(p))>0.  
\end{align*}}

Now consider $\beta^{adj}_{L,1}$, which is based on the CF restriction 2. It assumes
that $ac$ and $aa$ are at the bottom $p^{CF}_2 $:=$\frac{\pi _{ac}+\pi _{aa}}{\pi _{ac}+\pi _{cc}+\pi _{aa}}$ fraction of the conditional distribution $%
Y_{1}|ac\cup cc\cup aa$, i.e.%

{\footnotesize\begin{align}
F_{Y_{1}|ac\cup aa}\left( y\right) =F_{Y_{1}|ac\cup aa}^{adj}\left( y\right) 
:=
\begin{cases}
\frac{1}{p^{CF}_2}F_{Y_{1}|ac\cup cc\cup aa}(y)&\text{ if }y<Q_{1}^{cf}\left(
p^{CF}_2 \right)  \\ 
1&\text{ if \ }y\geq Q_{1}^{cf}\left( p^{CF}_2 \right) 
\end{cases}%
.\label{eq:F_Y1_ac_aa_cf}
\end{align}}%

In contrast, we have shown that $\beta _{L,1}$ assumes%

{\footnotesize\begin{align*}
F_{Y_{1}|ac}\left( y\right) =F_{Y_{1}|ac}^{dh}\left( y\right) :=
\begin{cases}
\frac{1}{p^{CF}_1}F_{Y_{1}|ac\cup cc\cup aa}(y)-\frac{\pi _{aa}}{\pi _{ac}}%
F_{Y_{1}|aa}(y)&\text{ if }y<Q_{1}\left( p\right)  \\ 
1&\text{ if }y\geq Q_{1}\left( p\right) 
\end{cases} %
,
\end{align*}}%
which implies $\beta_{L,1}$ assumes, for $F_{Y_{1}|ac\cup aa}\left( y\right)$,  %

{\footnotesize \begin{align}
F_{Y_{1}|ac\cup aa}\left( y\right)  =F_{Y_{1}|ac \cup aa}^{dh}\left( y\right)
&:=\nonumber \frac{\pi _{ac}}{\pi _{ac}+\pi _{aa}}F_{Y_{1}|ac}^{dh}\left( y\right) +%
\frac{\pi _{aa}}{\pi _{ac}+\pi _{aa}}F_{Y_{1}|aa}\left( y\right)  \\
&= 
\begin{cases}
\frac{1}{p^{CF}_2}F_{Y_{1}|ac\cup cc\cup aa}(y)&\text{\ if }y<Q_{1}\left(
p\right)  \\ 
1-\frac{\pi _{aa}}{\pi _{ac}+\pi _{aa}}\left(1-F_{Y_{1}|aa}\left( y\right)\right) &\text{ if }y\geq Q_{1}\left( p\right) 
\end{cases}%
.\label{eq:F_Y1_ac_aa_dh}
\end{align}}%
\\
Putting together Equations (\ref{eq:F_Y1_ac_aa_cf}) and (\ref{eq:F_Y1_ac_aa_dh}) and noticing $Q_{1}^{cf}\left( p^{CF}_2\right) \geq Q_{1}\left( p\right)$ yield 

{\footnotesize\begin{align}
&F^{alt}_{Y_1|ac\cup aa}(y)-F^{dh}_{Y_1|ac\cup aa}(y) \notag\\
&\quad=
\begin{cases}
0,
& y< Q_1(p),\\[0.5em]
\displaystyle
\frac{\pi_{ac}+\pi_{cc}}{\pi_{ac}+\pi_{aa}}
\left(F_{Y_1|ac\cup cc}(y)-p\right) \ge 0 ,
& Q_1(p)\le y< Q^{cf}_1\!\left(p^{CF}_2\right),\\[0.9em]
\displaystyle
\frac{\pi_{aa}}{\pi_{ac}+\pi_{aa}}
\left(1-F_{Y_1|aa}(y)\right),
& y\ge Q^{cf}_1\!\left(p^{CF}_2\right).
\end{cases}
\label{eq:Diff_F_y1_ac_aa}
\end{align}} where the middle branch uses (\ref{eq:F_Y1_ac,cc}).

To show $Q_{1}^{cf}\left( p^{CF}_2\right) \geq Q_{1}\left( p\right)$, plugging into Equation (\ref{eq:F_Y1_ac_cc_aa_eval}) $p^{CF}_1=p^{CF}_2-\frac{\pi _{aa}}{\pi
_{ac}+\pi _{cc}+\pi _{aa}}$ and rearranging yields 

{\footnotesize\begin{align*}
F_{Y_{1}|ac\cup cc\cup aa}(Q_{1}\left( p\right) ) &=p^{CF}_2-\frac{\pi _{aa}}{\pi
_{ac}+\pi _{cc}+\pi _{aa}}\left(1-F_{Y_{1}|aa}(Q_{1}\left( p \right))\right) \\ 
&\leq p^{CF}_2 \\
&=F_{Y_{1}|ac\cup
cc\cup aa}(Q_{1}^{cf}\left( p^{CF}_2\right) ).
\end{align*}} 
$Q_{1}^{cf}\left( p^{CF}_2\right)\geq 
Q_{1}\left( p\right)$ follows, since $F_{Y_{1}|ac\cup
cc\cup aa}(\cdot)$ is strictly increasing.
Under the maintained regularity, Equation~(\ref{eq:F_Y1_ac,cc})
implies {\footnotesize$F_{Y_1\mid aa}(Q_1^{cf}(p_2^{CF}))<1 \iff Q_1^{cf}(p_2^{CF})>Q_1(p) \iff F_{Y_1\mid aa}(Q_1(p))<1$}. 

Equation (\ref{eq:Diff_F_y1_ac_aa}) implies $F_{Y_{1}|ac\cup aa}^{dh}\left( y\right) \leq F_{Y_{1}|ac\cup aa}^{adj}\left( y\right)$ in general, i.e., the conditional distribution $Y_{1}|ac\cup aa$ assumed in $\beta _{L,1}$ weakly stochastically
dominates that assumed in $\beta_{L,1}^{adj}$. Furthermore, $F_{Y_{1}|ac\cup aa}^{dh}\left( y\right) < F_{Y_{1}|ac\cup aa}^{alt}\left( y\right)$  iff $\pi_{aa}>0$ and $F_{Y_{1}|aa}(Q_{1}\left( p\right))<1$. The mean respects stochastic dominance, so in general

{\footnotesize\begin{align}
\frac{\pi _{ac}}{\pi _{ac}+\pi _{aa}}\beta _{L,1}+\frac{\pi _{aa}}{\pi
_{ac}+\pi _{aa}}E\left[ Y_{1}|aa\right] \geq E\left[ Y|DSZ=1,Y\leq
Q_{1}^{cf}\left( p^{CF}_2 \right) \right]   \label{eq:DH_vs_CF}
\end{align}}%
where the left-hand side is $E\left[ Y_{1}|ac\cup aa\right] $ assuming $%
F_{Y_{1}|ac\cup aa}\left( y\right) =F_{Y_{1}|ac\cup aa}^{dh}\left( y\right) $%
, while the right-hand side is $E\left[ Y_{1}|ac\cup aa\right] $ assuming $%
F_{Y_{1}|ac\cup aa}\left( y\right) =F_{Y_{1}|ac\cup aa}^{adj}\left( y\right) $.

Rearranging Equation (\ref{eq:DH_vs_CF}), we have, in general,

{\footnotesize \begin{align*}
\beta _{L,1} &\geq \left( 1+\frac{\pi _{aa}}{\pi _{ac}}\right) E\left[ Y|DSZ=1,Y\leq Q^{cf}_{1}\left( p^{CF}_2
\right) \right] -\frac{\pi _{aa}%
}{\pi _{ac}}E\left[ Y_{1}|aa\right]  \\
&=: \beta_{L,1}^{adj},
\end{align*}}
and further,

{\footnotesize\begin{align*}
\beta _{L,1 }> \beta_{L,1}^{adj}, \quad \text{iff} \quad \pi_{aa}>0 \quad\text{and}\quad F_{Y_{1}|aa}(Q_{1}\left( p\right))<1.  
\end{align*}}

The above shows (I) $\beta _{L,1}\geq \beta_{L,1}^{mix}$ in general  and  $\beta _{L,1}>\beta_{L,1}^{mix}$ iff $\pi_{aa}>0$ and $F_{Y_1|aa}\left(Q_1\left(p\right)\right)>0 $ and (II) $\beta
_{L,1}\geq \beta_{L,1}^{adj}$ in general and $\beta
_{L,1}> \beta_{L,1}^{adj}$ iff $\pi_{aa}>0$ and  $F_{Y_{1}|aa}(Q_{1}\left( p\right))<1$. Together, they imply 

{\footnotesize \begin{align*}
\beta _{L,1}\geq \beta _{L,1}^{CF}:=\max \{ \beta_{L,1}^{mix}, \beta_{L,1}^{adj}\}\quad \text{in general,}
\end{align*}}
and

{\footnotesize\begin{align*}
\beta _{L,1}> \beta _{L,1}^{CF},\quad\text{iff} \quad \pi_{aa}>0 \quad\text{and}\quad 0<F_{Y_{1}|aa}(Q_{1}\left( p\right))<1.  
\end{align*}}

\end{proof}

The intuition is as follows. By construction, $\beta_{L,1}$ is the sharp worst-case
lower bound for $E[Y_1\mid ac]$ given the identified complier mixture
$F_{Y_1\mid ac\cup cc}$ and the fact that $ac$ makes up the fraction $p$ of that
mixture. The worst case is attained by placing all $ac$ individuals below all
$cc$ individuals in $Y_1\mid ac\cup cc$.

Viewed inside the larger mixture $Y_1\mid ac\cup cc\cup aa$, this configuration
implies only that the $ac$ mass must be contained in the bottom

{\footnotesize\begin{align*}
1-\frac{\pi_{cc}}{\pi_{ac}+\pi_{cc}+\pi_{aa}}=p_2^{CF}
\end{align*}}
fraction of the distribution. Its exact location within that region is otherwise
unrestricted. The CF basic bound corresponds to the most unfavorable  placement within
this class, namely the one in which $ac$ occupies exactly the bottom

{\footnotesize\begin{align*}
p_1^{CF}=\frac{\pi_{ac}}{\pi_{ac}+\pi_{cc}+\pi_{aa}}
\end{align*}}
fraction. In that limiting case, $\beta_{L,1}=\beta_{L,1}^{mix}$. Any less
adverse placement shifts some $ac$ mass upward within the bottom $p_2^{CF}$
fraction, which raises $E[Y_1\mid ac]$ and hence $\beta_{L,1}$, while leaving
$\beta_{L,1}^{mix}$ unchanged. Therefore,

{\footnotesize\begin{align*}
\beta_{L,1}\geq \beta_{L,1}^{mix}.
\end{align*}}
For $\beta_{L,1}^{adj}$, CF Restriction~2 requires that $ac$ and $aa$ together
occupy the bottom $p_2^{CF}$ fraction of $Y_1\mid ac\cup cc\cup aa$, and then
uses the identified mean $E[Y_1\mid aa]$ to back out a lower bound on
$E[Y_1\mid ac]$. Under our sharp worst-case configuration, $ac$ still lies
entirely below $cc$, whereas $aa$ may be located anywhere in the mixture. If all
$aa$ mass also lies in the bottom $p_2^{CF}$ fraction, then
$\beta_{L,1}=\beta_{L,1}^{adj}$. If some $aa$ mass is shifted to higher
quantiles, the lower bound implied by CF Restriction~2 becomes smaller, while
$\beta_{L,1}$ remains unchanged. Hence,

{\footnotesize\begin{align*}
\beta_{L,1}\geq \beta_{L,1}^{adj}.
\end{align*}}
Combining the two comparisons yields

{\footnotesize\begin{align*}
\beta_{L,1}\geq \beta_{L,1}^{CF}:=\max\{\beta_{L,1}^{mix},\beta_{L,1}^{adj}\}.
\end{align*}}
The upper-bound comparison follows by the same argument applied to the upper
tails.

\newpage

\section{Supplementary Material for Section \ref{sec:bounds_covariates}}
\subsection{Moment representations with covariates}

\label{app:covariate_moments}

Under Assumption~\ref{ass:conditional_IV}, the conditional versions of Lemma~%
\ref{lem:IV_identities} imply, 

{\footnotesize\begin{align}
\lambda _{1}(x)& :=P (S_{1}=1,D_{1}>D_{0}| X=x)=E[W(Z,X)\,DS|
X=x],  \label{eq:lambda1x} \\
\lambda _{0}(x)& :=P (S_{0}=1,D_{1}>D_{0}|
X=x)=E[W(Z,X)\,(D-1)S| X=x],  \label{eq:lambda0x}
\end{align}}%
and, for any $y$, 

{\footnotesize\begin{align}
F_{1}\left( y|x\right) & :=F_{Y_{1}| S_{1}=1,D_{1}>D_{0},X=x}(y)=%
\frac{E[W(Z,X)\,DS \cdot \mathbbm{1}\{Y\leq y\}| X=x]}{E[W(Z,X)DS| X=x]},
\label{eq:F1_cond_appendix} \\
F_{0}\left( y|x\right) & :=F_{Y_{0}| S_{0}=1,D_{1}>D_{0},X=x}(y)=%
\frac{E[W(Z,X)\,(D-1)S\cdot \mathbbm{1}\{Y\leq y\}| X=x]}{E[W(Z,X)(D-1)S| X=x]%
}.  \label{eq:F0_cond_appendix}
\end{align}}%
where the inverse probability weight $W(Z,X)$:=$\frac{Z}{e(X)}-\frac{1-Z}{1-e(X)}$ with $e(x)
$:=$P(Z=1|X=x)$. The conditional quantiles $Q_{d}(u,x)$, $d=0,1$, used in %
\eqref{eq:betaL_cond_plus}-\eqref{eq:betaU_cond_minus} are defined as the
generalized inverses of these CDFs.

Using these identities and the same arguments as in Section~\ref%
{subsec:baseline_bounds}, one can verify the numerator moment function for $%
\beta _{L}$. For $x\in \mathcal{X}^{+}$, it is 

{\footnotesize \begin{align*}
\beta _{L}^{+}\left( x\right) \pi _{ac}\left( x\right)  
&=\beta _{L}^{+}\left( x\right) P \left( S_{0}=1,D_{1}>D_{0}|X=x\right) 
\\
&=E\left[ S_{1}Y_{1}\cdot \mathbbm{1}\left\{ Y_1\leq Q_{1}\left( p\left( x\right)
,x\right) \right\} \cdot \mathbbm{1}\left\{ D_{1}>D_{0}\right\} |{\small X=x}\right] 
\\
&-E\left[ S_{0}Y_{0}\cdot \mathbbm{1}\left\{ D_{1}>D_{0}\right\} |{\small X=x}\right] 
\\
&=E\left[ W\left( Z,X\right) \left( DSY\cdot \mathbbm{1}\left\{ Y\leq Q_{1}\left(
p\left( x\right) ,x\right) \right\} +\left( 1-D\right) SY\right) |X=x\right] 
\end{align*}}%
where the first equality follows from Assumption~\ref{ass:conditional_monotoneS}, the second from the definition
of $\beta_L^+(x)$ and the law of iterated expectations, and the last equality follows
from the conditional version of Lemma~\ref{lem:IV_identities}, applied with
$g(y)=y\,\mathbbm{1}\{y\le Q_1(p(x),x)\}$ for the treated term and with $D$ replaced by $(D-1)$
for the untreated term. 

Similarly for $x\in \mathcal{X}^{-}$, the numerator moment function is 

{\footnotesize \begin{align*}
\beta _{L}^{-}\left( x\right) \pi _{ac}\left( x\right)  
&=\beta _{L}^{-}\left( x\right) P \left( S_{1}=1,D_{1}>D_{0}|X=x\right) 
\\
&=E\left[ W\left( Z,X\right) \left( DSY+\left( 1-D\right) SY\cdot \mathbbm{1}\left\{
Y\geq Q_{0}\left( 1-1/p\left( x\right) ,x\right) \right\} \right) |X=x\right]
.
\end{align*}}

Analogously, one can derive the numerator moment function for $\beta _{U}$.
For $x\in \mathcal{X}^{+}$, it is 

{\footnotesize \begin{align*}
\beta _{U}^{+}\left( x\right) \pi _{ac}\left( x\right)  
&=\beta _{U}^{+}\left( x\right) P \left( S_{0}=1,D_{1}>D_{0}|X=x\right) 
\\
&=E\left[ S_{1}Y_{1}\cdot \mathbbm{1}\left\{ Y_1\geq Q_{1}\left( 1-p\left( x\right)
,x\right) \right\} \cdot \mathbbm{1}\left\{ D_{1}>D_{0}\right\} |{\small X=x}\right] 
\\
&-E\left[ S_{0}Y_{0}\cdot \mathbbm{1}\left\{ D_{1}>D_{0}\right\} |{\small X=x}\right] 
\\
&=E\left[ W\left( Z,X\right) \left( DSY\cdot \mathbbm{1}\left\{ Y\geq Q_{1}\left(
1-p\left( x\right) ,x\right) \right\} +\left( 1-D\right) SY\right) |X=x%
\right]. 
\end{align*}}

Similarly, for $x\in \mathcal{X}^{-}$, it is

{\footnotesize \begin{align*}
\beta _{U}^{-}\left( x\right) \pi _{ac}\left( x\right)  
&=\beta _{U}^{-}\left( x\right) P \left( S_{1}=1,D_{1}>D_{0}|X=x\right) 
\\
&=E\left[ W\left( Z,X\right) \left( DSY+\left( 1-D\right) SY\cdot \mathbbm{1}\left\{
Y\leq Q_{0}\left( 1/p\left( x\right) ,x\right) \right\} \right) |X=x\right] .
\end{align*}}

|

\newpage

\section{Supplementary Material for Section \ref{sec:inference}}
\subsection{Derivation of Influence Functions}

\subsubsection{Unconditional Case with Strong Sample Selection Monotonicity}
\subsubsection*{Basic Decomposition}
Under strong sample selection monotonicity and no covariates for both bounds we can decompose\footnote{
Note that $\beta_0$ here is used different from Section \ref{subsec:baseline_bounds}, i.e., it is the sign-adjusted control mean ($+$ instead of $-$ in the definition of $\beta_B$). This is for notational simpilcity only.} 

{\footnotesize \begin{align*}
    \beta_B &= \beta_{B,1} + \beta_0,
\end{align*}}
where 

{\footnotesize \begin{align*}
    \beta_{B,1} &= -\frac{1}{(m_1 - m_0)}[r_1\psi_{B,1} - r_0\psi_{B,0}], \\
    \beta_{0} &= -\frac{1}{(m_1 - m_0)}[\mu_1 - \mu_0],
\end{align*}}
where we denote shorthand 

{\footnotesize \begin{align*}
    m_z &= E[(1-D)S|Z=z], \\
    r_z &= E[DS|Z=z], \\
    \mu_z &= E[(1-D)SY|Z=z], \\
    \psi_{L,z} &= E[Y\mathbbm{1}(Y \leq Q_1(p))|DS=1,Z=z], \\
    \psi_{U,z} &= E[Y\mathbbm{1}(Y > Q_1(1-p))|DS=1,Z=z]. 
\end{align*}}
Now consider a parametric submodel indexed by $t \in (0,1]$. From the definition and the product rule it follows that 

{\footnotesize \begin{align*}
    -[m_1 - m_0]\dt \beta_{B,1,t}\bigg|_{t=0}
    &= \dt r_{1,t}\bigg|_{t=0}\psi_{B,1} - \dt r_{0,t}\bigg|_{t=0}\psi_{B,0} \\
    &\quad + r_1\dt \psi_{B,1,t}\bigg|_{t=0} - r_0 \dt\psi_{B,0,t}\bigg|_{t=0} + \beta_{B,1}\dt[m_{1,t} - m_{0,t}]\bigg|_{t=0}
\end{align*}}
and

{\footnotesize \begin{align*}
    -[m_1 - m_0]\dt \beta_{0,t}\bigg|_{t=0}
    &= \dt[\mu_{1,t}- \mu_{0,t}]\bigg|_{t=0} + \beta_0
\dt [m_{1,t} - m_{0,t}]\bigg|_{t=0}.\end{align*}}
We omit denoting the evaluation of the derivative at $t = 0$ in what follows whenever it does not cause confusion.

\subsubsection*{Quantile Auxiliary Results}
Differentiating the result of Lemma \ref{lem:IV_identities} for $g(Y) = \mathbbm{1}(Y\leq y)$ with respect to $y$ yields density function 

{\footnotesize \begin{align*}
    f_{Y_1|S_1=1,D_1>D_0}(y) &= \frac{r_1f_1(y) - r_0f_0(y)}{r_1 - r_0},
\end{align*}}
where $f_z(y) = f_{Y|DS=1,Z=z}(y)$. Using the definition of the quantile as well as Leibniz' rule for differentiation then yields

{\footnotesize \begin{align*}
    p &= \int^{Q_1(p)}f_{Y_1|S_1=1,D_1>D_0}(y)dy \\
    \Rightarrow -[m_1 - m_0] &= p[r_1 - r_0] \\
    &= r_1 \int^{Q_1(p)}f_1(y)dy - r_0\int^{Q_1(p)}f_0(y)dy \\
    \Rightarrow -\dt[m_{1,t} - m_{0,t}] &= \dt r_{1,t}F_1(Q_1(p)) - \dt r_{0,t} F_0(Q_1(p))\\
    &\quad + [r_1f_1(Q_1(p)) - r_0f_0(Q_1(p))]\dt Q_{1,t}(p_t) \\
    &\quad + r_1 \int^{Q_1(p)}\dt f_{1,t}(y)dy - r_0 \int^{Q_1(p)}\dt f_{0,t}(y)dy \\
    \Rightarrow [r_1f_1(Q_1(p)) &- r_0f_0(Q_1(p))]\dt Q_{1,t}(p_t) \\
    &= -\dt[m_{1,t} - m_{0,t}] - \left[\dt r_{1,t}F_1(Q_1(p)) - \dt r_{0,t} F_0(Q_1(p)) \right] \\
    &\quad - \left[r_1 \int^{Q_1(p)}\dt f_{1,t}(y)dy - r_0 \int^{Q_1(p)}\dt f_{0,t}(y)dy\right].
\end{align*}}
Equivalently 

{\footnotesize \begin{align*}
    p &= \int_{Q_1(1-p)}f_{Y_1|S_1=1,D_1>D_0}(y)dy \\
    \Rightarrow -[m_1 - m_0] &= p[r_1 - r_0] \\
    &= r_1 \int_{Q_1(1-p)}f_1(y)dy - r_0\int_{Q_1(1-p)}f_0(y)dy \\
    \Rightarrow -\dt[m_{1,t} - m_{0,t}] &= \dt r_{1,t}[1-F_1(Q_1(1-p))] - \dt r_{0,t} [1-F_0(Q_1(1-p))]\\
    &\quad - [r_1f_1(Q_1(1-p)) - r_0f_0(Q_1(1-p))]\dt Q_{1,t}(1-p_t) \\
    &\quad + r_1 \int_{Q_1(1-p)}\dt f_{1,t}(y)dy - r_0 \int_{Q_1(1-p)}\dt f_{0,t}(y)dy \\
    \Rightarrow -[r_1f_1(Q_1(p)) &- r_0f_0(Q_1(p))]\dt Q_{1,t}(1-p_t) \\
    &= -\dt[m_{1,t} - m_{0,t}] - \left[\dt r_{1,t}[1-F_1(Q_1(1-p))] - \dt r_{0,t} [1- F_0(Q_1(1-p))] \right] \\
    &\quad - \left[r_1 \int_{Q_1(1-p)}\dt f_{1,t}(y)dy - r_0 \int_{Q_1(1-p)}\dt f_{0,t}(y)dy\right].
\end{align*}}

\subsubsection*{Trimmed Mean Auxiliary Result}
Applying Leibniz' rule to the trimmed means yields 

{\footnotesize \begin{align*}
    \dt \psi_{L,z,t} &= \dt \int^{Q_{1,t}(p_t)}yf_{z,t}(y)dy \\
    &= Q_1(p)f_z(Q_1(p))\dt Q_{1,t}(p_t) + \int^{Q_1(p)}y \dt f_{z,t}(y)dy \\
    \dt \psi_{U,z,t} &= \dt \int_{Q_{1,t}(1-p_t)}yf_{z,t}(y)dy \\
    &= -Q_1(1-p)f_z(Q_1(1-p))\dt Q_{1,t}(1-p_t) + \int_{Q_1(1-p)}y \dt f_{z,t}(y)dy,
\end{align*}}
which implies that

{\footnotesize \begin{align*}
    r_1 \dt \psi_{L,1,t} - r_0 \dt \psi_{L,0,t}
    &= Q_1(p)[r_1f_1(Q_1(p)) - r_0f_0(Q_1(p))]\dt Q_{1,t}(p_t)\\
    &\quad + \int^{Q_1(p)}y\left[r_1 \dt f_{1,t}(y) - r_0 \dt f_{0,t}(y)\right]dy \\
     r_1 \dt \psi_{U,1,t} - r_0 \dt \psi_{U,0,t}
    &= -Q_1(1-p)[r_1f_1(Q_1(1-p)) - r_0f_0(Q_1(1-p))]\dt Q_{1,t}(1-p_t) \\
    &\quad + \int_{Q_1(1-p)}y\left[r_1 \dt f_{1,t}(y) - r_0 \dt f_{0,t}(y)\right]dy. \\
\end{align*}}
\subsubsection*{Scaled $\beta_B$ Pathwise Derivative}
Plugging in both auxiliary results thus yields 

{\footnotesize \begin{align*}
    -&[m_1 - m_0]\dt \beta_{L,1,t}
    = \dt r_{1,t}\psi_{L,1} - \dt r_{0,t}\psi_{L,0} + \beta_{L,1}\dt [m_{1,t} - m_{0,t}] \\
    &\quad - Q_1(p)\bigg( \dt[m_{1,t} - m_{0,t}] + \dt r_{1,t}F_1(Q_1(p)) - \dt r_{0,t}F_0(Q_1(p)) \\
    &\qquad \ \qquad + \int^{Q_1(p)}\left[r_1 \dt f_{1,t}(y) - r_0 \dt f_{0,t}\right]dy\bigg) 
     + \int^{Q_1(p)}y\left[r_1 \dt f_{1,t}(y) - r_0 \dt f_{0,t}\right]dy,
\end{align*}}
as well as 

{\footnotesize \begin{align*}
     -&[m_1 - m_0]\dt \beta_{U,1,t}
    = \dt r_{1,t}\psi_{U,1} - \dt r_{0,t}\psi_{U,0} + \beta_{U,1}\dt [m_{1,t} - m_{0,t}] \\
    &\quad - Q_1(1-p)\bigg( \dt[m_{1,t} - m_{0,t}] + \dt r_{1,t}[1-F_1(Q_1(1-p))] - \dt r_{0,t}[1-F_0(Q_1(1-p))]\\
    &\qquad \ \qquad  + \int_{Q_1(1-p)}\left[r_1 \dt f_{1,t}(y) - r_0 \dt f_{0,t}\right]dy\bigg) 
     + \int_{Q_1(1-p)}y\left[r_1 \dt f_{1,t}(y) - r_0 \dt f_{0,t}\right]dy.
\end{align*}}

\subsubsection*{Influence Function Primitive Components}
 In the following, we provide the influence function of the primitives via the pathwise derivative of the parameter \(\beta_B(P)\) in a fully nonparametric model satisfying Assumptions A.1--A.7 without covariates $X$. We use the previous derivations and standard influence functions for all unrestricted conditional mean primitives, see,  e.g., \cite{kennedy2024semiparametric}. 

{\footnotesize \begin{align*}
    \dt \mu_{z,t} &= \frac{\mathbbm{1}(Z=z)}{P(Z=z)}((1-D)SY - \mu(z)), \\
    \Rightarrow \dt[\mu_{1,t} - \mu_{0,t}] &= W(Z)((1-D)SY - \mu(Z)),
\end{align*}}
and

{\footnotesize \begin{align*}
    \dt m_{z,t} &= \frac{\mathbbm{1}(Z=z)}{P(Z=z)}((1-D)S - m(z)), \\
    \Rightarrow \dt [m_{1,t} - m_{0,t}] &= W(Z)((1-D)S - m(Z)), \\
    \dt r_{z,t} &= \frac{\mathbbm{1}(Z=z)}{P(Z=z)}(DS - r(z)), \\
    \Rightarrow \dt r_{1,t}\psi_{B,1} - \dt r_{0,t}\psi_{B,0} &= W(Z)(DS - r(Z))\psi_{B}(Z), \\
    \dt r_{1,t}F_1(Q_1(p))- \dt r_{0,t}F_0(Q_1(p)) &= W(Z)(DS - r(Z))F(Q_1(p)|1,Z), \\
    \dt r_{1,t}[1-F_1(Q_1(1-p))]- \dt r_{0,t}[1-F_0(Q_1(1-p))] &= W(Z)(DS - r(Z))[1-F(Q_1(1-p)|1,Z)],
\end{align*}}
where $\psi_{B}(Z) = Z\psi_{B,1} + (1-Z)\psi_{B,0}$ and $F(u|Z) = ZF_1(u) + (1-Z)F_0(u)$. Now consider the conditional probability $f_z(y) = f(y|DS=1,Z=z)$. By standard arguments, its influence function is given by 

{\footnotesize \begin{align*}
    \dt f_{z,t}(y) &= \frac{1(Z=z)DS}{P(DS=1|Z=z)P(Z=z)}(\mathbbm{1}(Y = y) - f_z(y)) \\
    \Rightarrow r_z \int^{Q_1(p)}\dt f_{z,t}(y)dy &= \frac{1(Z=z)DS}{P(Z=z)}\left(\mathbbm{1}(Y \leq Q_1(p)) - \int^{Q_1(p)}f_z(y)dy\right) \\&= \frac{1(Z=z)DS}{P(Z=z)}\left(\mathbbm{1}(Y \leq Q_1(p)) - F_z(Q_1(p))\right).
\end{align*}}
Thus we have that 

{\footnotesize \begin{align*}
    \int^{Q_1(p)}\left[r_1\dt f_{1,t}(y) - r_0 \dt f_{0,t}(y)\right]dy &= W(Z)DS\bigg(\mathbbm{1}(Y \leq Q_1(p)) - F(Q_1(p)|1,Z)\bigg) \\
\end{align*}}
and, by similar derivations, 

{\footnotesize \begin{align*}
    \int^{Q_1(p)}y\left[r_1\dt f_{1,t}(y) - r_0 \dt f_{0,t}(y)\right]dy &= W(Z)DS\bigg(Y\mathbbm{1}(Y \leq Q_1(p)) - \psi_L(Z)\bigg). \\
\end{align*}}
The upper bound components are found analogously as 

{\footnotesize \begin{align*}
    \int_{Q_1(1-p)}\left[r_1\dt f_{1,t}(y) - r_0 \dt f_{0,t}(y)\right]dy &= W(Z)DS\bigg(\mathbbm{1}(Y > Q_1(1-p)) - [1-F(Q_1(1-p)|1,Z)]\bigg) \\
\end{align*}}
and  

{\footnotesize \begin{align*}
    \int_{Q_1(1-p)}y\left[r_1\dt f_{1,t}(y) - r_0 \dt f_{0,t}(y)\right]dy &= W(Z)DS\bigg(Y\mathbbm{1}(Y > Q_1(1-p)) - \psi_U(Z)\bigg). \\
\end{align*}}

\subsubsection*{Influence Function for Combined Rescaled $\beta_0$ and $\beta_1$}
Plugging in the previous result into the rescaled pathwise derivative of $\beta_0$ yields

{\footnotesize \begin{align*}
    -[m_1 - m_0]\dt \beta_{0,t}
    &= W(Z)((1-D)SY - \mu(Z)) + \beta_0 W(Z)((1-D)S - m(Z)).
\end{align*}}
For the lower bound we obtain 

{\footnotesize \begin{align*}
    -[m_1 - m_0]\dt \beta_{L,1} 
    &= W(Z)(DS - r(Z))\psi_L(Z) + \beta_{L,1}W(Z)((1-D)S - m(Z)) \\
    &\ - Q_1(p)W(Z)\bigg[((1-D)S - m(Z)) + (DS-r(Z))F(Q_1(p)|1,Z)\\
    &\qquad \ \qquad + DS(\mathbbm{1}(Y \leq Q_1(p)) - F(Q_1(p)|1,Z)))\bigg] \\
    &\ + W(Z)DS(Y\mathbbm{1}(Y \leq Q_1(p)) - \psi_L(Z)),
\end{align*}}
while for the upper bound we have 

{\footnotesize \begin{align*}
    -[m_1 - m_0]\dt \beta_{U,1} 
    &= W(Z)(DS - r(Z))\psi_U(Z) + \beta_{U,1}W(Z)((1-D)S - m(Z)) \\
    &\ - Q_1(1-p)W(Z)\bigg[((1-D)S - m(Z)) + (DS-r(Z))[1-F(Q_1(1-p)|1,Z)]\\
    &\qquad \ \qquad + DS(\mathbbm{1}(Y > Q_1(1-p)) - [1-F(Q_1(1-p)|1,Z)]))\bigg] \\
    &\ + W(Z)DS(Y\mathbbm{1}(Y > Q_1(1-p)) - \psi_U(Z)).
\end{align*}}

\subsubsection{Conditional Case with Strong Sample Selection Monotonicity}
We now add covariates to the problem. Recall that, for any bound $\beta_B$, by definition, {\footnotesize \begin{align*}
    \beta_B &= \frac{\int \beta_B(x)\pi_{ac}(x)dP(x)}{\int \pi_{ac}(x)dP(x)} \\
    &= \frac{1}{\int \pi_{ac}(x)dP(x)}\left[\int \beta_{B,1}(x)\pi_{ac}(x)dP(x) +\int\beta_0(x)\pi_{ac}(x)dP(x)\right] \\
    &\equiv \frac{N_{B,1} + N_0}{\pi_{ac}}.
\end{align*}}

For simplicity and without loss of generality, we now treat \(X\) as taking values in a countable set.  
In general, the tangent space of the model is the closure of the linear span of 
scores that are constant on finitely many measurable subsets of the support of \(X\).  
Coarsening \(X\) by such a finite partition produces a discrete analogue \(X^{disc}\) for which 
the calculation below applies verbatim.  As the partition becomes finer, the functional 
\(\beta_B^{disc}(P)\) converge to \(\beta_B(P)\), and the corresponding pathwise derivatives converge 
as well, with Assumptions A.3--A.5 ensuring the passing limits through the derivative.  
Since scores are dense in the tangent space and the derivative operator is continuous, 
the influence functions obtained under the discrete-\(X\) approximation will also be valid for general \(X\).  In particular, we consider point-mass perturbations in the direction 
\(O = (YS,S,D,Z,X)\) using submodels of the form 
\(dP_t = (1-t)\,dP_0 + t\,\delta_o\), where $P_0 \in \mathcal{P}$ denotes the true law and 
\(\mathcal{P}\) the nonparametric model.  By standard results for conditional means 
(e.g., \cite{bickel1993efficient}, Ch.~3), these perturbations span the 
tangent space.
For any functional parameter, $\IF$ is then influence function operator that maps from the functional to its Riesz representer under the point mass perturbation submodel. We obtain 

{\footnotesize \begin{align*}
    \IF(\beta_B) &= \frac{1}{\pi_{ac}}\bigg[\IF(N_{B,1}) + \IF(N_0) - \beta_B \IF(\pi_{ac})\bigg].
\end{align*}}  

To derive the components first note the following auxiliary results {\footnotesize \begin{align*}
    \IF(\pi_{ac}(x)) &= - \IF(m(1,x) - m(0,x)) \\
    \IF(m(z,x)) &= \frac{\mathbbm{1}(Z=z,X=x)}{P(Z=z|X=x)P(X=x)}((1-D)S - m(z,x)) \\
    \Rightarrow \IF(\pi_{ac}(x))P(X=x) &= -\mathbbm{1}(X=x)W(Z,X)((1-D)S - m(Z,X)),
\end{align*}}
where $W(z,x) = z/P(Z=1|X=x) - (1-z)/(1-P(Z=1|X=x))$. The applicability of the operator to the conditional case follows from dominated convergence. Now note that 

{\footnotesize \begin{align*}
    \IF(\pi_{ac}) &= \IF\left(\sum_x\pi_{ac}(x)P(X=x)\right) \\
    &= \sum_x \IF(\pi_{ac}(x))P(X=x) + \pi_{ac}(X) - \pi_{ac} \\
    &= -\bigg[W(Z,X)((1-D)S - m(Z,X)) + [m(1,X) - m(0,X)]\bigg] - \pi_{ac}.
\end{align*}}
Next note that, for any $B$ and $j=0,1$, we obtain decomposition 

{\footnotesize \begin{align*}
    &\IF({N_j}) = \IF\left({\sum_x\beta_{B,j}(x)\pi_{ac}(x)P(X=x)}\right) \\
    &= \sum_x \IF(\beta_{B,j}(x))\pi_{ac}(x)P(X=x) + \sum_x \beta_{B,j}(x)\IF(\pi_{ac}(x))P(X=x) + \sum_x \beta_{B,j}(x)\pi_{ac}(x)\mathbbm{1}(X=x) - N_j 
\end{align*}}
and equivalently for $\beta_0 = \beta_{B,0}$. Now note that, as $\pi_{ac}(x) = -[m(1,x) - m(0,x)]$, we have that 

{\footnotesize \begin{align*}
\pi_{ac}(x)\IF(\beta_{B,1}(x))P(X=x)
&= \mathbbm{1}(X=x) \times \bigg[ -[m(1,x) - m(0,x)] \times\text{``IF of  $\beta_{B,1}$ conditional on $x$''} \bigg].
\end{align*}}
For example 

{\footnotesize \begin{align*}
    \pi_{ac}(x)&\IF(\beta_{0}(x))P(X=x) \\
    &= \mathbbm{1}(X=x)\bigg[W(Z,X)((1-D)SY - \mu(Z,X)) + \beta_0(X) W(Z,X)((1-D)S - m(Z,X))\bigg]
\end{align*}}
and analogously for $\beta_{B,1}$. Now also note that 

{\footnotesize \begin{align*}
    \beta_{B,j}(x)\IF(\pi_{ac}(x))P(X=x) &= - \mathbbm{1}(X=x)\beta_{B,j}(X)W(Z,X)((1-D)S - m(Z,X)).
\end{align*}}
Lastly, by definition, we can write 

{\footnotesize \begin{align*}
    -[m(1,x) - m(0,x)]\beta_0(x) &= \mu(1,x) - \mu(0,x), \\
    -[m(1,x) - m(0,x)]\beta_{B,1}(x) &= \psi_{B}(1,x)r(1,x) - \psi_{B}(0,x)r(0,x). 
\end{align*}}
Thus, combining expressions yields 

{\footnotesize \begin{align*}
    \IF(N_0)
    &= W(Z,X)((1-D)SY - \mu(Z,X)) + \beta_0(X)W(Z,X)((1-D)S - m(Z,X)) \\
    &\quad - \beta_0(X)W(Z,X)((1-D)S - m(Z,X)) - \beta_0(X)(m(1,X) - m(0,X)) - N_0 \\
    &= W(Z,X)((1-D)SY - \mu(Z,X)) + [\mu(1,X) - \mu(0,X)] - N_0
\end{align*}}
and 

{\footnotesize \begin{align*}
    \IF(N_{L,1})
    &= W(Z,X)(DS - r(Z,X))\psi_L(Z,X) \\
    &\quad - Q_1(p(X),X)W(Z,X)\bigg[((1-D)S - m(Z,X)) \\& \qquad  + (DS - r(Z,X))F(Q_1(p(X),X)|1,Z,X) \\
    &\qquad \ \qquad + DS(\mathbbm{1}(Y \leq Q_1(p(X),X)) - F(Q_1(p(X),X)|1,Z,X)) \bigg] \\
    &\quad + W(Z,X)DS[Y\mathbbm{1}(Y \leq Q_1(p(X),X)) - \psi_L(Z,X)] \\
    &\quad + \psi_L(1,X)r(1,X) - \psi_L(0,X)r(0,X) - N_{L,1}, 
\end{align*}}
as well as 

{\footnotesize \begin{align*}
     \IF(N_{U,1})
    &= W(Z,X)(DS - r(Z,X))\psi_U(Z,X) \\
    &\quad - Q_1(1-p(X),X)W(Z,X)\bigg[((1-D)S - m(Z,X)) \\& \qquad  + (DS - r(Z,X))[1-F(Q_1(1-p(X),X)|1,Z,X)] \\
    &\qquad \ \qquad + DS(\mathbbm{1}(Y > Q_1(1-p(X),X)) - [1-F(Q_1(1-p(X),X)|1,Z,X)]) \bigg] \\
    &\quad + W(Z,X)DS[Y\mathbbm{1}(Y > Q_1(1-p(X),X)) - \psi_U(Z,X)] \\
    &\quad + \psi_U(1,X)r(1,X) - \psi_U(0,X)r(0,X) - N_{U,1}, 
\end{align*}}
or simplified 

{\footnotesize \begin{align*}
    \IF(N_{L,1})
    &= - Q_1(p(X),X)W(Z,X)\bigg[((1-D)S - m(Z,X)) \\& \qquad + (DS - r(Z,X))F(Q_1(p(X),X)|1,Z,X) \\
    &\qquad \ \qquad + DS(\mathbbm{1}(Y \leq Q_1(p(X),X)) - F(Q_1(p(X),X)|1,Z,X)) \bigg] \\
    &\quad + W(Z,X)[DSY\mathbbm{1}(Y \leq Q_1(p(X),X)) - \psi_L(Z,X)r(Z,X)] \\
    &\quad + \psi_L(1,X)r(1,X) - \psi_L(0,X)r(0,X) - N_{L,1}, \\
     \IF(N_{U,1})
    &= - Q_1(1-p(X),X)W(Z,X)\bigg[((1-D)S - m(Z,X)) \\& \qquad  + (DS - r(Z,X))[1-F(Q_1(1-p(X),X)|1,Z,X)] \\
    &\qquad \ \qquad + DS(\mathbbm{1}(Y > Q_1(1-p(X),X)) - [1-F(Q_1(1-p(X),X)|1,Z,X)]) \bigg] \\
    &\quad + W(Z,X)[DSY\mathbbm{1}(Y > Q_1(1-p(X),X)) - \psi_U(Z,X)r(Z,X)] \\
    &\quad + \psi_U(1,X)r(1,X) - \psi_U(0,X)r(0,X) - N_{U,1}. 
\end{align*}}

\subsubsection{Conditional Case with Weak Sample Selection Monotonicity} \label{app:EIF-conditional1}
Recall, that by linearity, we have that the influence function of the bounds can be decomposed as influence functions of the following main elements 

{\footnotesize \begin{align*}
     \IF(\beta_B)
     &= \frac{1}{\pi_{ac}}\left[ \IF(N_0^+)  + \IF(N_0^-) + \IF(N_{B,1}^+) + \IF(N_{B,1}^-) - \beta_B \IF(\pi_{ac})\right].
\end{align*}}
Next note that by Assumption A.5, we have that all the components are pathwise differentiable with the classification error vanishing under the integral \citep{heiler2024treatmentevaluationintensiveextensive}. Using the analogous derivations as for the positive monotonicity case, we obtain 

{\footnotesize \begin{align*}
    \IF(\pi_{ac})
    &= -\mathbbm{1}^+\left[W(Z,X)((1-D)S - m(Z,X)) + m(1,X) - m(0,X)\right] \\
    &\quad + \mathbbm{1}^- \left[W(Z,X)(DS - r(Z,X)) + r(1,X) - r(0,X)\right] - \pi_{ac}
\end{align*}}
and as before, up to a multiplier $\mathbbm{1}^+$-based modification,

{\footnotesize \begin{align*}
    \IF(N_0^+) + N_0^+
    &= \mathbbm{1}^+(\IF(N_0) + N_0), \\
    \IF(N_{B,1}^+) + N_{B,1}^+
    &= \mathbbm{1}^+(\IF(N_{B,1}) + N_{B,1}),\\
\end{align*}}
and fully new components

{\footnotesize \begin{align*}
    \IF(N_0^-)  
    &= \mathbbm{1}^-\left\{W(Z,X)(DSY - \nu(Z,X)) + \nu(1,X)- \nu(0,X)\right\} - N_0^-,\\
    \IF(N_{L,1}^-)
    &=  \mathbbm{1}^-\bigg\{
    - Q_0(1-\dfrac{1}{p(X)},X)W(Z,X)\bigg[(DS-r(Z,X)) \\& \qquad  + ((1-D)S - m(Z,X))\left(1- F(Q_0(1-\dfrac{1}{p(X)},X)|0,Z,X)\right) \\ 
    &\quad \quad \quad + (1-D)S\left(\mathbbm{1}(Y \geq Q_0(1-\dfrac{1}{p(X)},X)) - \left[1- F(Q_0(1-\dfrac{1}{p(X)},X)|0,Z,X)\right]\right)\bigg] \\
    &\quad \quad + W(Z,X)\left((1-D)SY\mathbbm{1}(Y \geq Q_0(1-\dfrac{1}{p(X)},X)) - \psi_L^-(Z,X)m(Z,X)\right) \\
    &\quad \quad + \psi_L^-(1,x)m(1,x) - \psi_L^-(0,x)m(0,x)\bigg\} - N_{L,1}^-,\\
    \IF(N_{U,1}^-)
    &= \mathbbm{1}^-\bigg\{
    - Q_0(\dfrac{1}{p(X)},X)W(Z,X)\bigg[(DS-r(Z,X)) \\& \qquad  +  ((1-D)S - m(Z,X))\left(F(Q_0(\dfrac{1}{p(X)},X)|0,Z,X)\right) \\ 
    &\quad \quad \quad + (1-D)S\left(\mathbbm{1}(Y \leq Q_0(\dfrac{1}{p(X)},X)) - F(Q_0(\dfrac{1}{p(X)},X)|0,Z,X)\right)\bigg] \\
    &\quad \quad + W(Z,X)\left((1-D)SY\mathbbm{1}(Y \leq Q_0(\dfrac{1}{p(X)},X)) - \psi_U^-(Z,X)m(Z,X)\right) \\
    &\quad \quad + \psi_U^-(1,x)m(1,x) - \psi_U^-(0,x)m(0,x)\bigg\} - N_{U,1}^-.
\end{align*}}

\subsection{Bounding the Machine Learning Bias Remainder}
\subsubsection{Preliminaries, Structure of Influence Functions and Remainder}

For the following, at any $x$, let $F_1(\cdot| x)$ denote the true mixture CDF and $\hat F_1(\cdot| x)$ its estimator. 
For the proof, we denote shorthand $Q_1(x) = Q_1(p(x),x)$ and $\hat{Q}(x) = \hat{Q}(\hat{p}(x),x)$ obtained by inversion the mixture CDF. In particular, they are given by  the generalized left-continuous inverses at the threshold $p(x)$ and $\hat{p}(x)$ respectively. 

{\footnotesize \begin{align*}
    Q_1(x) 
= \inf\{y\in\mathbb{R}: F_1(y| x)\ge p(x)\}, \quad \hat Q_1(x) 
= \inf\{y\in\mathbb{R}: \hat F_1(y| x)\ge \hat{p}(x)\}.
\end{align*}}
We use the local neighborhood $\mathcal{N}_x:=[Q_1(x)-\xi,Q_1(x)+\xi]$. 
For generic $\Tilde{\eta}$, let 

{\footnotesize \begin{align*}
\Psi(\Tilde{\eta}):=E\big[\IF(\beta_B,\Tilde{\eta})\big]
\end{align*}}
denote the population moment of the rescaled/stabilized efficient influence function, and let \(\eta\) denote the true nuisance vector.  
The Neyman orthogonality of the EIF implies that for every admissible direction \(h\),

{\footnotesize \begin{align*}
D\Psi[\eta](h)
:=\lim_{t\to0}\frac{\Psi(\eta+t h)-\Psi(\eta)}{t}
=E\big[\partial_{\eta}\IF(\beta_B,\eta)[h]\big]
=0.
\end{align*}}
For directions \(h_1,h_2\) in the nuisance space define second directional derivative

{\footnotesize \begin{align*}
D^2\Psi[\eta](h_1,h_2)
:=\lim_{t\to0}\frac{D\Psi[\eta+ t h_2](h_1)-D\Psi[\eta](h_1)}{t},
\end{align*}}
whenever the limit exists. Under (A.3), (A.4), (A.5), and (A.6), \(D^2\Psi[\eta]\) exists in a neighborhood of \(\eta\). For \(\Delta\eta = \hat\eta - \eta\), define the path \(\eta_t=\eta+t\Delta\eta\), \(t\in[0,1]\).
A second-order expansion with integral remainder yields

{\footnotesize \begin{align*}
\Psi(\eta+\Delta\eta)-\Psi(\eta)
=\int_0^1 (1-t)\, D^2\Psi[\eta_t](\Delta\eta,\Delta\eta)\, dt,
\end{align*}}
since \(D\Psi[\eta]=0\). The second-order term is

{\footnotesize \begin{align*}
R_n(\beta_B) := \Psi(\hat\eta)-\Psi(\eta).
\end{align*}}
The cross-fitted stabilized estimating equation uses the EIF

{\footnotesize \begin{align*}
\IF(\beta_B)=
\frac{\IF(N_0^+)+\IF(N_0^-)+\IF(N_{B,1}^+)+\IF(N_{B,1}^-)-\beta_B\,\IF(\pi_{ac})}{E[\pi_{ac}(X)]}.
\end{align*}}
We now denote shorthand in what follows $\IF{\beta_B} = \IF(\beta_B,\eta)$ and $\hat\IF(\beta_B) = \IF(\beta_B,\hat{\eta})$ being the EIF with estimated nuisances.
By construction, our EIF are Neyman-orthogonal, first-order nuisance errors cancel and we have second order bias remainder

{\footnotesize \begin{align*}
R_n(\beta_B)=E[\hat\IF({\beta_B})]-E[\IF({\beta_B})].
\end{align*}}
By definition of the EIF, this can be decomposed as

{\footnotesize \begin{align*}
R_n(\beta_B)=R_n(N_0^+)+R_n(N_0^-)+R_n(N_{B,1}^+)+R_n(N_{B,1}^-)-\beta_B\,R_n(\pi_{ac}),
\end{align*}}
We now bound all these components. 
Let $O = (SY,Y,S,D,Z,X)$ denote the data.
Using the explicit form of the EIF, note that each term of 
\(D^2\Psi[\eta_t](\Delta\eta,\Delta\eta)\) can be written as a finite sum of expressions of the following form: 

{\footnotesize \begin{align*}
    \text{(A)} &\quad
E\big[\,U_t(O)\,\Delta a(O)\,\Delta b(O)\big], \\
\text{(B)} &\quad
E\big[\,U_t(O)\,\Delta a(O)\,\{G(Q(X)+\Delta Q(X))-G(Q(X))\}\big], \\
\text{(C)} & \quad
E\big[\,U_t(O)\,\Delta\varepsilon(O)\big],
\end{align*}}
where
\begin{itemize}
\item \(U_t(O)\in L^2\) are square integrable or bounded, components of the EIF (e.g., \(W(Z,X)[DS-r(Z,X)]\), \(W(Z,X)[(1-D)S-m(Z,X)]\) and similar).
\item \(\Delta a,\Delta b\) are primitive nuisance errors (one of \(\Delta\mu,\Delta\nu,\Delta m,\Delta r,\Delta e,\Delta F ,\Delta G_L, \Delta G_U\)).
\item \(G\in\{G_L,G_U\}\) are the trimmed expectations. 
\item \(Q_1\) is a true trimming or quantile boundary, and \(\Delta Q_1\) is the associated boundary error.
\item $\Delta\varepsilon(O)$ are of the classification error type arising from the lack of knowledge of the positive or negative monotonicity direction, i.e.,~proportional to $\mathbbm{1}(\lambda_0(x) < \lambda_1(x)) - \mathbbm{1}(\hat{\lambda}_0(x) < \hat{\lambda}_1(x))$ or with reversed sign.
\end{itemize}

By Cauchy--Schwarz and the local Lipschitz property of \(G\) near the boundary guaranteed by Assumption A.3 and A.4, we will show that

{\footnotesize \begin{align*}
|E[U_t\,\Delta a\,\Delta b]|
\le \|U_t\|_2\,\|\Delta a\|_2\,\|\Delta b\|_\infty,
\end{align*}}
and

{\footnotesize \begin{align*}
|E[U_t\,\Delta a\,\{G(Q_1+\Delta Q_1)-G(Q_1)\}]|
\le \|U_t\|_2\,\|\Delta a\|_2\, L_G\,\|\Delta Q_1\|_\infty.
\end{align*}}
Moreover, in the following we present a quantile bound

{\footnotesize \begin{align*}
\|\Delta Q_1\|_\infty \lesssim \|\Delta F\|_{\infty,\mathcal{N}}+\|\Delta m\|_\infty+\|\Delta r\|_\infty,
\end{align*}}
and analogously for the ``minus'' quantile. Therefore every \(\|\Delta Q_1\|_\infty\) can be bounded by a combination of the primitive nuisance errors. Overall, we show that the total remainder is bounded via mixed \(L^2\) and (localized) \(L^\infty\) norms of the primitive nuisance components, and classification error, i.e.,~of the shape 

{\footnotesize \begin{align*}
|R_n(\beta_B)|
=|\Psi(\hat\eta)-\Psi(\eta)|
\lesssim 
\sum_{(i,j)} C_{ij}\,\|\Delta\eta_i\|_{p_i}\,\|\Delta\eta_j\|_{p_j}
\;+\; R_{\mathrm{cls}}
\;+\;o_p(n^{-1/2}),
\end{align*}}
where $C_{ij}$ are finite constants only depending on the joint distribution of observables and  \(R_{\mathrm{cls}}\lesssim(\|\Delta m\|_\infty+\|\Delta r\|_\infty)^{\kappa+1}\) is a margin-dependent classification remainder.

\subsubsection{Remainder Decomposition}

By the stabilized score and orthogonality,

{\footnotesize \begin{align*}
R_n(\beta_B)=R_n(N_0^+)+R_n(N_0^-)+R_n(N_{B,1}^+)+R_n(N_{B,1}^-)-\beta_B\,R_n(\pi_{ac}),
\end{align*}}
where each block remainder is the quadratic (or classification) part arising from plugging $\hat\eta$ into their respective influence function $\IF({\cdot})$. Concretely:
\begin{itemize}
\item $R_n(N_0^\pm)$: products of errors for $\mu$, $\nu$ and $e$ plus classification errors from using estimated $\mathbbm{1}^{\pm}$.
\item $R_n(N_{B,1}^+)$: uses $\psi_{B}(z,x)$, $Q_1$, $F_1$, $m,r,e$ plus classification errors from using estimated $\mathbbm{1}^{+}$.
\item $R_n(N_{B,1}^-)$: uses $\psi_{B}^{-}(z,x)$, $Q_0$, $F_0$, $m,r,e$ plus classification errors from using estimated $\mathbbm{1}^{-}$.
\item $R_n(\pi_{ac})$: quadratic in errors for $m,r$ plus classification errors from using estimated $\mathbbm{1}^{\pm}$.
\end{itemize}
We now provide detailed bounds in terms of the nuisance estimation errors. For any derived nuisances, we reduce them to their primitives as given in Table \ref{tab:nuisance}.

\subsubsection{Bahadur Expansion and Quantile Deviation Bound}

\subsubsection*{Key Identity for the True CDF}
By Assumption A.4, for any $y\in\mathcal{N}_x$,

{\footnotesize\begin{align}
\label{eq:B-key}
\big|F_1(y| x)-F_1(Q_1(x)| x)\big|
\;=\;\Big|\int_{Q_1(x)}^{y} f_1(t| x)\,dt\Big|
\;\ge\; f_{\min}\,|y-Q_1(x)|.
\end{align}}

\subsubsection*{Triangle Inequality and Monotonicity}
By definition of quantiles, $F_1(Q_1(x)| x)=p(x)$ and $\hat F_1(\hat Q_1(x)| x)=\hat{p}(x)$. Hence

{\footnotesize\begin{align*}
F_1(\hat Q_1(x)| x) - F_1(Q_1(x)| x)
&= \big[F_1(\hat Q_1(x)| x)-\hat F_1(\hat Q_1(x)| x)\big] 
+ \big[\hat{p}(x)-p(x)\big] 
\end{align*}}
Taking absolute values and using the local sup-norm on $\mathcal{N}_x$,

{\footnotesize\begin{align}
\label{eq:B-localbound}
\big|F_1(\hat Q_1(x)| x) - F_1(Q_1(x)| x)\big|
\;\le\; \sup_{y\in\mathcal{N}_x}\big|\hat F_1(y| x)-F_1(y| x)\big|
\;+\; |\hat{p}(x)-p(x)|.
\end{align}}

\subsubsection*{Inversion}
Applying \eqref{eq:B-key} with $y=\hat Q_1(x)$ and combining with \eqref{eq:B-localbound} yields

{\footnotesize\begin{align*}
f_{\min}\,|\hat Q_1(x)-Q_1(x)|
\;\le\; \sup_{y\in\mathcal{N}_x}\big|\hat F_1(y| x)-F_1(y| x)\big|
\;+\; |\hat{p}(x)-p(x)|.
\end{align*}}
Thus,

{\footnotesize\begin{align}
\label{eq:B-ineq}
|\hat Q_1(x)-Q_1(x)|
\;\le\; \frac{1}{f_{\min}}\left(
\sup_{y\in\mathcal{N}_x}\big|\hat F_1(y| x)-F_1(y| x)\big|
\;+\; |\hat{p}(x)-p(x)|
\right).
\end{align}}

\subsubsection*{Neighborhood Localization}
Now if

{\footnotesize\begin{align*}
\sup_{y\in\mathcal{N}_x}\big|\hat F_1(y| x)-F_1(y| x)\big|+|\hat{p}(x)-p(x)| \overset{p}{\to}0
\end{align*}}
uniformly in $x$, then for large $n$ bound \eqref{eq:B-ineq} implies $|\hat Q_1(x)-Q_1(x)|\le \xi$ with probability $\to 1$, so $P(\hat Q_1(x)\in\mathcal{N}_x) \to 1$ and \eqref{eq:B-ineq} is valid.

\subsubsection*{Uniform Bound over $x$}
Recall the uniform local norm $\|\Delta F\|_{\infty,\mathcal{N}}:=\sup_x\sup_{y\in\mathcal{N}_x}|\hat F_1(y| x)-F_1(y| x)|$. If

{\footnotesize \begin{align*}
|\hat{p}(x)-p(x)| \;\lesssim \; \big(\|\Delta m\|_\infty+\|\Delta r\|_\infty\big),
\end{align*}}
then taking $\sup_x$ in \eqref{eq:B-ineq} yields

{\footnotesize\begin{align}
\label{eq:B-final}
\ \ \|\Delta Q_1\|_\infty
\;\lesssim\; \frac{1}{f_{\min}}\left(\|\Delta F\|_{\infty,\mathcal{N}} + (\|\Delta m\|_\infty+\|\Delta r\|_\infty)\right).\
\end{align}}
Note that no differentiability of $\hat F_1$ is required but only monotonicity of $\hat F_1$ and the structural lower bound $f_{\min}$ for $F_1$ near $Q_1(x)$.
Applying the same argument to the negative equivalents with $F_0$, $Q_0$, yields

{\footnotesize\begin{align}
\label{eq:B-minus-final}
\ \ \|\Delta Q_0\|_\infty
\;\lesssim\; \frac{1}{f_{\min}}\left(\|\Delta F\|_{\infty,\mathcal{N}} + \|\Delta m\|_\infty+\|\Delta r\|_\infty)\right).\
\end{align}}
as the norms are also uniform over the conditional arguments $z,d$.

\bigskip
\subsubsection{Trimmed Expectations Bound}
We now show the bound for the trimmed expectations for the positive monotonicity part of the lower bound case and then for the negative part. The upper bound follows analogously. We suppress the irrelevant components to ease notation of the trimmed expectation $\psi_L(z,x)=G_L(Q_1(p(X),X)| 1,z,x) = G_L(Q_1(x)|x)$ and its estimator $\hat\psi_L(z,x)=\hat G_L(\hat Q_1(\hat{p(X)},X)|1,z,x) = \hat{G}_L(\hat{Q}(x)|x)$.
Decompose

{\footnotesize\begin{align}
\label{eq:C-lower-decomp}
\Delta\psi_L(z,x)
&= \hat G_L(\hat Q_1(x)| x)- G_{L}(Q_1(x)| x) \nonumber\\
&=\underbrace{\big[\hat G_L(Q_1(x)| x)- G_{L}(Q_1(x)| x)\big]}_{(I)}
+\underbrace{\big[\hat G_L(\hat Q_1(x)| x)- \hat G_L(Q_1(x)| x)\big]}_{(II)}.
\end{align}}
Term $(I)$ is a pointwise error which can be bounded by the local worst-case

{\footnotesize \begin{align*}
|(I)|\le 
 \sup_{z}\sup_{x}\sup_{|y-Q_1(x)|\le\xi}|\hat G_L(y| x)-G_{L}(y| x)| = \|\Delta G\|_{\infty,\mathcal{N}}.
\end{align*}}
For $(II)$, adding and subtracting $\hat Q_1$ at the true threshold and applying the triangle inequality plus Assumption A.4 yields:

{\footnotesize\begin{align}
\label{eq:C-lower-II}
|(II)|
&\le |\hat G_L(\hat Q_1(x)| x)- G_{L}(\hat Q_1(x)| x)|
+ |G_{L}(\hat Q_1(x)| x)- G_{L}(Q_1(x)| x)| \nonumber\\
&\le \|\Delta G\|_{\infty,\mathcal{N}} + L_G\,|\hat Q_1(x)-Q_1(x)|.
\end{align}}
Combining \eqref{eq:C-lower-decomp}--\eqref{eq:C-lower-II} and taking the supremum over $(z,x)$ then yields

{\footnotesize\begin{align}
\label{eq:C-lower-final}
\ \ \|\Delta\psi_L\|_\infty
\;\le\; 2\,\|\Delta G\|_{\infty,\mathcal{N}} + L_G\,\|\Delta Q_1\|_\infty.\
\end{align}}
By the same argument with $G_U$ in place of $G_L$, we also obtain

{\footnotesize\begin{align}
\label{eq:C-upper-final}
\ \ \|\Delta\psi_U\|_\infty
\;\le\; 2\,\|\Delta  G_U\|_{\infty,\mathcal{N}} + L_G\,\|\Delta Q_1\|_\infty.\
\end{align}}
For the negative monotonicity side, we equivalently can derive

{\footnotesize\begin{align}
\ \ \|\Delta\psi_L^{-}\|_\infty
&\le 2\,\|\Delta G_L\|_{\infty,\mathcal{N}} + L_G^{}\,\|\Delta Q_0\|_\infty, \label{eq:C-minus-final} \\
\|\Delta\psi_U^{-}\|_\infty
&\le 2\,\|\Delta  G_U\|_{\infty,\mathcal{N}} + L_G^{}\,\|\Delta Q_0\|_\infty,\ \label{eq:C-minus-final2}
\end{align}}
as all the norms are also supremum over $d$.
Thus, using the quantile deviation bounds \eqref{eq:B-final}--\eqref{eq:B-minus-final} inside 
\eqref{eq:C-lower-final}--\eqref{eq:C-minus-final2} yields full control of 
$\Delta\psi$ solely in terms of primitives
$\Delta F,\Delta G$, and 
$\Delta m,\Delta r$ (through $\Delta p$), together with structural constants $f_{\min}$ and $L_G$.

\subsubsection{Margin-partition Remainder for $\pi_{ac}$}
We now control the second type of component that arises from boundary displacement/classification error in the density. The identical term arises in the classification error in $\mathbbm{1}^{\pm}$ that enters $N^0$ and $N_{B,1}^{\pm}$.
Define the contrast of the arguments and its estimator as

{\footnotesize \begin{align*}
\Gamma(x) := \lambda_0(x)-\lambda_1(x), \qquad \hat{\Gamma}(x) := \hat{\lambda}_0(x)-\hat{\lambda}_1(x).
\end{align*}}
Let the uniform estimation error be

{\footnotesize \begin{align*}
\varepsilon := \|\Delta \Gamma\|_\infty = \sup_x |\hat{\Gamma}(x)-\Gamma(x)|.
\end{align*}}
The true and estimated positive regions are then given by

{\footnotesize \begin{align*}
\mathcal{X}^+ := \{x : \Gamma(x)\le 0\}, \qquad 
\hat{\mathcal{X}}^+ := \{x : \hat{\Gamma}(x)\le 0\}.
\end{align*}}
Also denote their complements as $\mathcal{X}^{+c}$ and $\hat{\mathcal{X}}^{+c}$ respectively. 
The symmetric difference between the sets is then

{\footnotesize \begin{align*}
\hat{\mathcal{X}}^+ \oplus \mathcal{X}^+
=
\Big(\mathcal{X}^+ \cap \hat{\mathcal{X}}^{+c}\Big)
\cup
\Big(\mathcal{X}^{+c} \cap \hat{\mathcal{X}}^{+}\Big).
\end{align*}}

We now show that

{\footnotesize \begin{align*}
\hat{\mathcal{X}}^+ \oplus \mathcal{X}^+ \subseteq \{x : |\Gamma(x)| \le \varepsilon\}.
\end{align*}}

\textbf{Case 1:} Let \(x \in \mathcal{X}^+ \cap \hat{\mathcal{X}}^{+c}\).  
Then \(\Gamma(x)\le 0\) and \(\hat{\Gamma}(x)>0\). Hence

{\footnotesize \begin{align*}
0 < \hat{\Gamma}(x) = \Gamma(x) + (\hat{\Gamma}(x)-\Gamma(x))
\Rightarrow
-\Gamma(x) < \hat{\Gamma}(x)-\Gamma(x).
\end{align*}}
Since \(\Gamma(x)\le 0\), we have \(|\Gamma(x)| = -\Gamma(x)\), and therefore

{\footnotesize \begin{align*}
|\Gamma(x)| \le |\hat{\Gamma}(x)-\Gamma(x)| \le \|\Delta \Gamma\|_\infty = \varepsilon.
\end{align*}}

\textbf{Case 2:} Let \(x \in \mathcal{X}^{+c} \cap \hat{\mathcal{X}}^{+}\).  
Then \(\Gamma(x)>0\) and \(\hat{\Gamma}(x)\le 0\). Hence

{\footnotesize \begin{align*}
0 \ge \hat{\Gamma}(x) = \Gamma(x)+(\hat{\Gamma}(x)-\Gamma(x))
\Rightarrow
\Gamma(x) \le \Gamma(x)-\hat{\Gamma}(x) = |\hat{\Gamma}(x)-\Gamma(x)|.
\end{align*}}
Since \(\Gamma(x)>0\), we have \(|\Gamma(x)|=\Gamma(x)\), and therefore

{\footnotesize \begin{align*}
|\Gamma(x)| \le |\hat{\Gamma}(x)-\Gamma(x)| \le \|\Delta \Gamma\|_\infty = \varepsilon.
\end{align*}}

Therefore we obtain that the symmetric difference set is subset of a small epsilon set for the population difference

{\footnotesize \begin{align*}
\widehat{\mathcal{X}}^+ \oplus \mathcal{X}^+\ \subseteq\ \{x:\ |\lambda_0(x)-\lambda_1(x)|\le \epsilon\}
\end{align*}}
and thus, by margin assumption A.5

{\footnotesize \begin{align*}
P\big(\widehat{\mathcal{X}}^+ \oplus \mathcal{X}^+\big)\le C_M\,\epsilon^\kappa.
\end{align*}}
Overall, since each misclassification multiplies an $L^\infty$ integrand error $O(\epsilon)$, this yields overall classification remainder

{\footnotesize\begin{align}
R_{\mathrm{cls}}\ \lesssim \big(\|\Delta m\|_\infty+\|\Delta r\|_\infty\big)^{\kappa+1}. \label{eq:B-classificationRcls}
\end{align}}

\subsubsection{Bounding EIF Block Remainders and $R_n(\beta_B)$}
We now obtain the remainders for all objects of type $(A)$, $(B)$ and $(C)$. $R_n(N_0^+),R_n(N_0^-)$ contain no quantiles and are of classical ``augmented IPW'' (AIPW) form interacted with different mean functions $\mu$ and $\nu$ respectively and a classification indicator. Thus, their second order remainder is bounded by

{\footnotesize \begin{align*}
|R_n(N_0^+)|+|R_n(N_0^-)|\ \lesssim \ \|\Delta\mu\|_2\,\|\Delta e\|_2\ +\ \|\Delta\nu\|_2\,\|\Delta e\|_2\ +\ R_{\mathrm{cls}}.
\end{align*}}

For the positive side mixture component error $R_n(N_{B,1}^+)$ we have a product form with nuisances $\psi_B = G_B(Q_1),Q_1,F,m,r,e$ and  classification indicator. Thus, we obtain that 

{\footnotesize \begin{align*}
    |&R_n(N_{B,1}^+)| \\
    &\lesssim (||\Delta e ||_2 + ||\Delta m ||_2 + ||\Delta r ||_2) \times \bigg[||\Delta e ||_2 + ||\Delta m ||_2 + ||\Delta r ||_2 + ||\Delta G||_{\infty,\mathcal{N}} + ||\Delta F ||_{\infty,\mathcal{N}} + ||\Delta Q_1 ||_{\infty} \bigg] \\
    &\lesssim (||\Delta e ||_2 + ||\Delta m ||_2 + ||\Delta r ||_2) \times \bigg[||\Delta e ||_{\infty} + ||\Delta m ||_{\infty} + ||\Delta r ||_{\infty} + ||\Delta G ||_{\infty,\mathcal{N}} + ||\Delta F ||_{\infty,\mathcal{N}} \bigg],
\end{align*}}
as $||a||_2 \lesssim ||a||_{\infty}$  as well as the upper bound for $||\Delta Q_1||_{\infty}$ from \eqref{eq:B-final}. Now note that the same bound applies to $|R_n(N_{B,1}^-)|$. Moreover, for the always-selected complier density we have the mixture probability nuisances as well as classification in the remainder:

{\footnotesize \begin{align*}
|R_n(\pi_{ac})|
\ \lesssim \ \big(\|\Delta m\|_2^2+\|\Delta r\|_2^2\big)\ +\ R_{\mathrm{cls}}.
\end{align*}}
The rate for $R_{cls}$ is given in \eqref{eq:B-classificationRcls}. Combining this with the primitive rates for the remaining elements as well as the result from Section B.4 and the fact that $\beta_B$ is finite, then yields total bound the remainder 

{\footnotesize\begin{align}
    |R_n(\beta_B)| 
    &\lesssim
    \|\Delta\mu\|_2\,\|\Delta e\|_2\ + \|\Delta\nu\|_2\,\|\Delta e\|_2\ + (||\Delta m||_{\infty}  + ||\Delta r||_{\infty})^{\kappa + 1}  \label{eq_remainder_beta_all_final1} \\ &\qquad+(||\Delta e ||_2 + ||\Delta m ||_2 + ||\Delta r ||_2) \times  \bigg[||\Delta e ||_{\infty} + ||\Delta m ||_{\infty} + ||\Delta r ||_{\infty} + ||\Delta G||_{\infty,\mathcal{N}} + ||\Delta F ||_{\infty,\mathcal{N}} \bigg], 
    \notag
\end{align}}
with $B$ arbitrary. 



\subsection{Proof of Theorem \ref{thm_asyN1}}
We now denote $E_n[X] = \frac{1}{n}\sum_i^nX_i$ and $G_n[X] = \frac{1}{\sqrt{n}}\sum_i^n(X_i - E[X_i])$. Recall that the bounds are estimated via

 {\footnotesize \begin{align*}
    \hat{\beta}_B = \frac{\hat{N}_0^+ + \hat{N}_0^- + \hat{N}_{B,1}^+ + \hat{N}_{B,1}^-}{\hat\pi_{ac}},
\end{align*}}
where all estimators are obtained by solving their respective influence functions at the empirical solution with cross-fitted nuisances. Thus $\hat{N}_0$ is defined by solving $E_n[\IF(\hat{N_0},\hat{\eta})] = 0$ and equivalently for the remaining parameters. Thus, using standard Slutzky arguments, we obtain the following linearization 

 {\footnotesize \begin{align*}
    \sqrt{n}(\hat{\beta}_B - \beta_B) 
    &= G_n[\hat\IF(\beta_B)]\bigg(1 + O\bigg(\sup_{N \in \{N_0^+,N_0^-,N_1^+,N_1^-\}}|E_n[\hat{\IF}(N)] - E[\IF(N)]|\bigg)\bigg).
\end{align*}}
Further decomposing the leading term yields

 {\footnotesize \begin{align*}
    G_n[\hat\IF(\beta_B)] &= G_n[\IF(\beta_B)] + (G_n[\hat{\IF}(\beta_B)]  - G_n[\IF(\beta_B)]) \\
    &\quad + \sqrt{n}E[\hat\IF(\beta_B)-\IF(\beta_B)].
\end{align*}}

As nuisances are cross-fitted, the empirical process is $o_p(1)$ as long as the nuisances are consistent as implied by Assumption A.7. The big $O(\cdot)$ remainder will be $O_p(n^{-1/2})$ as their $\sqrt{n}$-analogues are already $O_p(1)$ as implied by the results in Appendix B. Thus, overall, we obtain that

{\footnotesize \begin{align*}
    \sqrt{n}(\hat{\beta}_B - \beta_B) &= G_n[\IF(\beta_B)] + o_p(1) \overset{d}{\rightarrow} \mathcal{N}(0,E[\IF(\beta_B)^2]).
\end{align*}}
Here $E[\IF(\beta_B)^2]$ is the semiparametric efficiency bound for $\beta_B$ as we have explicitly used the Riesz representer of the pathwise derivative under the nonparametric model and regularity conditions, see \cite{bickel1993efficient} or \cite{kennedy2024semiparametric}. Moreover, since the leading expression of each bound is asymptotically linear, we also obtain joint convergence of $\sqrt{n}(\hat{\beta} - \beta)$ by Assumption A.2 and the Cramer--Wold device.

\subsection{Proof of Proposition \ref{prop:nesting-of-IF}}
We now show the equivalence to existing results under additional restrictions. We focus on the lower bound under strong sample selection monotonicity, but the analogous derivations apply to the weak sample selection monotonicity case and the upper bound as well via simple sign/partition adjustment as in Section \ref{app:EIF-conditional1}. Statements about equality of random variables are almost surely. Note that, in this case we denote $N_B^+ = N_B$ and have $N_B^- = 0$. 
\subsubsection{(i) $Z = D$: Lee Bounds}
Denote $s(d,x) = P(S=1|D=d,X=x)$ and {\footnotesize\begin{align*}
    \beta_{1,d}(x,u) = E[Y|S=1,D=d,X=x,Y\leq Q_1(u,x)], \\
    \beta_{0,d}(x,u) = E[Y|S=1,D=d,X=x,Y\geq Q_1(u,x)],
\end{align*}}
which yields

{\footnotesize\begin{align*}
 \beta_{1,d}(x,1) = \beta_{0,d}(x,0) =  E[Y|S=1,D=d,X=x].
\end{align*}}
When $Z=D$  (perfect compliance), the nuisances simplify as follows 

{\footnotesize \begin{align*}
    P(Z=1|X) &= P(D=1|X) = e(x), \\
    r(z,x) &= E[DS|Z=z,X=x] = \begin{cases}
        0 &\text{ if } Z = D = 0 \\
        s(1,x) &\text{ if } Z = D = 1 
    \end{cases}, \\
    m(z,x) &= E[S|Z=z,X=x] - E[DS|Z=z,X=x] = \begin{cases}
        s(0,x) &\text{ if } Z = D = 0 \\
        0 &\text{ if } Z = D = 1 
    \end{cases}, \\
    \mu(z,x) &= E[Y(1-D)S|Z=z,X=x] = \begin{cases}
        0 &\text{ if } Z = D = 1 \\
        E[Y|S=1,D=0,X=x]s(0,x) &\text{ if } Z = D = 0 
    \end{cases},\\
    F(y|d,z,x) &= P(Y\leq y|D=d,S=1,Z=z,X=x) \\ &= \begin{cases}
        0 &\text{ if } Z \neq D  \\
        P(Y \leq y|DS=1,X=x) &\text{ if } Z = D 
    \end{cases}.
\end{align*}}
These yield the following simplifications 

{\footnotesize \begin{align*}
    F_{Y_1|S_1 = 1,D_1 > D_0,X}(y|x)
    &= \frac{F(y|1,1,x)r(1,x) - F(y|1,0,x)r(0,x)}{r(1,x) - r(0,x)} \\
    &= F(y|1,1,x) \\
    &= P(Y\leq y|DS=1,Z=1,x) \\
    &=: F(y|1,x), \\
    p(x) 
    &= -\frac{(m(1,X) - m(0,X))}{(r(1,X) - r(0,X))} \\
    &= \frac{s(0,x)}{s(1,x)}, \\
    W(z,x) &= W(d,x) = \frac{d}{P(D=1|X=x)} - \frac{1-d}{1-P(D=1|X=x)}. \\
\end{align*}}
Lastly, for the trimmed mean we have 

{\footnotesize \begin{align*}
    &\psi_L(z,x)
    = E[Y\mathbbm{1}(Y \leq Q_1(p(X),X))|DS=1,Z=z,X=x] \\
    &= \begin{cases}
        0 &\text{ if } Z = D = 0 \\
        E[Y|Y\leq Q_1(p(X),X),DS=1,X=x]P(Y\leq Q_1(p(X),X)|DS=1,X=x) &\text{ if } Z = D = 1 
    \end{cases}\\ & = \begin{cases}
        0 &\text{ if } Z = D = 0 \\
        \beta_{1,1}(x,p(x)) p(x) &\text{ if } Z = D = 1 
    \end{cases}. 
\end{align*}}

We now look at the three components of the influence functions for $\beta_L$. First, we start with $\pi_{ac}$. Noting that $Z = D$ we have that 

{\footnotesize \begin{align*}
    \IF(\pi_{ac}) + \pi_{ac}
    &= -\left[W(Z,X)((1-D)S - m(Z,X)) + m(1,X) - m(0,X)\right] \\
    &= -\left[\left(\frac{D}{e(X)} - \frac{1-D}{1-e(X)}\right)((1-D)S - m(D,X)) + m(1,X) - m(0,X) \right] \\
    &= \frac{(1-D)}{(1-e(X))}(S - s(0,x)) + s(0,x),
\end{align*}}
as $m(1,x) = 0$ and thus $Dm(D,X) = 0$. This is exactly the denominator in \cite{heiler2024treatmentevaluationintensiveextensive}, Table 6.1. We now consider the numerator given by $N_1 + N_0$. 

Now for the $N_0$ component we exploit that $\mu(1,x) = 0$ and thus $D\mu(D,X) = 0$. This yields 

{\footnotesize \begin{align*}
    \IF(N_0) + N_0
    &= W(Z,X)((1-D)SY - \mu(Z,X)) + \mu(1,X) - \mu(0,X) \\
    &= \left(\frac{D}{e(X)} - \frac{1-D}{1-e(X)}\right)((1-D)SY - \mu(D,X)) + \mu(1,X) - \mu(0,X) \\
    &= -\left[\frac{(1-D)}{(1-e(X))}SY - \mu(0,X)\left(1 - \frac{(1-D)}{(1-e(X))}\right)\right] \\
    &= -\left[\frac{(1-D)}{(1-e(X))}SY - s(0,x)\beta_{0,0}(x,0)\left(1 - \frac{(1-D)}{(1-e(X))}\right)\right].
\end{align*}}
Now for $N_1$ we first denote the following auxiliary results: \paragraph{\textbf{(N1.i)}}{\footnotesize \begin{align*}
    W(Z,X)(DS - r(Z,X))\psi_L(Z,X)
    &= \frac{D}{e(X)}(S - s(1,X))\beta_{1,1}(X,p(X))p(X).
\end{align*}}
\paragraph{\textbf{(N1.ii)}}{\footnotesize \begin{align*}
    W(Z,X)&DS(Y\mathbbm{1}(Y\leq Q_1(p(X),X)) - \psi_L(Z,X)) \\
    &= \frac{DS}{e(X)}Y\mathbbm{1}(Y\leq Q_1(p(X),X)) - \frac{DS}{e(X)}\beta_{1,1}(X,p(X))p(X).
\end{align*}} 
\paragraph{\textbf{(N1.iii)}}{\footnotesize \begin{align*}
    \psi_L(1,X)r(1,X) - \psi_L(0,X)r(0,X)
    &= \psi_L(1,X)r(1,X) \\
    &= \beta_1(X,p(X))p(X)s(1,X) \\
    &= \beta_1(X,p(X))s(0,X). 
\end{align*}} 
\paragraph{\textbf{(N1.iv)}}
Note that as when $D=Z=1$ we have $F_1(y|1,x) = P(Y\leq y|DS=1,Z=1,x) =: F(y|1,x) $ and same for the corresponding quantiles. Thus

{\footnotesize \begin{align*}
    -&Q_1(p(X),X)W(Z,X)\bigg[((1-D)S - m(z,X)) \\ &\qquad + (DS - r(Z,X))F(Q_1(p(X),X)|Z,X) + DS(\mathbbm{1}(Y \leq Q_1(p(X),X)) - F(Q_1(p(X),X)|Z,X))\bigg] \\
    &= -Q_1(p(X),X)\bigg(-\left[\frac{(1-D)}{1-e(X)}(S- s(0,X))\right] + \frac{D}{e(X)}(S - s(1,X))F_1(Q_1(p(X),X)|1,X) \\ &\qquad + \frac{DS}{e(X)}(\mathbbm{1}(Y \leq Q_1(p(X),X) - F_1(Q_1(p(X),X)|1,X)) \bigg) \\
    &= -Q_1(p(X),X)\left(\frac{DS}{e(X)}(\mathbbm{1}(Y \leq Q_1(p(X),X) - p(X)) \right) \\
    &\quad +  Q_1(p(X),X)\left(\frac{(1-D}{1-e(X)}(S - s(0,X)) - \frac{D}{e(X)}(S - s(1,X))p(X) \right).
\end{align*}}
Summing up  $(N1.i)$ to  $(N1.iv)$ and adding the results of $N_0$ then yields  

{\footnotesize \begin{align*}
    \IF(N_1) + N_1 +  \IF(N_0) + N_0
    &= \frac{DS}{e(X)} Y 1\{Y \le Q_1(p(X),X)\} 
-\frac{(1-D)S}{1-e(X)}Y \\
&- \frac{DS}{e(X)}Q_1(p(X),X)
\left[1\{Y \le Q_1(p(X),X)\}-p(X) \right] \\
&+Q_1(p(X),X)\left[\frac{1-D}{1-e(X)}(S-s(0,X))- p(X) \frac{D}{e(X)}(S-s(1, X)) \right] \\
&+ s(0,X)\left[\beta_{1,1}(X,p(X))\left(1-\frac{D}{e(X)}\right) - \beta_{0,0}(X,0)\left(1-\frac{1-D}{1-e(X)}\right)\right]
\end{align*}},

which is exactly the corresponding component of the influence function proposed in \cite{heiler2024treatmentevaluationintensiveextensive}, Table 6.1., see also \cite{semenova2025generalized}, Section 6.6.

\subsubsection{(ii) $S = 1$: LATE}
If $S=1$, we have the following nuisance simplifications:    

{\footnotesize \begin{align*}
        r(z,x) &= E[DS|Z=z,X=x] = E[D|Z=z,X=x], \\
        m(z,x) &= E[(1-D)S|Z=z,X=x] = E[(1-D)|Z=z,X=x].
\end{align*}}
Thus we have that $r(z,x) = 1- m(z,x)$ and thus $-(m(1,x) - m(0,x)) = r(1,x) - r(0,x)$ which implies that $p(x) = 1$. Moreover, we have that $F(y|d,z,x) = P(Y\leq y|D=1,Z=z,X=x)$ (without $S$).
\paragraph{Auxiliary Result: CDFs at $Q_1(1,X)$} 
 We now show that when the trimming threshold $p =1$, the conditional cdfs must also be equal to one. We omit $x$ for simplicity. In particular,  \(F(y| z)\) is a continuous CDF  with $
\lim_{y\to-\infty}F(y| z)=0$ and $\lim_{y\to+\infty}F(y| z)=1$.
Define generic

{\footnotesize \begin{align*}
b_z \;=\;\inf\{\,y : F(y| z)=1\}, 
\qquad
F_1(y) \;=\;\frac{r_1\,F(y|1)\;-\;r_0\,F(y|0)}{r_1 - r_0},
\qquad
Q_1 \;=\;\inf\{\,y : F_1(y)=1\}.
\end{align*}}
Then for \(y<\max(b_0,b_1)\), one has \(F_1(y)<1\) and hence \(Q_1\ge\max(b_0,b_1)\).  Conversely, for \(y\ge\max(b_0,b_1)\), \(F(y|0)=F(y|1)=1\) implies \(F_1(y)=1\) and thus \(Q_1\le\max(b_0,b_1)\).  Therefore$
Q_1=\max(b_0,b_1)$. Lastly, since \(Q_1\ge b_z\), continuity gives $
F(Q_1| z)=1 \quad\text{for }z=0,1$.
In the original notation this means that almost surely

{\footnotesize \begin{align*}
F(Q_1(p(X),X)|Z,X) = 1.
\end{align*}}
We now move the components of the influence function. First regarding $\pi_{ac}$ we obtain the following simplification

{\footnotesize \begin{align*}
    \IF(\pi_{ac}) + \pi_{ac}
    &= -\left[W(Z,X)((1-D) - m(Z,X)) + m(1,X) - m(0,X)\right] \\
    &= W(Z,X)(D - r(Z,X)) + r(1,X) - r(0,X) \\
    &= \frac{Z}{P(Z=1|X)}(D - r(1,X)) - \frac{(1-Z)}{1-P(Z=1|X)}(D - r(0,X)) + r(1,X) - r(0,X),
\end{align*}}
which is exactly the classic 
uncentered AIPW EIF of the share of compliers. Now for $N_0$, we obtain 

{\footnotesize \begin{align*}
    \IF(N_0) + N_0 
    &= W(Z,X)((1-D)Y - \mu(Z,X)) + \mu(1,X) - \mu(0,X),
\end{align*}}  
while for $N_1$ we first note that

{\footnotesize \begin{align*}
    ((1-D)S& - m(Z,X)) + (DS - r(Z,X))F(Q_1(p(X),X)|Z,X) \\ &\quad + DS(\mathbbm{1}(Y \leq Q_1(p(X),X)) - F(Q_1(p(X),X)|Z,X)) \\
    &= ((1-D) - m(Z,X)) + (D - r(Z,X) + 0 \\
    &= 1- m(Z,X) - r(Z,X) \\
    &= 0.
\end{align*}}
Thus we have that

{\footnotesize \begin{align*}
    \IF(N_1) + N_1
    &= W(Z,X)(D - r(Z,X))\psi_L(Z,X) + W(Z,X)D[Y - \psi_L(Z,X)] \\ &\quad + \psi_L(1,X)r(1,X) - \psi_L(0,X)r(0,X) \\
    &= W(Z,X)[DY - \psi_L(Z,X)r(Z,X)] + \psi_L(1,X)r(1,X) - \psi_L(0,X)r(0,X).
\end{align*}}
Now note that, by definition 

{\footnotesize \begin{align*}
    \psi(z,x)r(z,x) + \mu(z,x)
    &= E[Y|D=1,Z=z,X=x]E[D|Z=z,X] + E[Y(1-D)|Z=z,X=x] \\
    &= E[YD|Z=z,X=x] + E[Y(1-D)|Z=z,X=x] \\
    &= E[Y|Z=z,X=x].
\end{align*}}
Thus, for the total numerator we obtain

{\footnotesize \begin{align*}
    \IF&(N_1)  + N_1 + \IF(N_0) + N_0 \\
    &= W(Z,X)[DY - \psi_L(Z,X)r(Z,X)] + \psi_L(1,X)r(1,X) - \psi_L(0,X)r(0,X) \\
    &\quad + W(Z,X)[(1-D)Y - \mu(Z,X)] + \mu(1,X) - \mu(0,X) \\
    &= W(Z,X)Y + [\psi_L(1,X)r(1,X) + \mu(1,X)]\left(1 - \frac{Z}{P(Z=1|X)}\right) \\ &\qquad - [\psi_L(0,X)r(0,X) + \mu(0,X)]\left(1 - \frac{(1-Z)}{1-P(Z=1|X)}\right) \\
    &= W(Z,X)Y + E[Y|Z=1,X]\left(1 - \frac{Z}{P(Z=1|X)}\right) - E[Y|Z=0,X]\left(1 - \frac{(1-Z)}{1-P(Z=1|X)}\right) \\
    &= \frac{Z}{P(Z=1|X)}(Y - E[Y|Z=1,X]) - \frac{(1-Z)}{1-P(Z=1|X)}(Y - E[Y|Z=0,X]) \\ &\qquad  + E[Y|Z=1,X] - E[Y|Z=0,X],
\end{align*}}
which is exactly the EIF for the numerator/reduced form component of the LATE. Combining results with the denominator yield the well-known EIF for the LATE \citep{frolich2007nonparametric,chernozhukov2018double,heiler2022efficient}.

\subsubsection{(iii) $Z = D, S = 1$: ATE}
In this case we have all the $S=1$ simplifications as well as $Z=D$ and thus $ZD = D$, $(1-Z)D = 0$, and  

{\footnotesize \begin{align*}
    r(0,x) = m(1,x) = 0, \\
    r(1,x) = m(0,x) = 1.
\end{align*}}
This yields 

{\footnotesize \begin{align*}
    \IF(\pi_{ac}) &= \frac{D}{P(D=1|X)}(D - r(1,X)) - \frac{(1-D)}{1-P(D=1|X)}(D - r(0,X)) + r(1,X) - r(0,X) \\
    &= 1
\end{align*}}
and 

{\footnotesize \begin{align*}
    \IF(N_1) &+ N_1 + \IF(N_0) + N_0 \\
    &= \frac{D}{P(D=1|X)}(Y - E[Y|D=1,X]) - \frac{(1-D)}{1-P(D=1|X)}(Y - E[Y|D=0,X]) \\ &\qquad  + E[Y|D=1,X] - E[Y|D=0,X],
\end{align*}}
which is the usual efficient AIPW influence function for the ATE \citep{hahn1998role}.

\section{Monte Carlo Simulations} \label{sec:simulation}
\subsection{Design}
In this section we report a simple Monte Carlo design to assess coverage and power properties of the suggested confidence intervals in Section \ref{sec:estimation1} in finite samples. The true DGP and its parameterization are in Equation \eqref{eq:MC_design1} and Table \ref{tab:MC_parameters1}. It is a single index treatment response model with normal errors and an additional single index sample selection step similar to the designs in \cite{heckman2005structural}, \cite{heiler2022efficient} and \cite{bartalotti2023identifying}:

{\footnotesize
\begin{align} 
\begin{aligned}
    Z &\sim \mathrm{Binom}(p_z) \\
    X_1,\dots,X_J &\sim \mathrm{Unif}(0,1) \\
    Y_1 &= \mu_y(1,X) + \varepsilon_1 \\
    Y_0 &= \varepsilon_0 \\
    D_z &= \mathbbm{1}(U_d \leq \mu_d(z,x))\\
    S_d &= \mathbbm{1}(U_s \leq \mu_s(d,x))
\end{aligned}
\quad
\begin{aligned}
    \begin{pmatrix}
        \varepsilon_1 \\
        \varepsilon_0 \\
        U_d \\
        U_s
    \end{pmatrix}
    &\sim \mathcal{N}\left(
    \begin{pmatrix}
        0 \\ 0 \\ 0 \\ 0
    \end{pmatrix},
    \sigma^2
    \begin{pmatrix}
        1 & 0 & 0 & \rho_{1s} \\
        0 & 1 & \rho_{0d} & 0 \\
        0 & \rho_{0d} & 1 & 0 \\
        \rho_{1s} & 0 & 0 & 1
    \end{pmatrix}
    \right) \\ & \\ & \\
\end{aligned} \label{eq:MC_design1}
\end{align}
}

\begin{table}[!h]  \footnotesize
    \caption{Monte Carlo: Design Parameters}
    \label{tab:MC_parameters1}
    \centering
    \begin{tabular}{c|c}
    \hline 
      Parameter/Function   &  Value \\ \hline
       $p_z$  &  0.5 \\
       $J$ & 25 \\
       $\sigma^2$ & 1 \\
       $\rho_{1s}$ & 0.9999 \\
       $\rho_{0d}$ & 0.5 \\
       $\mu_y(1,x)$ & $f_x/7.205 + \mathbbm{1}(x_6 > 0.5) + \mathbbm{1}(x_7 > 0.5)$ \\
       $\mu_d(z,x)$ & $f_x/7.205 -1 + 2z$ \\
       $\mu_s(d,x)$ & $f_x - 0.5 + 2d$ \\
       $f_x$ & $\sum_{j=1}^5 x_j - 2.5$ \\ \hline 
    \end{tabular}
\end{table}
Let $\phi(\cdot)$ and $\Phi(\cdot)$ denote the standard normal density and cumulative distribution function respectively. 
The design implies that the conditional SLATE is given 

{\footnotesize\begin{align*}
    \theta_{SLATE}(x) &= E[Y_1 - Y_0|S_0 = 1, D_1 > D_0,X=x] \\
    &= \mu_y(1,x) - \rho_{1,s}\frac{\phi(\mu_s(0,x))}{\Phi(\mu_s(0,x))} + \rho_{0,d}\frac{(\phi(\mu_d(1,x)) - \phi(\mu_d(0,x)))}{(\Phi(\mu_d(1,x)) - \Phi(\mu_d(0,x)))}.
\end{align*}}
Moreover, the conditional share of always-selected compliers is 

{\footnotesize\begin{align*}
    \pi_{ac}(x) &= P(S_0 = 1, D_1 > D_0|X=x) \\
    &= \Phi(\mu_s(0,x))(\Phi(\mu_d(1,x)) - \Phi(\mu_d(0,x))).
\end{align*}}
$\theta_{SLATE}$ can then be obtained as

{\footnotesize
 \begin{align}
     \theta_{SLATE} = \frac{\int \beta(x)\pi_{ac}(x)dP(x)}{\int \pi_{ac}(x)dP(x)}.
 \end{align}}
 
For the Monte Carlo Simulations, our estimation follows the algorithm in Section \ref{sec:estimation1}. We estimate all nuisances with generalized random forests with default parameters, including the functional parameters which are estimated on a large grid. Strong sample selection monotonicity is not imposed, i.e.,~there is risk of missclassification of the response types. Before inversion of the conditional CDFs, we use isotonic regression for post-processing. The design is relatively sparse in the underlying nuisances such that convergence requirements as in Section \ref{sec:largesample1} are credible. For evaluation, we report the average number of times where $\theta_{SLATE}$ is not contained in the confidence interval as a function of deviation from the true null with and without cross-fitting. At zero, this naturally corresponds to the size of the associated test. 

\begin{figure}[!h]
    \caption{Monte Carlo Study: Power Curves}
    \label{fig:mc1_power1}
    \centering
    \includegraphics[width=0.45\linewidth]{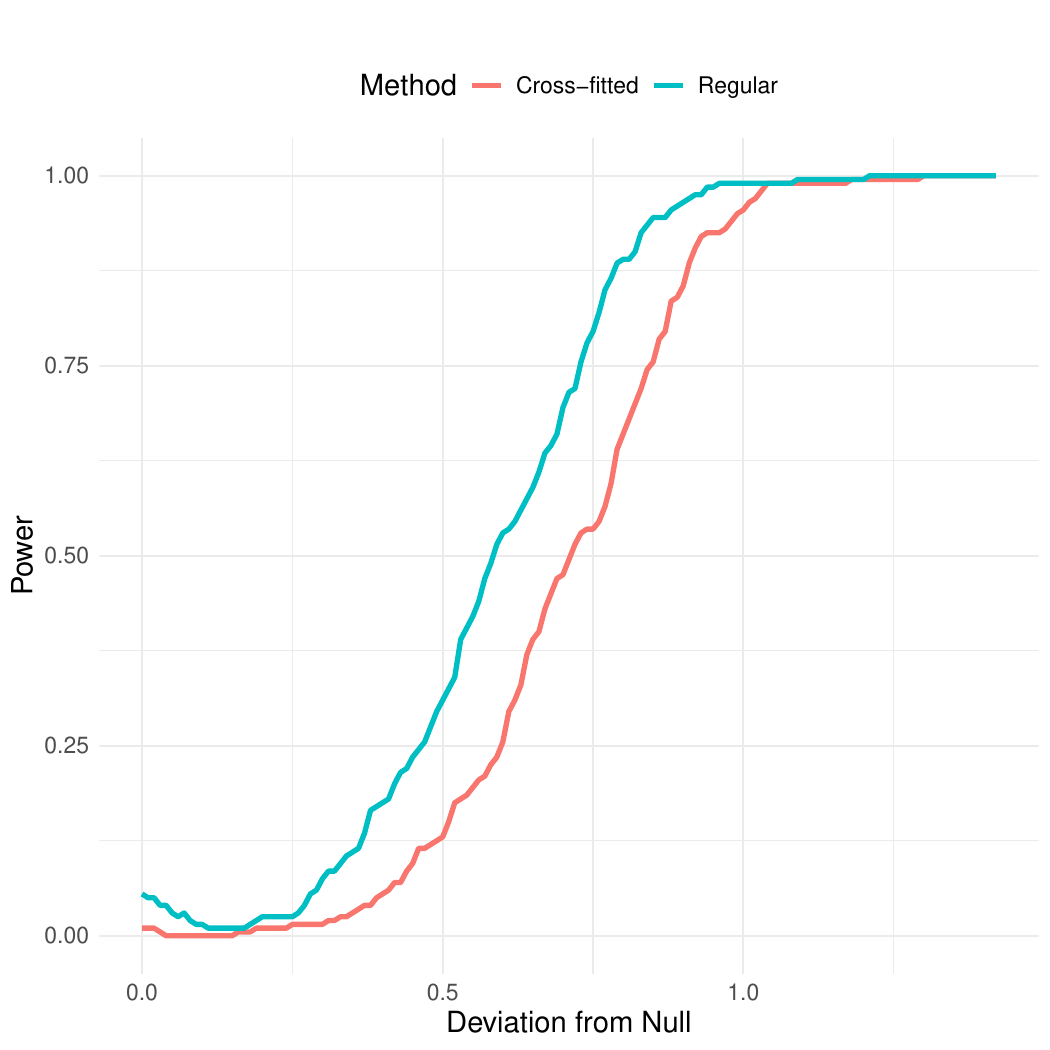}  \includegraphics[width=0.45\linewidth]{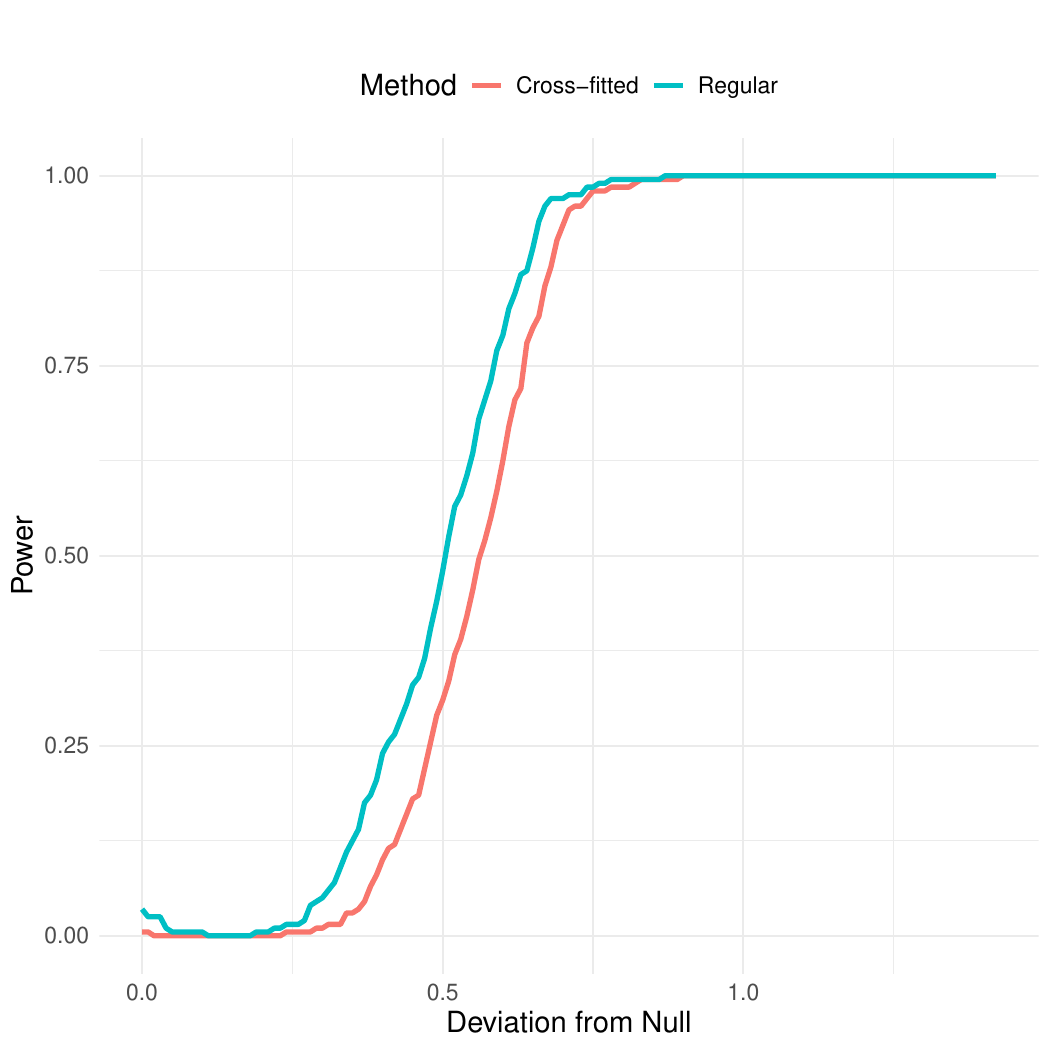}
    \begin{justify}
      Power curves as a function of deviation from the null for Sharp-DML bounds using cross-fitting according to Algorithm \ref{alg:estimation}. Regular is equivalent but without cross-fitting. $n = 2000$ and $n = 5000$ observations and $M = 200$ replications.
    \end{justify}
\end{figure}

Figure \ref{fig:mc1_power1} contains the power curves for both $n=2000$ and $n=5000$. Given the design, this corresponds to an effective sample size of approximately $600$ and $1500$ always-selected compliers respectively. Cross-fitting seems to be required to have the minimal rejection area close to the true null. The confidence intervals have conservative coverage at and close to the null. This is expected as the effect is not at the boundary of the identified set in line with, e.g., \cite{stoye2020simple} or \cite{heiler2024heterogeneous}. For larger deviations power quickly increases. Overall, results suggest that the theory in Section \ref{sec:largesample1} approximates the finite sample distribution well.

\section{Empirical Study II: Oregon Health Insurance Experiment} \label{sec:OHIE}
\subsection{Data and Methods}

In this section, we study the effects of Medicaid coverage on healthcare utilization in the Oregon Health Insurance Experiment (OHIE), following \cite{finkelstein2012oregon}. In 2008, Oregon expanded Medicaid for low-income uninsured adults through a lottery because available funding was insufficient to cover all eligible applicants. Individuals selected by lottery obtained the opportunity to apply for coverage for themselves and household members, but actual enrollment required submitting paperwork and verifying eligibility, so take-up was incomplete.

We use the publicly available OHIE survey data analyzed by \cite{finkelstein2012oregon}. The survey covers individuals from both lottery-selected and non-selected households and includes demographic information collected at sign-up. Our treatment is ever enrolling in Medicaid during the study period, while the instrument is the lottery-based assignment. We consider three outcomes: the number of outpatient visits, the number of prescription drugs, and log annual total medical expenditure. For outpatient visits and prescription drugs, outcomes are only observed for individuals with positive utilization, 
while for log expenditure the outcome is non-missing only when expenditure is positive. We additionally use pre-randomization covariates including gender, age, race and ethnicity, education, survey-wave fixed effects, and household-size fixed effects. The list of covariates along with sample summary statistics are provided in Table \ref{tab:ohie_balance}.

The data are well suited to our framework. Lottery selection provides a credible source of exogenous assignment, while exclusion is plausible since effects on utilization should operate through actual Medicaid coverage rather than lottery selection itself. As in the JC application, noncompliance is substantial because only a minority of lottery winners ultimately enrolled, so assignment and treatment differ materially. Depending on the outcome, analysis samples range from 16,868 to 22,064 observations after restricting to units with non-missing treatment, instrument, outcome, and covariates.

We evaluate always-selected complier effects $\theta_{SLATE}$ for all three outcomes using (i) \cite{chen2015bounds} bounds and (ii) sharp-basic bounds - both assuming strong sample selection monotonicity - as well as (iii) sharp DML bounds with covariates under weak sample selection monotonicity. Implementation of (i) follows \cite{chen2015bounds}. (ii) uses simple sample analogues of \eqref{eq:betaL} and \eqref{eq:betaU}. (iii) is obtained via the procedure in Section \ref{sec:estimation1}, with the same implementation details as in the JC application.

\subsection{Results: Bounds and Shares}

\begin{table}[htbp]
\centering
\caption{Bounds for the Effect of Health Insurance on Healthcare Utilization}
\label{tab:ohie_bounds}
\footnotesize
\begin{tabular}{@{}l ccc@{}}
\toprule
& CF & Sharp-basic & Sharp-DML \\
\midrule
\multicolumn{4}{l}{\textit{Outpatient visits ($N=22,064$)}} \\
\quad Bounds        & $[-1.447,\ 1.520]$    & $[-1.129,\ 0.915]$   & $[-0.821,\ 1.597]$    \\
\quad Standard Errors            & ---                   & $(0.247,\ 0.333)$    & $(0.299,\ 0.374)$     \\
\quad 95\% CI       & $[-1.879,\ 2.157]$    & $[-1.533,\ 1.460]$   & $[-1.311,\ 2.211]$    \\[3pt]
\quad Share of positive sample selection $^{a}$ &\multicolumn{2}{c}{$1 \text{ (assumed)}$}& \multicolumn{1}{l}{$0.957\ (0.004)$} \\
\quad Share of $ac$ $^{b}$ & \multicolumn{2}{c} {$0.147\ (0.007)$} &  \multicolumn{1}{l}{$0.145\ (0.007)$} \\[4pt]
\multicolumn{4}{l}{\textit{Prescription drugs ($N=17,223$)}} \\
\quad Bounds        & $[-0.832,\ 0.125]$    & $[-0.535,\ -0.060]$  & $[-0.295,\ 0.113]$    \\
\quad Standard Errors            & ---                   & $(0.289,\ 0.244)$    & $(0.257,\ 0.113)$     \\
\quad 95\% CI       & $[-1.565,\ 0.574]$    & $[-1.010,\ 0.340]$   & $[-0.717,\ 0.518]$    \\[3pt]
\quad Share of positive sample selection $^{a}$ &\multicolumn{2}{c}{$1 \text{ (assumed)}$}& \multicolumn{1}{l}{$0.626\ (0.005)$} \\
\quad Share of $ac$ $^{b}$  & \multicolumn{2}{c}{ $0.186\ (0.008)$} &  \multicolumn{1}{l}{$0.185\ (0.008)$} \\[4pt]
\multicolumn{4}{l}{\textit{Annual total expenditure ($N=16,868$)}} \\
\quad Bounds        & $[-0.079,\ 0.186]$    & $[-0.019,\ 0.062]$   & $[-0.058,\ 0.244]$    \\
\quad Standard Errors            & ---                   & $(0.114,\ 0.105)$    & $(0.104,\ 0.102)$     \\
\quad 95\% CI       & $[-0.302,\ 0.386]$    & $[-0.206,\ 0.234]$   & $[-0.228,\ 0.410]$    \\[3pt]
\quad Share of positive sample selection $^{a}$ &\multicolumn{2}{c}{$1 \text{ (assumed)}$}& \multicolumn{1}{l}{$0.638\ (0.005)$} \\
\quad Share of $ac$ $^{b}$  & \multicolumn{2}{c} {$0.213\ (0.008)$} & \multicolumn{1}{l}{$0.211\ (0.008)$} \\[4pt]
\bottomrule
\end{tabular}
\par\medskip
\begin{minipage}{0.9\linewidth}
\footnotesize
\textit{Notes:} Outcomes are measured at 12 months after randomization. 
 All calculations use design weights.
\textbf{CF} refers to the \cite{chen2015bounds} bounds under strong sample selection monotonicity without covariates using half-median-unbiased estimates with 95\% confidence intervals following \cite{chernozhukov2013intersection}.
CF reports no standard errors as inference relies on the CLR projection.
\textbf{Sharp-basic} refers to the sharp bounds without covariates under strong sample selection monotonicity.
\textbf{Sharp-DML} refers to the sharp bounds estimated by DML under weak sample selection monotonicity.
Standard errors for the two bounds are in parentheses as
$(\widehat{SE}_{\mathrm{lower}},\,\widehat{SE}_{\mathrm{upper}})$,
and the 95\% CI for $\theta_{SLATE} = E[Y_{1}-Y_{0}\mid ac]$ is calculated using \cite{stoye2020simple}.
\\
$^{a}$ Estimated share of positive sample selection~$= E[\mathbb{I}^{+}(X)]$, the fraction of observations
in the positive sample selection class (Sharp-DML only. CF and Sharp-basic assume this fraction
equals~1).\\
$^{b}$ Estimated share of always-selected compliers $\pi_{ac}$.
\end{minipage}
\end{table}

Table \ref{tab:ohie_bounds} reports results for outpatient visits, prescription drugs, and annual total expenditure. Across outcomes, the share of always-selected compliers is relatively small but non-negligible, ranging from about 15\% to 21\%. Sharp-DML is again the preferred specification because the data do not support strong sample selection monotonicity uniformly across outcomes. The estimated share of positive sample selection under Sharp-DML is 0.957 for outpatient visits, but only 0.626 for prescription drugs and 0.638 for expenditure, indicating that allowing the direction of sample selection to vary with covariates is empirically important, especially for the latter two outcomes.

First, we focus on the comparison between CF and Sharp-basic under strong sample selection monotonicity. The sharp bounds are substantially shorter than CF across all outcomes, with reductions in identified-set length of about 31.1\% for outpatient visits, 50.4\% for prescription drugs, and 69.4\% for total expenditure. The corresponding confidence intervals are also shorter by roughly 25.8\%, 36.9\%, and 36.0\%, respectively. Moreover, for prescription drugs and total expenditure, the Sharp-basic point estimates exclude large portions of the CF range, although their confidence intervals still include zero.

Second, estimates using Sharp-DML are broadly comparable to CF despite relying only on weak sample selection monotonicity. In terms of identified-set length, Sharp-DML reduces the width relative to CF by about 18.5\% for outpatient visits and 57.4\% for prescription drugs, while for total expenditure the width is similar (slightly larger by about 14\%). For confidence intervals, Sharp-DML delivers shorter intervals across all outcomes, with reductions of about 12.7\% for outpatient visits, 42.3\% for prescription drugs, and 7.3\% for total expenditure. Thus, despite the less restrictive assumptions, Sharp-DML maintains or improves precision relative to CF, highlighting the relevance of both sharpness and efficient use of covariates.

Substantively, however, the preferred Sharp-DML confidence intervals include zero for all three outcomes. Hence, the data do not provide statistically precise evidence of intensive margin effects of Medicaid on outpatient visits, prescription-drug use, or annual medical expenditure for always-selected compliers. Taken together with \cite{finkelstein2012oregon}, this suggests that the positive utilization effects found in the original OHIE analysis appear more consistent with a large role for extensive-margin responses, i.e.,~increased probability of any use, over increases in utilization conditional on positive use.




















\subsection{Always-selected Complier Profiling}
\begin{table}[htbp]
\centering
\caption{Oregon Health Insurance Experiment Baseline Covariates: Always-Selected Compliers vs.\ Full Sample (Outcome: Outpatient Visits)}
\label{tab:ohie_profiling_doc}
\footnotesize
\begin{tabular}{@{}l lll@{}}
\toprule
Covariate & $ac$ & Full sample & \multicolumn{1}{c}{Difference} \\
\midrule
Age (in yrs.)                  & $45.47$~$(0.551)$ & $43.76$~$(0.093)$ & \diffcell{1}{711}{~$(0.545)^{***}$}  \\
Female                         & $0.630$~$(0.022)$ & $0.600$~$(0.004)$ & \diffcell{0}{030}{~$(0.022)$}        \\
White, non-Hispanic            & $0.866$~$(0.016)$ & $0.827$~$(0.003)$ & \diffcell{0}{039}{~$(0.016)^{**}$}   \\
Black, non-Hispanic            & $0.025$~$(0.009)$ & $0.033$~$(0.001)$ & \diffcell{-0}{008}{~$(0.009)$}       \\
Hispanic                       & $0.063$~$(0.014)$ & $0.114$~$(0.002)$ & \diffcell{-0}{051}{~$(0.014)^{***}$} \\
American Indian/Alaska Native  & $0.052$~$(0.012)$ & $0.060$~$(0.002)$ & \diffcell{-0}{008}{~$(0.011)$}       \\
Asian                          & $0.026$~$(0.007)$ & $0.033$~$(0.001)$ & \diffcell{-0}{008}{~$(0.006)$}       \\
Other race/ethnicity           & $0.067$~$(0.013)$ & $0.092$~$(0.002)$ & \diffcell{-0}{026}{~$(0.013)^{*}$}   \\

Education: & & & \\
    \quad\quad\quad\quad No high school diploma          & $0.148$~$(0.017)$ & $0.162$~$(0.003)$ & \diffcell{-0}{013}{~$(0.017)$}      \\
    \quad\quad\quad\quad High school diploma or GED      & $0.481$~$(0.023)$ & $0.495$~$(0.004)$ & \diffcell{-0}{014}{~$(0.023)$}      \\
    \quad\quad\quad\quad Some college or vocational      & $0.259$~$(0.020)$ & $0.225$~$(0.003)$ & \diffcell{0}{033}{~$(0.019)^{*}$}   \\
    \quad\quad\quad\quad Four-year college degree or higher & $0.113$~$(0.014)$ & $0.119$~$(0.003)$ & \diffcell{-0}{006}{~$(0.014)$}   \\

Wave of lottery draws: & & & \\
    \quad\quad\quad\quad 1 & $0.148$~$(0.015)$ & $0.115$~$(0.002)$ & \diffcell{0}{034}{~$(0.014)^{**}$}  \\
    \quad\quad\quad\quad 2 & $0.145$~$(0.015)$ & $0.115$~$(0.002)$ & \diffcell{0}{030}{~$(0.015)^{**}$}  \\
    \quad\quad\quad\quad 3 & $0.076$~$(0.015)$ & $0.114$~$(0.002)$ & \diffcell{-0}{039}{~$(0.015)^{***}$}\\
    \quad\quad\quad\quad 4 & $0.142$~$(0.015)$ & $0.140$~$(0.003)$ & \diffcell{0}{003}{~$(0.015)$}       \\
    \quad\quad\quad\quad 5 & $0.134$~$(0.016)$ & $0.142$~$(0.003)$ & \diffcell{-0}{008}{~$(0.016)$}      \\
    \quad\quad\quad\quad 6 & $0.202$~$(0.018)$ & $0.204$~$(0.003)$ & \diffcell{-0}{002}{~$(0.018)$}      \\
    \quad\quad\quad\quad 7 & $0.153$~$(0.017)$ & $0.171$~$(0.003)$ & \diffcell{-0}{018}{~$(0.017)$}      \\

Household members listed: & & & \\
    \quad\quad\quad\quad 1  & $0.799$~$(0.020)$ & $0.693$~$(0.004)$ & \diffcell{0}{106}{~$(0.020)^{***}$} \\
    \quad\quad\quad\quad 2  & $0.200$~$(0.020)$ & $0.305$~$(0.004)$ & \diffcell{-0}{105}{~$(0.020)^{***}$}\\
    \quad\quad\quad\quad 3+ & $0.002$~$(0.002)$ & $0.003$~$(0.000)$ & \diffcell{-0}{001}{~$(0.002)$}      \\

\bottomrule
\end{tabular}
\par\medskip
\begin{minipage}{\linewidth}
\footnotesize
\textit{Notes:} $N=22,064$. The table reports estimated means for the always-selected complier ($ac$) subpopulation and for the full sample, along with their difference ($ac$ minus full sample). Standard errors are in parentheses. All estimates use 12-month survey design weights and are based on the weak-monotonicity DML specification with GRF learners.
Joint $\chi^{2}$ tests for the null that all differences within a category are jointly zero:
Race/ethnicity ($p = 0.008$);
Education ($p = 0.534$);
Wave of lottery draws ($p = 0.035$);
Household members listed ($p < 0.001$).
$^{*}$~$p<0.1$; $^{**}$~$p<0.05$; $^{***}$~$p<0.01$.
\end{minipage}
\end{table}

\begin{table}[htbp]
\centering
\caption{Oregon Health Insurance Experiment Baseline Covariates: Always-Selected Compliers vs.\ Full Sample (Outcome: Prescription Drugs)}
\label{tab:ohie_profiling_rx}
\footnotesize
\begin{tabular}{@{}l lll@{}}
\toprule
Covariate & $ac$ & Full sample & \multicolumn{1}{c}{Difference} \\
\midrule
Age (in yrs.)                  & $45.47$~$(0.551)$ & $43.76$~$(0.093)$ & \diffcell{1}{711}{~$(0.545)^{***}$}  \\
Female                         & $0.630$~$(0.022)$ & $0.600$~$(0.004)$ & \diffcell{0}{030}{~$(0.022)$}        \\
White, non-Hispanic            & $0.866$~$(0.016)$ & $0.827$~$(0.003)$ & \diffcell{0}{039}{~$(0.016)^{**}$}   \\
Black, non-Hispanic            & $0.025$~$(0.009)$ & $0.033$~$(0.001)$ & \diffcell{-0}{008}{~$(0.009)$}       \\
Hispanic                       & $0.063$~$(0.014)$ & $0.114$~$(0.002)$ & \diffcell{-0}{051}{~$(0.014)^{***}$} \\
American Indian/Alaska Native  & $0.052$~$(0.012)$ & $0.060$~$(0.002)$ & \diffcell{-0}{008}{~$(0.011)$}       \\
Asian                          & $0.026$~$(0.007)$ & $0.033$~$(0.001)$ & \diffcell{-0}{008}{~$(0.006)$}       \\
Other race/ethnicity           & $0.067$~$(0.013)$ & $0.092$~$(0.002)$ & \diffcell{-0}{026}{~$(0.013)^{*}$}   \\

Education: & & & \\
    \quad\quad\quad\quad No high school diploma          & $0.148$~$(0.017)$ & $0.162$~$(0.003)$ & \diffcell{-0}{013}{~$(0.017)$}      \\
    \quad\quad\quad\quad High school diploma or GED      & $0.481$~$(0.023)$ & $0.495$~$(0.004)$ & \diffcell{-0}{014}{~$(0.023)$}      \\
    \quad\quad\quad\quad Some college or vocational      & $0.259$~$(0.020)$ & $0.225$~$(0.003)$ & \diffcell{0}{033}{~$(0.019)^{*}$}   \\
    \quad\quad\quad\quad Four-year college degree or higher & $0.113$~$(0.014)$ & $0.119$~$(0.003)$ & \diffcell{-0}{006}{~$(0.014)$}   \\

Wave of lottery draws: & & & \\
    \quad\quad\quad\quad 1 & $0.148$~$(0.015)$ & $0.115$~$(0.002)$ & \diffcell{0}{034}{~$(0.014)^{**}$}  \\
    \quad\quad\quad\quad 2 & $0.145$~$(0.015)$ & $0.115$~$(0.002)$ & \diffcell{0}{030}{~$(0.015)^{**}$}  \\
    \quad\quad\quad\quad 3 & $0.076$~$(0.015)$ & $0.114$~$(0.002)$ & \diffcell{-0}{039}{~$(0.015)^{***}$}\\
    \quad\quad\quad\quad 4 & $0.142$~$(0.015)$ & $0.140$~$(0.003)$ & \diffcell{0}{003}{~$(0.015)$}       \\
    \quad\quad\quad\quad 5 & $0.134$~$(0.016)$ & $0.142$~$(0.003)$ & \diffcell{-0}{008}{~$(0.016)$}      \\
    \quad\quad\quad\quad 6 & $0.202$~$(0.018)$ & $0.204$~$(0.003)$ & \diffcell{-0}{002}{~$(0.018)$}      \\
    \quad\quad\quad\quad 7 & $0.153$~$(0.017)$ & $0.171$~$(0.003)$ & \diffcell{-0}{018}{~$(0.017)$}      \\

Household members listed: & & & \\
    \quad\quad\quad\quad 1  & $0.799$~$(0.020)$ & $0.693$~$(0.004)$ & \diffcell{0}{106}{~$(0.020)^{***}$} \\
    \quad\quad\quad\quad 2  & $0.200$~$(0.020)$ & $0.305$~$(0.004)$ & \diffcell{-0}{105}{~$(0.020)^{***}$}\\
    \quad\quad\quad\quad 3+ & $0.002$~$(0.002)$ & $0.003$~$(0.000)$ & \diffcell{-0}{001}{~$(0.002)$}      \\

\bottomrule
\end{tabular}
\par\medskip
\begin{minipage}{\linewidth}
\footnotesize
\textit{Notes:} $N=17,223$. The table reports estimated means for the always-selected complier ($ac$) subpopulation and for the full sample, along with their difference ($ac$ minus full sample). Standard errors are in parentheses. All estimates use 12-month survey design weights and are based on the weak-monotonicity DML specification with GRF learners.
Joint $\chi^{2}$ tests for the null that all differences within a category are jointly zero:
Race/ethnicity ($p = 0.008$);
Education ($p = 0.534$);
Wave of lottery draws ($p = 0.035$);
Household members listed ($p = 0.000$).
$^{*}$~$p<0.1$; $^{**}$~$p<0.05$; $^{***}$~$p<0.01$.
\end{minipage}
\end{table}

\begin{table}[htbp]
\centering
\caption{Oregon Health Insurance Experiment Baseline Covariates: Always-Selected Compliers vs.\ Full Sample (Outcome: Log Annual Total Expenditure)} 
\label{tab:ohie_profiling_log_cost}
\footnotesize
\begin{tabular}{@{}l lll@{}}
\toprule
Covariate & $ac$ & Full sample & \multicolumn{1}{c}{Difference} \\
\midrule
Age (in yrs.)                  & $44.06$~$(0.498)$ & $43.72$~$(0.094)$ & \diffcell{0}{337}{~$(0.491)$}       \\
Female                         & $0.607$~$(0.020)$ & $0.600$~$(0.004)$ & \diffcell{0}{007}{~$(0.020)$}       \\
White, non-Hispanic            & $0.849$~$(0.016)$ & $0.828$~$(0.003)$ & \diffcell{0}{021}{~$(0.015)$}       \\
Black, non-Hispanic            & $0.028$~$(0.008)$ & $0.033$~$(0.001)$ & \diffcell{-0}{005}{~$(0.008)$}      \\
Hispanic                       & $0.082$~$(0.013)$ & $0.113$~$(0.002)$ & \diffcell{-0}{032}{~$(0.013)^{**}$} \\
American Indian/Alaska Native  & $0.062$~$(0.011)$ & $0.060$~$(0.002)$ & \diffcell{0}{002}{~$(0.010)$}       \\
Asian                          & $0.034$~$(0.006)$ & $0.033$~$(0.001)$ & \diffcell{0}{000}{~$(0.006)$}       \\
Other race/ethnicity           & $0.082$~$(0.012)$ & $0.092$~$(0.002)$ & \diffcell{-0}{010}{~$(0.012)$}      \\

Education: & & & \\
    \quad\quad\quad\quad No high school diploma          & $0.155$~$(0.015)$ & $0.161$~$(0.003)$ & \diffcell{-0}{006}{~$(0.015)$}      \\
    \quad\quad\quad\quad High school diploma or GED      & $0.472$~$(0.020)$ & $0.495$~$(0.004)$ & \diffcell{-0}{022}{~$(0.020)$}      \\
    \quad\quad\quad\quad Some college or vocational      & $0.262$~$(0.017)$ & $0.226$~$(0.003)$ & \diffcell{0}{036}{~$(0.017)^{**}$}  \\
    \quad\quad\quad\quad Four-year college degree or higher & $0.110$~$(0.013)$ & $0.119$~$(0.003)$ & \diffcell{-0}{008}{~$(0.012)$}   \\

Wave of lottery draws: & & & \\
    \quad\quad\quad\quad 1 & $0.138$~$(0.013)$ & $0.114$~$(0.002)$ & \diffcell{0}{024}{~$(0.012)^{*}$}   \\
    \quad\quad\quad\quad 2 & $0.148$~$(0.013)$ & $0.114$~$(0.002)$ & \diffcell{0}{034}{~$(0.013)^{***}$} \\
    \quad\quad\quad\quad 3 & $0.076$~$(0.014)$ & $0.114$~$(0.002)$ & \diffcell{-0}{038}{~$(0.014)^{***}$}\\
    \quad\quad\quad\quad 4 & $0.138$~$(0.014)$ & $0.141$~$(0.003)$ & \diffcell{-0}{002}{~$(0.014)$}      \\
    \quad\quad\quad\quad 5 & $0.127$~$(0.014)$ & $0.142$~$(0.003)$ & \diffcell{-0}{014}{~$(0.014)$}      \\
    \quad\quad\quad\quad 6 & $0.203$~$(0.016)$ & $0.205$~$(0.003)$ & \diffcell{-0}{002}{~$(0.016)$}      \\
    \quad\quad\quad\quad 7 & $0.169$~$(0.016)$ & $0.170$~$(0.003)$ & \diffcell{-0}{002}{~$(0.016)$}      \\

Household members listed: & & & \\
    \quad\quad\quad\quad 1  & $0.794$~$(0.018)$ & $0.692$~$(0.004)$ & \diffcell{0}{102}{~$(0.018)^{***}$} \\
    \quad\quad\quad\quad 2  & $0.206$~$(0.018)$ & $0.305$~$(0.004)$ & \diffcell{-0}{099}{~$(0.018)^{***}$}\\
    \quad\quad\quad\quad 3+ & $0.000$~$(0.002)$ & $0.003$~$(0.000)$ & \diffcell{-0}{002}{~$(0.002)$}      \\
\bottomrule
\end{tabular}
\par\medskip
\begin{minipage}{\linewidth}
\footnotesize
\textit{Notes:} $N=16,868$. The table reports estimated means for the always-selected complier ($ac$) subpopulation and for the full sample, along with their difference ($ac$ minus full sample). Standard errors are in parentheses. All estimates use 12-month survey design weights and are based on the weak-monotonicity DML specification with GRF learners.
Joint $\chi^{2}$ tests for the null that all differences within a category are jointly zero:
Race/ethnicity ($p = 0.241$);
Education ($p = 0.299$);
Wave of lottery draws ($p = 0.019$);
Household members listed ($p < 0.001$).
$^{*}$~$p<0.1$; $^{**}$~$p<0.05$; $^{***}$~$p<0.01$.
\end{minipage}
\end{table}

We now conduct the $ac$ profiling analysis as discussed in Section \ref{sec:profiling}.
Tables \ref{tab:ohie_profiling_doc}--\ref{tab:ohie_profiling_log_cost} compare baseline characteristics of always-selected compliers $ac$ and the full sample across the three utilization outcomes. A consistent pattern emerges. On average, for outpatient visits and prescription drugs, $ac$ are somewhat older (about 1.7 years), more likely to be White, and substantially less likely to be Hispanic, while gender differences are small. Educational attainment is broadly similar, with a slight shift toward some college among $ac$. Differences in lottery timing indicate somewhat greater representation in earlier waves. The most pronounced and robust difference across all outcomes concerns household composition: $ac$ are about 10 percentage points more likely to live alone and correspondingly less likely to be in two-person households. For log total expenditure, these differences are attenuated--age and most demographics are very similar across groups--while household structure and lottery-wave composition remain the main margins of divergence.

Taken together, these patterns indicate mild but systematic selection into the always-selected complier group, primarily along age (for utilization outcomes) and living arrangements. These characteristics are plausibly related to baseline healthcare use conditional on positive utilization. For example, older individuals and single-person households may have more stable or regular interaction with the healthcare system. At the same time, the absence of strong differences in education and other covariates within an already low-income and previously uninsured population suggests that selection is somewhat limited in scope. Given the imprecision of the estimated intensive margin effects, these observable differences do not point to a clear direction for extrapolating $\theta_{SLATE}$ to the full sample. If anything, they suggest modest heterogeneity in baseline utilization intensity, but not a strong or systematic gradient that would allow one to infer whether intensive margin responses should be larger or smaller outside the $ac$ subpopulation.

\section{Supplementary Material for JC and OHIE}\label{app:supp-empirical}
\subsection{Implementation Details} 
We summarize the implementation of the DML procedure used in both applications. All nuisance functions are estimated using generalized random forests \citep{athey2019generalized}, as implemented in the \texttt{grf} package in
\texttt{R}. Propensity scores for the instrument and the joint probabilities of
treatment and selection status are estimated via probability forests, while
conditional mean outcomes are estimated via regression forests. The
conditional distribution function of the outcome, which is required for the
trimming quantiles entering the bounds, is estimated by applying probability
forests pointwise on a fine grid of trimming values with step size 0.005, and
the resulting estimates are post-processed by isotonic regression to enforce
monotonicity. For all forests, the minimum leaf size is set to $\lfloor
n^{0.65}\rfloor $ and each forest comprises 1{,}000 trees. Under weak
monotonicity, where covariate adjustment is performed, nuisance estimates
are obtained via $K=5$-fold cross-fitting so that the predicted nuisances
for each observation are generated by a model trained on the complementary
folds, thereby avoiding overfitting bias in the moment conditions. Under
strong sample selection monotonicity, no covariates are included, so the nuisance parameters
reduce to simple unconditional averages over the whole sample and cross-fitting is not required.

\subsection{Balancing Tables for Original Assignment}
\subsubsection{Job Corps}
Because the proposed bounds in Section~\ref{sec:bounds_covariates} allow the
direction of sample selection to vary with predetermined covariates, we
exploit the 27 baseline characteristics available in the \cite{chen2015bounds} data. 
The same set of covariates is used in \cite{lee2009training}, and also explored
partially in \cite{semenova2025generalized}. These covariates include gender, age at
application, months employed in the previous year, usual weekly hours and
weekly earnings at the most recent job, race and ethnicity indicators,
indicators for children and marital status, parents' education, number of
children, highest grade completed, employment at random assignment and in
the year before random assignment, lagged earnings, household and personal
income categories, and an indicator for whether the person had ever been
arrested. Following \cite{schochet2008does}, we use design weights throughout
because some subpopulations were randomized into the program group with
differing, but known, probabilities for programmatic reasons. Table \ref{tab:jc_balance} reports weighted summary statistics for the groups $Z=0$ and $Z=1$
as well as their mean differences. The covariates are broadly balanced, as
expected under randomization: only age and father's education completed are
marginally significant at the 10\% level, and usual weekly hours at the most
recent job is significant at the 5\% level.

\begin{table}[!htbp]
\centering
\caption{Baseline Covariates: Job Corps}
\label{tab:jc_balance}
\footnotesize
\begin{tabular}{@{}l lll@{}}
\toprule
Covariate & $Z=0$ & $Z=1$ & Diff. (Std.\ Err.) \\
\midrule
Female       & 0.458  & 0.454  & \diffcell{-0}{004}{~$(0.011)$}          \\
Age (in yrs.) at baseline     & 18.35 & 18.44 & \diffcell{0}{087}{~$(0.046)^{*}$}       \\

Black, non-Hispanic        & 0.491  & 0.494  & \diffcell{0}{003}{~$(0.011)$}           \\
Hispanic       & 0.172  & 0.169  & \diffcell{-0}{003}{~$(0.008)$}          \\
Other race/ethnicity & 0.074  & 0.072  & \diffcell{-0}{002}{~$(0.006)$}          \\

Married   & 0.023  & 0.020  & \diffcell{-0}{003}{~$(0.003)$}          \\
Living together   & 0.040  & 0.039  & \diffcell{-0}{002}{~$(0.004)$}          \\
Separated  & 0.021  & 0.024  & \diffcell{0}{003}{~$(0.003)$}          \\
Has Children   & 0.193  & 0.189  & \diffcell{-0}{004}{~$(0.008)$}          \\
Number of children    & 0.268  & 0.270  & \diffcell{0}{002}{~$(0.014)$}           \\
Education      & 10.11 & 10.12 & \diffcell{0}{013}{~$(0.034)$}           \\

Mother's education & 11.47 & 11.49 & \diffcell{0}{024}{~$(0.051)$}           \\
Father's education & 11.50 & 11.41 & \diffcell{-0}{089}{~$(0.048)^{*}$}      \\
Ever arrested  & 0.249  & 0.248  & \diffcell{-0}{001}{~$(0.009)$}          \\

Household income: & & & \\
        \quad\quad\quad\quad$ [\$ 3,000, \$ 6,000)$   & 0.208  & 0.206  & \diffcell{-0}{002}{~$(0.007)$}          \\
        \quad\quad\quad\quad$ [\$ 6,000, \$ 9,000)$  & 0.114  & 0.116  & \diffcell{0}{002}{~$(0.006)$}           \\
        \quad\quad\quad\quad$ [\$ 9,000, \$ 18,000)$   & 0.245  & 0.245  & \diffcell{0}{001}{~$(0.008)$}           \\
        \quad\quad\quad\quad$\geq \$ 18,000$   & 0.181  & 0.179  & \diffcell{-0}{001}{~$(0.007)$}          \\
Personal income: & & & \\
        \quad\quad\quad\quad$ [\$ 3,000, \$ 6,000)$   & 0.131  & 0.128  & \diffcell{-0}{003}{~$(0.007)$}          \\
        \quad\quad\quad\quad$ [\$ 6,000, \$ 9,000)$   & 0.046  & 0.053  & \diffcell{0}{006}{~$(0.004)$}           \\
        \quad\quad\quad\quad$\geq \$ 9,000$   & 0.034  & 0.032  & \diffcell{-0}{002}{~$(0.004)$}          \\

At baseline: & & & \\
          \quad\quad\quad\quad Has job   & 0.192  & 0.198  & \diffcell{0}{007}{~$(0.009)$}           \\
          \quad\quad\quad\quad Months worked, previous year     & 3.530  & 3.603  & \diffcell{0}{074}{~$(0.093)$}           \\
          \quad\quad\quad\quad Had a job, previous year  & 0.627  & 0.635  & \diffcell{0}{008}{~$(0.010)$}           \\
          \quad\quad\quad\quad Earnings, previous year  & 2815 & 2909 & \diffcell{94}{07}{~$(111.7)$}            \\
          \quad\quad\quad\quad Weekly hours, most recent job     & 20.91 & 21.83 & \diffcell{0}{922}{~$(0.454)^{**}$}      \\
          \quad\quad\quad\quad Weekly earnings, most recent job      & 102.9  & 111.1  & \diffcell{8}{183}{~$(5.241)$}           \\
\midrule
$N$          & 3,599  & 5,491  &                             \\
\bottomrule
\end{tabular}
\par\medskip
\begin{minipage}{0.82\linewidth}
\footnotesize
\textit{Notes:} The table reports sample means and mean differences. All statistics use design weights. The p-value for the joint test that all covariate coefficients in a logit regression of $Z$ on the full set of covariates are zero is 0.549.
$^{*}$~$p<0.1$; $^{**}$~$p<0.05$; $^{***}$~$p<0.01$.
\end{minipage}
\end{table}

\subsubsection{Oregon Health Insurance Experiment}
We use the publicly available survey data collected and analyzed by
\cite{finkelstein2012oregon}. The survey was administered by mail in seven
waves during July and August 2009 and includes 29,589 individuals from
lottery-selected households and 28,816 individuals from non-selected
households. The data also contain pre-randomization demographic information
collected at the time of lottery sign-up. Covariates $X$ include gender, age in
years, race and ethnicity indicators, education categories, survey-wave
fixed effects and household-size fixed effects. Table \ref{tab:ohie_balance} reports
weighted summary statistics for the $Z=0$ and $Z=1$ groups, along with their
mean differences, using survey weights. The analysis sample size varies by
outcome and ranges from 16,868 to 22,064 because of outcome-specific item
nonresponse and restrictions to observations with nonmissing treatment,
instrument, outcome, and covariates.

\begin{table}[!htbp]
\centering
\caption{Baseline Covariates: Oregon Health Insurance Experiment}
\label{tab:ohie_balance}
\footnotesize
\begin{tabular}{@{}l lll@{}}
\toprule
Covariate & $Z=0$ & $Z=1$ & Diff.\ (Std.\ Err.) \\
\midrule
Age (in yrs.)                  & 42.57 & 42.69 & \diffcell{0}{087}{~$(0.178)$}           \\
Female                         & 0.596 & 0.583 & \diffcell{-0}{012}{~$(0.007)$}          \\
White, non-Hispanic            & 0.827 & 0.815 & \diffcell{-0}{012}{~$(0.006)^{**}$}     \\
Black, non-Hispanic            & 0.038 & 0.032 & \diffcell{-0}{005}{~$(0.003)^{*}$}      \\
Hispanic                       & 0.123 & 0.129 & \diffcell{0}{007}{~$(0.005)$}           \\
American Indian/Alaska Native  & 0.067 & 0.064 & \diffcell{-0}{003}{~$(0.004)$}          \\
Asian                          & 0.026 & 0.030 & \diffcell{0}{003}{~$(0.002)$}           \\
Other race/ethnicity           & 0.103 & 0.108 & \diffcell{0}{005}{~$(0.005)$}           \\

Education: & & & \\
    \quad\quad\quad\quad No high school diploma          & 0.175 & 0.171 & \diffcell{-0}{005}{~$(0.006)$}          \\
    \quad\quad\quad\quad High school diploma or GED      & 0.492 & 0.499 & \diffcell{0}{007}{~$(0.007)$}           \\
    \quad\quad\quad\quad Some college or vocational      & 0.220 & 0.223 & \diffcell{0}{003}{~$(0.006)$}           \\
    \quad\quad\quad\quad Four-year college degree or higher & 0.113 & 0.107 & \diffcell{-0}{005}{~$(0.005)$}       \\

Wave of lottery draws: & & & \\
    \quad\quad\quad\quad 1 & 0.082 & 0.146 & \diffcell{0}{064}{~$(0.005)^{***}$}     \\
    \quad\quad\quad\quad 2 & 0.081 & 0.152 & \diffcell{0}{070}{~$(0.005)^{***}$}     \\
    \quad\quad\quad\quad 3 & 0.074 & 0.153 & \diffcell{0}{078}{~$(0.005)^{***}$}     \\
    \quad\quad\quad\quad 4 & 0.148 & 0.133 & \diffcell{-0}{015}{~$(0.005)^{***}$}    \\
    \quad\quad\quad\quad 5 & 0.151 & 0.132 & \diffcell{-0}{018}{~$(0.005)^{***}$}    \\
    \quad\quad\quad\quad 6 & 0.241 & 0.166 & \diffcell{-0}{075}{~$(0.006)^{***}$}    \\
    \quad\quad\quad\quad 7 & 0.224 & 0.117 & \diffcell{-0}{105}{~$(0.005)^{***}$}    \\

Household members listed: & & & \\
    \quad\quad\quad\quad 1  & 0.744 & 0.659 & \diffcell{-0}{084}{~$(0.007)^{***}$}    \\
    \quad\quad\quad\quad 2  & 0.255 & 0.336 & \diffcell{0}{080}{~$(0.007)^{***}$}     \\
    \quad\quad\quad\quad 3+ & 0.001 & 0.005 & \diffcell{0}{004}{~$(0.001)^{***}$}     \\

\midrule
$N$                            & 11,203 & 11,107 &                            \\
\bottomrule
\end{tabular}
\par\medskip
\begin{minipage}{0.84\linewidth}
\footnotesize
\textit{Notes:} The table reports sample means and mean differences. All statistics use 12-month survey design weights. Wave of lottery draws refers to the wave number of the lottery draw in which the household's name(s) appeared on the waitlist. Randomization was within each wave. Household members listed is the number of household members entered on the lottery waitlist. The $p$-value for the joint test that the coefficients on Age, Female, Race/Ethnicity, and Education are jointly zero in a logit regression of $Z$ on these covariates, lottery-wave fixed effects, and household-size fixed effects is 0.575.
$^{*}$~$p<0.1$; $^{**}$~$p<0.05$; $^{***}$~$p<0.01$.
\end{minipage}
\end{table}

\end{document}